\pdfoutput=1
\documentclass[acmsmall]{mystyle}

\acmJournal{PACMPL}
\acmVolume{1}
\acmNumber{CONF} 
\acmArticle{1}
\acmYear{2018}
\acmMonth{1}
\acmDOI{} 
\startPage{1}

\renewcommand\footnotetextcopyrightpermission[1]{}
\setcopyright{none}

\usepackage{booktabs}   
\usepackage{subcaption} 
\usepackage{listings}
\usepackage{multirow}
\usepackage{tikz}
\usepackage{tikzcodeblocks}
\usetikzlibrary{arrows,automata,calc,patterns,snakes}
\usepackage{subfiles}
\usepackage{hhline}
\usepackage[clock]{ifsym}
\usepackage{amsmath}
\usepackage{xspace}
\usepackage{wrapfig}
\usepackage{booktabs}
\usepackage{subcaption}
\usepackage{extarrows}
\usepackage{pdfcomment}
\usepackage{subfiles}
\usepackage{multirow}

\usepackage{tcolorbox}
\usepackage{lipsum}
\usepackage{fancyvrb}
\usepackage{listings}
\usepackage{tikzit}
\definecolor{mygreen}{rgb}{0,0.6,0}
\definecolor{mygray}{rgb}{0.98,0.98,0.98}
\definecolor{mymauve}{rgb}{0.58,0,0.82}
\definecolor{sangria}{rgb}{0.57, 0.0, 0.04}
\definecolor{denim}{rgb}{0.08, 0.38, 0.74}

\newlength{\MaxSizeOfLineNumbers}%
\lstdefinestyle{customc}{
  belowcaptionskip=0\baselineskip,
  breaklines=true,
  xleftmargin=\MaxSizeOfLineNumbers,
  language=C,
  numbers=left,
  showstringspaces=false,
  basicstyle=\footnotesize\ttfamily,
  keywordstyle=\bfseries\color{green!40!black},
  morekeywords={ASSUME}, 
  commentstyle=\itshape\color{purple!40!black},
  morecomment  = [is]{/*}{*/},
  identifierstyle=\color{cyan!40!black},
  stringstyle=\color{orange},
  tabsize=1,
  backgroundcolor=\color{mygray},
}

\lstdefinestyle{examples}{
  belowcaptionskip=1\baselineskip,
  breaklines=true,
  xleftmargin=0em,
  language=C,
  showstringspaces=false,
  basicstyle=\small\ttfamily,
  keywordstyle=\bfseries\color{blue!40!black},
  commentstyle=\itshape\color{purple!40!black},
  morekeywords={read, write, cas, ASSUME,assume},
  keywords=[2]{rlx,ra,P1,P2a,P2b,P2c},
  keywordstyle=[2]\bfseries\color{orange},
  identifierstyle=\bfseries\color{magenta!50!black},
  stringstyle=\color{orange},
  tabsize=1
}

\lstdefinestyle{pseudoc}{
  belowcaptionskip=1\baselineskip,
  breaklines=true,
  xleftmargin=2em,
  mathescape=true,
  language=C,
  numbers=left,
  showstringspaces=false,
  basicstyle=\small\ttfamily,
  keywordstyle=\bfseries\color{blue!40!black},
  commentstyle=\itshape\color{purple!40!black},
  morekeywords={read, write, cas},
  keywords=[2]{view, message,Memory, Promise, Promises, timestamp},
  keywordstyle=[2]\color{orange},
  keywords=[3]{generate,get,check,mark,replace,insert,load, save, update},
  keywordstyle=[3]\color{blue},
  identifierstyle=\bfseries\color{magenta!50!black},
  stringstyle=\color{orange},
  tabsize=2,
  backgroundcolor=\color{mygray},
}

\newcommand{\goldf}[1]{{\color{orange} #1}}
\newcommand{\gmess}{\goldf{\texttt{message}}~}
\newcommand{\gmem}{\goldf{\texttt{Memory}}~}
\newcommand{\gview}{\goldf{\texttt{view}}~}

\usepackage[ruled]{algorithm2e}
\usepackage[justification=centering]{caption}
\usepackage{multicol}

\usepackage[colorinlistoftodos,prependcaption]{todonotes}

\usepackage{mathtools}
\usepackage{xargs}
\usepackage{xstring}

\usepackage{subcaption} 
\usepackage{paralist}
\usepackage[shortlabels]{enumitem}

	\definecolor{applegreen}{rgb}{0.55, 0.71, 0.0}
	\definecolor{crimsonglory}{rgb}{0.75, 0.0, 0.2}
	\definecolor{ferngreen}{rgb}{0.31, 0.47, 0.26}
\newcommand\appendixmain[1]{}

	\definecolor{amethyst}{rgb}{0.6, 0.4, 0.8}

\newcommand\prog{{\it Pro	g}}
\newcommand\varset{{\mathsf{Loc}}}
\newcommand\procset{{\mathcal{P}}}

\newcommand\regset{{\mathsf{Reg}}}
\newcommand\regsetof[1]{\regset\left(#1\right)}
\newcommand\proc{{p}}
\newcommand\instr{\mathfrak{i}}

\newcommand\stmt{\mathfrak{s}}
\newcommand\assigned\leftarrow

\newcommand\terminated{term}

\newcommand{\arw}{{\mathbf{CAS}}}
\newcommand{\cas}{{\mathbf{CAS}}}
\newcommand{\rlx}{{\tt{rlx}}}
\newcommand{\lhc}{{\mathsf{LoHoW}}}
\newcommand{\fadd}{{\mathbf{FADD}}}

\newcommand{\var}{{\mathsf{loc}}}
\newcommand{\ch}{{\mathsf{HW}}}
\newcommand{\lst}{{\mathfrak{st}}}
\newcommand{\chh}{{\mathsf{\bf HW}}}

\newcommand{\wrt}{{\mathsf{wrt}}}
\newcommand{\from}{{\textcolor{cobalt}{\tt{frm}}}}
\newcommand{\msgg}{{\mathsf{msg}}}
\newcommand{\ptr}{{\mathsf{ptr}}}

\newcommand{\tool}{$\mathsf{PS2SC}$}
\newcommand{\too}{{\textcolor{cobalt}{\tt{to}}}}
\newcommand{\acq}{{\tt{acq}}}
\newcommand{\rel}{{\tt{rel}}}
\newcommand{\upd}{{\tt{U}}}

\newcommand{\Ss}{{\mathcal{S}}}

\newcommand{\adj}{{\mathsf{Adj}}}

\newcommand{\delete}{\mathsf{del}}

\newcommand{\ps}{\textsf{PS 2.0}\xspace}

\newcommand{\bps}{\textsf{bdPS 2.0-rlx}}

\newcommand{\psr}{\textsf{PS 2.0-rlx}}

\newcommand{\act}[2][]{\mathrel{\raisebox{-1pt}[10pt][0pt]{%
			\ensuremath{\underset{^{\raisebox{-6pt}[0pt][0pt]{\ensuremath{^{^{#1}}}}}}%
				{\raisebox{0pt}[3pt][0pt]{\ensuremath{\relbar\mspace{-8mu}\xrightarrow{#2}}}}}}}}

\newcommand\xvar{x}

\newcommand\lbl\lambda
\newcommand\reg{\$r}
\newcommand{\init}{\mathsf{init}}

\newcommand{\Ii}{\mathcal{I}}
\newcommand{\MS}{\mathcal{MS}}

\newcommand{\Tt}{\mathcal{T}}

\newcommand{\fence}{\mathsf{SC\text{-}fence}}

\newcommand{\mini}{\mathsf{min}}
\newcommand{\tar}{\tt{target}}

\newcommand{\prm}{{\tt{prm}}}

\newcommand{\cert}{{\tt{cert}}}
\newcommand{\nor}{{\tt{std}}}

\newcommand{\stdd}{{\tt{std}}}

\newcommand{\view}{{\mathsf{View}}}

	\definecolor{burgundy}{rgb}{0.5, 0.0, 0.13}
\definecolor{cadmiumgreen}{rgb}{0.0, 0.42, 0.24}
\definecolor{auburn}{rgb}{0.43, 0.21, 0.1}

\newcommand{\val}{\tt{{val}}}
	\definecolor{cobalt}{rgb}{0.0, 0.28, 0.67}
\newcommand{\assume}{\keyworr{assume}}
\newcommand\keywor[1]{\textcolor{violet!40!blue}{{\tt #1}}}
\newcommand\keyword[1]{\textcolor{auburn}{{\tt #1}}}
\newcommand\keyworr[1]{\textcolor{cadmiumgreen}{{\tt #1}}}

\newcommand\kw[1]{\texttt{\textcolor{burgundy}{#1}}}
\newcommand\kwif{\kw{if}}
\newcommand\kwthen{\kw{then}}
\newcommand\kwelse{\kw{else}}
\newcommand\kwendif{\kw{end\,if}}
\newcommand\kwwhile{\keywor{while}}
\newcommand\kwdo{\keywor{do}}
\newcommand\kwdone{\keywor{done}}
\definecolor{byzantium}{rgb}{0.44, 0.16, 0.39}

	\definecolor{blush}{rgb}{0.87, 0.36, 0.51}
\newcommand \lcomme{\textcolor{blush}{/*}}
\newcommand \rcomme{\textcolor{blush}{*/}}
\newcommand \pouto {\textcolor{green!50!blue}{~//}}
\newcommand \poutt {\textcolor{blush}{~//}}

\newcommand \rd{{\tt{rd}}}

\newcommand{\true}{\tt{true}}
\newcommand{\false}{\tt{false}}
\newcommand{\ra}{{\tt{ra}}}

\newcommand \wt{{\tt{wt}}}

\newcommand\conf{\mathfrak{c}}
\newcommand\initconf{\mathfrak{c}_{\it init}}
\newcommand\confset{\mathcal{C}}

\newcommand\msg[4]{\ensuremath{({#1},{#2},({#3},{#4}])}}
\newcommand\msgnew[5]{\ensuremath{({#1},{#2},({#3},{#4}],{#5})}}
\newcommand\reserv[3]{\ensuremath{({#1},({#2},{#3}])}}
\newcommand\rcmc{\textsc{Rcmc}}
\newcommand\cdsc{\textsc{CDSChecker}}
\newcommand\genmc{\textsc{GenMC}}

\makeatletter
\newcommand\sr[2][]{\ext@arrow 0099{\longrightarrowfill@}{#1}{#2}}
\def\longrightarrowfill@{\arrowfill@\relbar\rightarrow}
\makeatother

\newcommand{\rtstep}[1]{\mathbin{[{#1}]^{*}}}

\newcommand{\upclos}[1]{{{#1}{\uparrow}}}

\newcommand{\sem}[1]{\llbracket{#1}\rrbracket}

\SetKwProg{Fn}{Algorithm}{}{}

\begin{document}

\title[The Decidability of Verification under Promising 2.0]{The Decidability of Verification under Promising 2.0}         

\author[Parosh Abdulla]{Parosh Aziz Abdulla}
\affiliation{
  \institution{Uppsala University}
  \city{Uppsala}
  \country{Sweden}
}

\author[M. Faouzi Atig]{Mohamed Faouzi Atig}
\affiliation{
  \institution{Uppsala University}
  \city{Uppsala}
  \country{Sweden}
}


\author{Adwait Godbole}
\affiliation{
  \institution{IIT Bombay}
  \city{Mumbai}
  \country{India}
}

\author[S. Krishna]{Shankaranarayanan Krishna}
\affiliation{
  \institution{IIT Bombay}
  \city{Mumbai}
  \country{India}
}

\author[V. Vafeiadis]{Viktor Vafeiadis}
\affiliation{
  \institution{MPI-SWS}
  \streetaddress{Saarland Informatics Campus (SIC)}
  \city{Kaiserslautern and Saarbr\"ucken}
  \country{Germany}
}

\begin{abstract}
In PLDI'20, Kang et al.\ introduced the \emph{promising } semantics ($\ps$) of the C++ concurrency that  captures most of the common program transformations while  satisfying the  DRF guarantee. The reachability problem for finite-state programs under $\ps$ with only release-acquire accesses ($\ps$-$\ra$) is already known to be undecidable. Therefore, we address, in this paper, the reachability problem  for programs running under $\ps$ with relaxed accesses ($\psr$) together with promises.  We show that this problem is undecidable even in the case where the input program
has finite state. Given this undecidability result, 
we consider the fragment of $\psr$ with a bounded number of promises.
We show that under this restriction, the  reachability is decidable, albeit very expensive: it is non-primitive recursive.
Given this high complexity for $\psr$ with bounded number of promises and the undecidability result for $\ps$-$\ra$, we consider a bounded version of the reachability problem. To this end, we bound both the number of promises and the ``view-switches'', i.e, the number of times the processes may switch their local views of the global memory.
We provide a code-to-code translation from an input program under $\ps$,  with 
 relaxed and release-acquire memory accesses along with promises, to a program under SC. This leads to a reduction of the bounded reachability problem under $\ps$ to the bounded context-switching problem under SC. We have implemented a prototype tool and tested it on a set of benchmarks, demonstrating that many bugs in programs can be found using a small bound. 
\end{abstract}

\sloppy


\maketitle

\section{Introduction}
\label{sec:intro}

An important long-standing open problem in PL research has been to define a   weak memory model that 
 captures the semantics of concurrent  memory accesses in languages like Java and C/C++.
A model is considered good if it can be implemented efficiently
(i.e., if it supports all usual compiler optimizations
and its accesses are compiled to plain x86/ARM/Power/RISCV accesses),
and is easy to reason about. 
After many attempts at solving this problem (e.g., \cite{jmm,zhang-feng:2013,bubbly,crary-sullivan:2015,rc11,jeffrey-riely:2019,batty:2011}),
a breakthrough was achieved by Kang et al.~\cite{promising}, who introduced the \emph{promising semantics}.
This was the first model that supported basic invariant reasoning, the DRF guarantee,
and even a non-trivial program logic~\cite{Svendsen:2018}.

In  the promising semantics, the memory is modeled as a set of timestamped messages, each corresponding to a write made by the program. 
Each process/thread records its own view of the memory---i.e., the latest timestamp for each memory location that it is aware of. A message has the form $\msgnew{x}{v}{f}{t}{V}$ where $x$ is a location, $v$ a value to be stored for $x$, $(f,t]$ is the timestamp 
interval corresponding to the write and $V$ is the local view 
of the process who made the write to $x$.
When reading from memory, a process can either return the value stored at the timestamp in its view
or advance its view to some larger timestamp and read from that message.
When a process $p$ writes to memory location $x$,  a new message with a timestamp larger than $p$'s view of $x$ is created, 
and $p$'s view is advanced to include the new message.
In addition, in order to allow load-store reorderings,  a process is allowed to \emph{promise} a certain write in the future.  A promise is also added as a message in the memory, except that the local view of the process is not updated using the timestamp interval in the message. 
This is done only when the promise is eventually fulfilled. 
A {\em consistency} check is used to ensure that every promised message can be {\em certified} (i.e., made fulfillable) by executing that process 
on its own. Furthermore, this should hold from any future  memory  (i.e., from any extension of the memory with additional  messages). The quantification  prevents deadlocks (i.e., processes from making promises they are not able  to fulfil).
 The promising semantics  generally allows program executions to contain unboundedly many concurrent promised messages,
provided that all of them can be certified. As one can immediately see, this is a fairly complex model,
and beyond its support for some basic reasoning patterns,
it is not at all obvious whether it is easy to reason about concurrent programs running under this model. 
Furthermore, the  unbounded number of  future memories, that need to be checked,   makes the verification of even simple programs practically infeasible. Moreover, a number of transformations based on global value range analysis as well as register promotion were  not supported in ~\cite{promising}.   

To address the above concerns, a new version of the promising semantics $\ps$ \cite{promising2} has been 
proposed, by redesigning key components of the promising semantics \cite{promising}. Mainly, $\ps$ supports register promotion and global value range analysis, while capturing all features (thread local optimizations, DRF guarantees, hardware mappings)  of the promising semantics of \cite{promising}.   $\ps$ simplifies also the consistency check and instead of checking the  promise fulfilment from all future memories, $\ps$
checks for promise fulfilment only from a specially crafted extension 
of the current memory called capped memory. $\ps$ also introduces 
the notion of reservations, which allows a process to secure an timestamp interval  in order to perform a future atomic read-modify-write instruction. The reservation 
blocks any other message from using that timestamp interval.    
Reservations allows register promotions.

 The wide umbrella of features of $\ps$ allowing two memory access modes, 
 relaxed ($\rlx$) and release-acquire ($\ra$) along with promises, reservations  and subsequent certification make $\ps$ a very complex 
 model. While the $\ps$ semantics is a breakthrough 
 contribution, a natural and fundamental question is to investigate  the  verification of concurrent programs under $\ps$. For that, investigating the decidability
 of  verification problems as well as defining efficient analysis techniques  are two extremely important problems.
  
One of the problems addressed in this paper  
is to ivestigate the decidability of the reachability problem for $\ps$. 
Let $\ps$-$\rlx$ and $\ps$-$\ra$ 
represent respectively, the fragment of $\ps$ allowing only relaxed 
($\rlx$) and release-acquire ($\ra$) memory accesses.  
The reachability with only $\ra$ accesses has been shown to be  undecidable \cite{pldi2019}, even without 
the features of promises and reservations. That leaves only the fragment $\psr$ of $\ps$ 
for investigation. We show that if unbounded number of promises is allowed, the reachability problem is 
 undecidable in $\psr$, while it becomes decidable if we  bound
the number of promises at any time (however, the total number of promises made with a run can be unbounded).  Our undecidability is obtained with just 2 threads, 
with an execution where the number of context switches between the two processes is three, where a context is a computation segment 
in which one process is active. 
  The proof of decidability is done 
by proposing a new memory model with higher order words $\lhc$, and showing the equivalence 
of $\psr$ and $\lhc$. Under the bounded promises assumption,,  
we use the decidability of the coverability problem of well structured transition systems (WSTS) \cite{wsts2,wsts1} to 
show that the reachability problem for $\lhc$ with bounded number of promises is decidable.
 
Given this high complexity for $\psr$ with bounded number of promises and the undecidability result for $\ps$-$\ra$  \cite{pldi2019}, we consider a bounded version of the reachability problem. To this end, we propose  a 
parametric under-approximation in the spirit of context bounding \cite{cb2,DBLP:conf/cav/TorreMP09,DBLP:journals/fmsd/LalR09,demsky,MQ07,DBLP:conf/tacas/QadeerR05,pldi2019,cb3}. 
The bounding concept chosen for concurrent programs depends on aspects related to the 
interactions between the processes. In the case of SC programs, context bounding 
has been shown experimentally to have extensive behaviour coverage for bug detection \cite{MQ07,DBLP:conf/tacas/QadeerR05}. 
A context in the SC setting is a computation segment where only one process is active. 
The concept of context bounding has been extended for  weak memory models. For instance, in TSO, 
 the notion of  context is extended to one where all updates to the main memory are done only from the buffer  
 of  the active thread \cite{cb2}. 
  In the case 
  of RA \cite{pldi2019}, context bounding was extended to view bounding, 
  using the notion of view-switching messages.  
  Since $\ps$ subsumes RA, we propose a 
  bounding notion that extends the view bounding 
  proposed in \cite{pldi2019}. Using this new bounding notion, we propose a source to source translation 
  from programs under $\ps$ to context-bounded executions 
  of the transformed program in SC. 
  The main challenge in the code-to-code translation of \cite{pldi2019} was to keep track of the causality between different variables. In our case, the challenge is  fundamentally different and  is to provide a  procedure that 
  (i) handles different memory accesses $\rlx$ and $\ra$, 
  (ii) guesses the promises and reservations in a non-deterministically manner, and (iii) verify that each promise so guessed is fulfilled using the capped memory.  
This reduction is implemented in a tool, called \tool.
Our experimental results   demonstrate the effectiveness of our approach.
We exhibit cases where hard-to-find bugs are detectable using a small view-bound $K$.
Our tool displays resilience to trivial changes in the position of bugs and the order of processes.

\smallskip

\noindent
{\textbf{Related Work.}}
  The decidability of the verification problems for programs running under weak memory models has been addressed for TSO \cite{ABBM10}, $\ps$-$\ra$ \cite{pldi2019}, Power \cite{Power_Netys20}, and for a subclass of  $\ps$-$\ra$ \cite{DBLP:conf/pldi/LahavB20}. To the best of our knoweldge, this the first time that this problem is investigated for $\psr$ and  \tool{} is the first tool for automated verification of programs under $\ps$, which also works 
  for the promising semantics    \cite{promising}. 
 Most of the existing work concerns the development of 
 stateless model checking (SMC), coupled with (dynamic) partial order
reduction techniques (e.g., \cite{phong,rcmc,genmc,demsky,cdsc}) and do not handle promises.  Context-bounding has been proposed in \cite{DBLP:conf/tacas/QadeerR05} for programs running under SC.  This work has been extended in different directions and has led to efficient and scalable techniques for the analysis of concurrent programs (see e.g., \cite{MQ07,DBLP:journals/fmsd/LalR09,DBLP:conf/cav/TorreMP09,madhu3,DBLP:conf/popl/EmmiQR11,DBLP:conf/cav/TorreMP10}). In the context of weak memory models, context-bounded analysis has been only proposed to programs running under  TSO/PSO in \cite{cb2,DBLP:conf/sefm/TomascoN0TP17} and under  POWER  in \cite{cb3}.  

\section{Preliminaries}
\label{sec:prels}

In this section, we introduce the simple programming language and the notation
that will be used throughout. Then, we  review $\ps$  definition, and present 
the model following  \cite{promising2}.

\subsection{Notations}
Given two natural numbers $i, j \in \mathbb{N}$ s.t. $i \leq j$, we use $[i,j]$ to denote the set $\{k \,|\, i \leq k \leq j\}$.
Let $A$ and $B$ be two sets. We use $f: A \rightarrow B$ to denote that $f$ is a function from $A$ to $B$. We define $f[a \mapsto b]$ to be the function $f'$ such that $f'(a)=b$ and $f'(a')=f(a')$ for all $a' \neq a$. 
For a binary relation $R$, we use $\rtstep R$ to denote its reflexive and transitive closure.
Given an alphabet $\Sigma$, we use $\Sigma^*$ (resp.\ $\Sigma^+$) to denote the set of possibly empty (resp.\ non-empty) finite words over $\Sigma$.
Let $w= a_1 a_2 \cdots a_n$ be a word over $\Sigma$, we use $|w|$ to denote the length of $w$.
Given an index $i$ in  $[1,|w|]$, we use $w[i]$ to denote the $i^{\text{th}}$ letter of $w$.
Given two indices $i$ and $j$ s.t. $1\leq i \leq j \leq |w|$, we use $w[i,j]$ to denote the word $a_i a_{i+1} \cdots a_j$. Sometimes, we consider a word as a function from $[1,|w|]$ to $\Sigma$.

\subsection{Program Syntax}

The simple programming language we use 
is described in Figure \ref{program_syntax}. 
A program $\prog$  consists of a set $\varset$ of (global) variables or memory locations,
and the  definition of a set $\procset$ of processes.  
Each  process $\proc$ declares a set $\regsetof\proc$ of (local) {\it registers} followed by  a sequence of labeled instructions. We assume that these sets of registers are disjoint and we use  $\regset:=\cup_\proc\regsetof\proc$ to denote their union. We assume also a (potentially unbounded) data domain $\mathsf{Val}$ from which the registers and locations take values. 
All locations and registers are assumed to be initialized with the   
special value $0 \in \mathsf{Val}$ (if not mentioned otherwise).
An instruction $\instr$ is of the form   $\lbl : \stmt$ where 
$\lbl$ is a unique label and $\stmt$ is a statement. We use  
$\mathbb{L}_{\proc}$ to denote the set of  all  labels 
of the process $\proc$, and  $\mathbb{L}=\bigcup_{\proc \in \procset}\mathbb{L}_{\proc}$ the set of all  labels of all processes. 
We assume that the  execution of the process $\proc$ starts always with a unique initial instruction  labeled by 
$\lambda_{\rm init}^{\proc}$.
A write instruction is of the form $\xvar^o=\reg$ assigns the value of register $\reg$ to the location $\xvar$, and 
$o$ denotes the access mode. If $o=\mathsf{rlx}$, the write is a \emph{relaxed} write, while 
if $o=\mathsf{ra}$, it is a  \emph{release} write. 
A read instruction $\reg=\xvar^o$  reads the value of the location $\xvar$ into the local register $\reg$. Again, 
if the access mode $o=\mathsf{rlx}$, it is a \emph{relaxed} read, and 
if $o=\mathsf{ra}$, it is an \emph{acquire} read. Atomic updates or $\mathsf{RMW}$  
instructions are either  
 compare-and-swap ($\arw^{o_r,o_w}$) or $\fadd^{o_r, o_w}$. Both have a pair of accesses ($o_r,o_w \in \{ \mathsf{rel}, \mathsf{acq}, 
 \mathsf{rlx}\}$) to the same location -- a read   followed by a write.  Following \cite{promising2},  $\fadd(x,v)$  stores the value of $x$ into a register $\reg$, and adds $v$ 
to $x$, while  $\cas(x, v_1, v_2)$  compares an expected value $v_1$ to the value in  $x$, and 
 if the values are same, sets the value of  $x$ to $v_2$. The old value of $x$
 is then stored in $\reg$.

\tikzset{background rectangle/.style={draw=black,rounded corners,fill = black!2}}
\setlength\intextsep{0pt}
\begin{wrapfigure}[12]{r}{5.5cm}
  \begin{tikzpicture}[codeblock/.style={line width=0.3pt, inner xsep=0pt, inner ysep=0pt}, show background rectangle]
\node[codeblock] (init) at (current bounding box.north east) {
\footnotesize
{
$
\begin{array}{l}
~~~~\prog ::=  \keyword{var}~ x^*(\keyword{proc} ~p|| \dots || \keyword{proc} ~p) \\
~~~~ \keyword{proc}~ p::=\regset(p)~ \instr^*\\
~~~\instr::=\lambda:\stmt\\
~\stmt \in \mathsf{St}::=   \\
~~\;\;\mathsf{skip} ~~\;\;|s;s ~~\;\;|\assume(\xvar=e)\\
~~\;\;|\kwdo ~s^* ~\kwwhile~ e ~~\;\;|\kwwhile ~e ~\kwdo~ s^* \kwdone \\
~~\;\;|\kwif~ e ~	\kwthen~ s~ \kwelse ~s ~\\
~~\;\;|\reg~:=~ e ~~\;\;|\reg~:=~ x^o ~~\;\;|x^o~:=~ \reg \\
~~\;\;|\reg~:=~\mathbf{FADD}^{o,o}(x,v) \\~~\;\;|\reg~:=~\mathbf{CAS}^{o,o}(x,v,v)~~\;\;|\fence\\
 o \in \mathsf{Mode} ::=  \mathsf{rlx}|\mathsf{ra}
\end{array}
$
}
};
\end{tikzpicture}
\caption{\footnotesize Syntax of concurrent programs.}
\label{program_syntax}
\end{wrapfigure}

 A {\em local} assignment instruction $\reg=e$ assigns to the register $\reg$ the value of $e$, where $e$ is an expression over a set of operators,
constants as well as the contents of the registers of the current process, but not referring to the set of locations. 
The fence instruction $\keyword\fence$ is used to enforce sequential consistency  if it is placed between two memory access operations. 
Finally, the conditional, assume and iterative instructions
 have the standard semantics.
For simplicity, we will write $\assume(\xvar=e)$ instead of 
${\reg=\xvar}; \assume(\reg=e)$.
This notation is extended in the straightforward manner to conditional statements.

\subsection{The Promising Semantics}
\label{sec:ps}
In this section, we recall the promising semantics \cite{promising2}.  We present here $\ps$ with three 
memory accesses, \emph{relaxed} (this is the default mode), 
\emph{release writes} ($\rel$) and \emph{acquire reads} 
($\acq$). Read-modify-writes (RMW) instructions  have two access modes	 - one for read 
and one for write. We keep aside 
the release and acquire fences (and subsequent 
access modes) which are part of $\ps$, since 
they do not affect the results of this paper.

\smallskip\noindent{\bf Timestamps.}
$\ps$ uses timestamps to maintain a total order over all the writes to the same variable. 
We assume an infinite set of timestamps $\mathsf{Time}$, densely totally ordered by $\leq$,
with $0$ being the minimum element. A \emph{view} is a timestamp function $V : \varset \rightarrow \mathsf{Time}$ records the largest known timestamp for each location.  
Let  $\mathbb{T}$ be  the set containing all the timestamp functions, along with the special 
symbol $\bot$. 
  Let $V_{\rm init}$ represent the initial view where  all locations are mapped to $0$.  Given two views $V$ and $V'$, we use $V \leq V'$ to denote that $V(x) \leq V'(x)$ for $x \in \varset$. The merge operation $\sqcup$ between the two views $V$ and $V'$ returns the pointwise maximum of   $V$ and $V'$,
i.e., $(V \sqcup V')(y)$ is the maximum of $V(y)$ and $V'(y)$.  Let $\Ii$ denote the set of all intervals over $\mathsf{Time}$.
  The timestamp intervals in $\Ii$ have the form $(f,t]$ where either $f=t=0$ or $f < t$, with $f, t \in \mathsf{Time}$. Given an interval $I=(f,t] \in \Ii$, $I.\from$ and $I.\too$ denote $f, t$ respectively. 

\smallskip\noindent{\bf Memory.}
In $\ps$, the memory is modelled as a set of concrete \emph{messages} (which we just call messages), and \emph{reservations}. 
Each message represents the effect of a write or a RMW operation and each reservation is a timestamp interval reserved for future use. 
In more detail, a message $m$ is a tuple $\msgnew{\xvar}{v}{f}{t}{V}$ where $\xvar \in \varset$, $v \in \mathsf{Val}$, $(f, t] \in \Ii$ and $V\in \mathbb{T}$. A reservation $r$ is a tuple $\reserv{x}{f}{t}$. Note that a reservation, unlike a message, does not commit to any particular value, but only specifies the interval which is reserved.
We use $m.\var$ ($r.\var$), $m.\val$, $m.\too$ ($r.\too$), $m.\from$ ($r.\from$) and $m.\view$ to denote respectively $x$, $v$, $t$, $f$ and $V$.
Two elements (either messages or reservations) are said to be \emph{disjoint} ($m_1 \# m_2$) if they concern different variables ($m_1.\var\neq m_2.\var$) or their intervals do not overlap ($m_1.\too < m_2.\from \lor m_1.\from > m_2.\too$).
Two sets of elements $M, M'$ are disjoint, denoted $M \# M'$, if $m \# m'$ for every $m \in M, m' \in M'$. 
Two elements $m_1, m_2$ are \emph{adjacent} denoted $\adj(m_1,m_2)$  if $m_1.\var=m_2.\var$ and $m_1.\too=m_2.\from$. A memory $M$ is a set of pairwise disjoint messages and reservations.  Let $\widetilde{M}$ be the subset of $M$ containing 
only messages (no reservations). For a location $x$, let $M(x)$ be  
$\{m \in M \mid m.\var=x\}$. Given a view $V$ and a memory $M$,  we say $V \in M$ if  $V(x)=m.\too$ for some 
message $m \in \widetilde{M}$  for every $x \in \varset$. Let ${\mathbb M}$ denote the set of all  memories.

\smallskip

\noindent
{\it Insertion into Memory.} Following \cite{promising2}, 
a memory $M$ can be extended with a \emph{message} (due to the execution of a write/RMW instruction) or a \emph{reservation} $m$ with $m.\var = x$, $m.\from = f$ and $m.\too = t$ in a number of ways:

\smallskip

\noindent
{[Additive insertion]}
$M \stackrel{A}{\hookleftarrow} m$ is defined only if (1) $M\# \{m\}$; (2) if $m$ is a message, then no message $m' \in M$ has $m'.\var = x$ and $m'.\from = t$; and (3) if $m$ is a reservation, then there exists  a message $m'\in \widetilde{M}$ with $m'.\var = x$ and $m'.\too = f$. The extended memory $M \stackrel{A}{\hookleftarrow} m$ is then $M \cup \{m\}$.

\noindent
{[Splitting insertion]}
$M \stackrel{S}{\hookleftarrow} m$ is defined if $m$ is a  message, and, if there exists a message $m'=(x, v', (f,t'],V)$  with $t < t'$ in $M$.  Then $M$ is updated to $M \stackrel{S}{\hookleftarrow} m = (M\backslash\{m'\} \cup \{m, (x,v', (t,t'],V)\})$.

\noindent
[Lowering Insertion] 
$M \stackrel{L}{\hookleftarrow} m$ is only defined if  there exists $m'$ in $M$ that is identical to $m=\msgnew x v f t V$ except for $m.\view \leq m'.\view$. Then, $M$ is updated to $M \stackrel{L}{\hookleftarrow} m = M\backslash \{m'\} \cup \{m\}$.

\smallskip

\noindent
[Cancellation] 
$M \stackrel{C}{\hookleftarrow} m$ is defined if $m$ is a reservation in $M$. 
Then  $M$ is updated as  $M\setminus \{m\}$.

\smallskip\noindent{\bf  Transition System of a Process.}
Given a process $\proc \in \procset$, a  state $\sigma$  of $\proc$ is defined as a pair $(\lambda,R)$ where  $\lambda  \in \mathbb{L}$  is the label of the next instruction to be executed by $\proc$ and $R : \regset \rightarrow \mathsf{Val}$ maps each register of $\proc$ to its current value. (Observe that we use the set of all labels $\mathbb{L}$ (resp. registers $\regset$) instead of  $\mathbb{L}_{\proc}$ (resp. $\regsetof\proc$) in the definition of $\sigma$   just for the sake of simplicity.) Transitions between the states of  $\proc$ are of the form $ (\lambda,R)\xRightarrow[p]{t} (\lambda',R')$ with $t \in \{\epsilon, \rd(o,x,v), \wt(o,x,v), \upd(o_r, o_w, x, v_r, v_w), \fence\,|\, x \in \varset, v \in \mathsf{Val}, o\in  \{\rlx,\ra\} \}$. A transition of the form $ (\lambda,R)\xRightarrow[p]{\rd(o,x,v)} (\lambda',R')$ denotes the execution of a read instruction of the form $\reg=\xvar^o$  labeled by $\lambda$ where  $(1)$  $\lambda'$ is the label 
of the next instructions that can be executed after the execution of the instruction labelled by $\lambda$, and   $(2)$ $R'$ is the mapping that results from the replacement  of the value of the register $\reg$ in $R$ by $v$. The transition relation $ (\lambda,R)\xRightarrow[p]{t} (\lambda',R')$ is defined in similar manner for the other cases of $t$ where $\wrt(o,x,v)$ stands for a write instruction that  writes the value $v$ to $x$, $\upd(o_r, o_w, x, v_r, v_w)$  stands for a RMW that reads the value $v_r$ from $x$ and write $v_w$ to it, $\fence$ stands for a $\fence$ instruction, and $\epsilon$ stands for the execution of the other local instructions. Observe that   $o, o_r, o_w$ are  the access modes which can be $\rlx$ 
or $\ra$. We use $\ra$  for both release and acquire. 
Finally, we use $ (\lambda,R)\xrightarrow[p]{t} (\lambda',R')$ with $t \in \{\rd(o,x,v), \wt(o,x,v), \upd(o_r, o_w, x, v_r, v_w), \fence\,|\, x \in \varset, v \in \mathsf{Val}, o\in  \{\rlx,\ra\} \}$ to denote that  $ (\lambda,R)\xRightarrow[p]{\epsilon} \sigma_1   \xRightarrow[p]{\epsilon}  \cdots \xRightarrow[p]{\epsilon} \sigma_n \xRightarrow[p]{t} \sigma_{n+1} \xRightarrow[p]{\epsilon} \cdots \xRightarrow[p]{\epsilon}  (\lambda',R')$.

\smallskip\noindent{\bf Machine States.} 
A machine state  $\MS$ is a tuple $(({\sf J}, {\sf R}), \sf {VS}, {\sf PS}, M, G)$, where
${\sf J} : \procset \mapsto \mathbb{L}$ maps each process $p$ to the label of the next instruction to be executed, 
${\sf R} : \regset \rightarrow \mathsf{Val}$ maps each register to its current value, 
${\sf VS} = \procset \rightarrow \mathbb{T}$ is the process view map, which maps each process to a view,	
$M$ is a memory and $PS: \procset \mapsto {\mathbb M}$ maps each process to a set of messages (called \emph{promise} set),  and 
$G \in \mathbb{T}$ is the global view (that will be used by SC fences).
We use  $\confset$ to denote the set of all machine states.

Given a machine state $\MS=(({\sf J}, {\sf R}), \sf {VS}, {\sf PS}, M, G)$ and a process $p$, let $\MS{\downarrow}p$ denote the projection, $(\sigma, \sf {VS}(p), {\sf PS}(p), M, G)$ with $\sigma=({\sf J}(p), {\sf R}(p))$, of the machine state to the process $p$. 

We call $\MS {\downarrow} p$ the process configuration. 
We use $\confset_p$ to denote the set of all process configurations.

The initial machine state $\MS_{\rm init}=(({\sf J}_{\rm init}, {\sf R}_{\rm init}), {\sf VS}_{\rm init}, {\sf PS}_{\rm init}, M_{\rm init},G_{\rm init})$ is one where:
(1) ${\sf J}_{\rm init}(p)$ is the label of the  initial   instruction of $\proc$; (2) ${\sf R}_{\rm init}(\reg)=0$ for every  $\reg \in \regset$;
(3) for each $p$, we have ${\sf VS}(p) = V_{\rm init}$ as the initial view (that maps each location to the timestamp 0);  

(4) for each process $p$, the set of promises ${\sf PS}_{\rm init}(p)$ is empty;
(5) the initial memory $M_{\rm init}$ contains exactly one initial message $(x,0, (0, 0], V_{\rm init})$ for each location $x$;
and (6) the initial global view maps each location to $0$.

\tikzset{background rectangle/.style={fill=none}}
\begin{figure}[t]
\centering
\small
\resizebox{\textwidth}{!}{
\begin{tikzpicture}[codeblock/.style={line width=0.5pt, inner xsep=0pt, inner ysep=5pt}  , show background rectangle]
\node[codeblock] (init) at (current bounding box.north west) {
$
\def\arraystretch{1.2}
\begin{array}{c}
\rowcolor{black!5}
\begin{array}{cc}
	\begin{array}{c}
		\textbf{Memory Helpers} \\
		\frac{\text{(MEMORY : NEW})}{( P,\ M) \ \xrightarrow{m}\left( P',\ M\ \ \stackrel{A}{\hookleftarrow } m\right)} \\
		\text{MEMORY FULFIL} \\
		\frac{\hookleftarrow \ \in \ \left\{\stackrel{S}{\hookleftarrow } ,\ \stackrel{L}{\hookleftarrow }\right\} ,\ P'=P\ \hookleftarrow \ m,\ M'=M\ \hookleftarrow m\ }{( P,M) \ \xrightarrow{m}( P'\ \backslash \{m\} ,\ M')}
	\end{array} &
	\begin{array}{c}
		\textbf{Process Helpers} \\
			\begin{array}{cc}
				\frac{\begin{array}{c}
					m=( x,-,( -,t],K) \in M \quad V(x) \leq t \\
					o=\rlx \Rightarrow V'=V[x\mapsto t] \\
					o=\ra \Rightarrow V'=V[x \mapsto t] \sqcup K
				\end{array}}{V \xrightarrow[\text{rd}]{o,m} V'} &
				\frac{ \begin{array}{c}
					m=(x,-,( -,t],K) \in M, V(x) < t \\
					o=\rlx \Rightarrow \ K=\bot ,\ \ o=\ra \Rightarrow \ P( x) =\emptyset \ \land \ K=V' \\
					( P,M) \ \xrightarrow{m}( P',\ M')\quad  V'=V[ x\ \mapsto \ t]
				\end{array} } {( V,P,M) \ \xrightarrow[\text{wt}]{o,m}( V',\ P',M')}
			\end{array}
	\end{array}
\end{array} \\
\rowcolor{black!2}
\begin{array}{c}
\textbf{Process Steps} \\
\begin{array}{ccc}
	\text{Read} & \text{Write} & \text{Promise}  \\
	\displaystyle\frac{\begin{array}{c}
	\sigma \xrightarrow[p]{rd( o,x,v)} \sigma '\\
	m=( x,v,( -,-] ,\ -) ,\ \ \ V\ \xrightarrow[\text{rd}]{o,m} V'
	\end{array} } {( \sigma ,\ V,\ P,\ M,\ G) \ \xrightarrow[p]{}( \sigma ',\ V',\ P,\ M,\ G)}  & 
	\displaystyle\frac{ \begin{array}{c}
	\sigma \xrightarrow[p]{wt( o,x,v)} \sigma' \\
	m=(x, v, (-,-], -), (V,P,M) \xrightarrow[\text{wt}]{o,m}(V',P',M')
	\end{array}}{( \sigma ,\ V,\ P,\ M,\ G) \ \xrightarrow[p]{}( \sigma ',\ V',\ P',\ M',\ G)} &
	\displaystyle\frac{ \begin{array}{c}
	m=( -,-,( -,-] ,K) ,\\
	M'=M\ \stackrel{A}{\hookleftarrow } m,\ K\ \in \ M'
	\end{array}}{( \sigma ,\ V,\ P,M,G) \ \ \xrightarrow[p]{}\left( \sigma ,\ V,\ P\ \stackrel{A}{\hookleftarrow } m,\ M',\ G\right)} \\
	
	\text{SC-fence} & \text{Reserve} & \text{Cancel} \\
	\displaystyle\frac{\sigma \ \xrightarrow[p]{SC\ fence} \ \sigma '}{( \sigma ,\ V,\ P,\ M,\ G) \ \xrightarrow[p]{}( \sigma ',\ V\ \sqcup \ G,\ P,\ M,\ G\ \sqcup \ V)} &
	
	\displaystyle\frac{r=( -,\ ( -,-]) ,\ M'=M\ \stackrel{A}{\hookleftarrow } r}{( \sigma ,\ V,\ P,M,G) \ \stackrel{}{}\xrightarrow[p]{}( \sigma ,\ V,\ P\ \cup \ \{r\} ,\ M',\ G)} &
	
	\displaystyle\frac{r=( -,( -,-]) \ \in \ P}{( \sigma ,\ V,\ P,\ M,\ G) \ \xrightarrow[p]{}( \sigma ,\ V,\ P\ \backslash \{r\} ,\ M\ \backslash \ \{r\} ,\ G)} 	 
\end{array} \\
\begin{array}{c}
	\text{Update} \\
	\displaystyle\frac{ \begin{array}{c}
	\sigma \ \xrightarrow[p]{U( o_{r} ,\ o_{w} ,\ x,\ v_{r} ,\ v_{w}) \ } \sigma '', 
	m_{r} =( x,\ v_{r} ,\ ( -,t] ,\ -) ,\ m_{w} =( x,v_{w} ,\ ( t,-] ,-), \\
	\ V\ \xrightarrow[\text{rd}]{o_{r} ,m_{r}} V'',\ ( V'',P,M) \ \xrightarrow[\text{wt}]{o_{w} ,m_{w}}( V',\ P',M')
	\end{array}}{( \sigma ,\ V,\ P,\ M,\ G) \ \xrightarrow[p]{}( \sigma ',\ V',\ P',\ M',\ G)} 
\end{array}
\end{array}
\end{array}$
};
\end{tikzpicture}
}
\vspace{-0.5cm}
\caption{\footnotesize $\ps$ inference rules at the process level, defining  the transition $(\sigma, V, P, M, G) \xrightarrow[p]{} (\sigma', V', P', M', G')$.}
\vspace{-0.6cm}
\label{program_sem}
\end{figure}
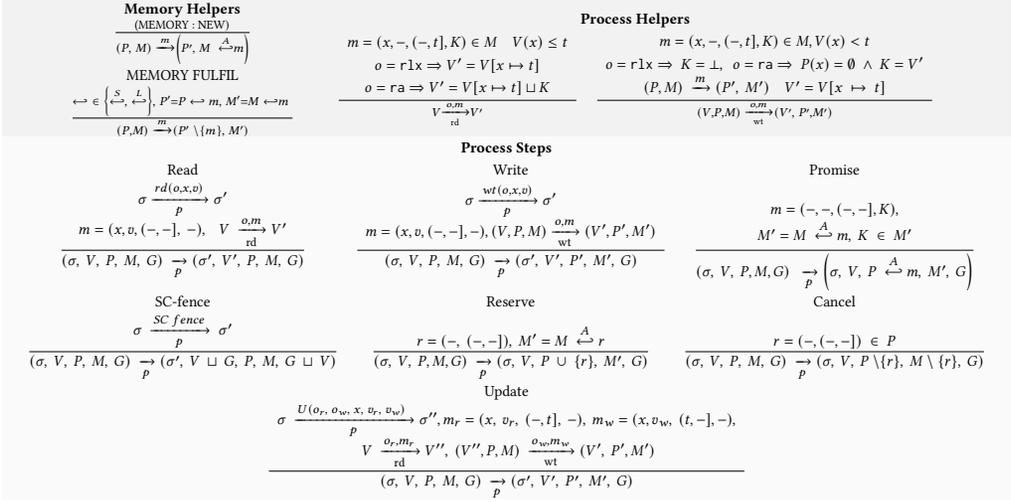

\smallskip\noindent{\bf Transition Relation.}
We first describe the transition  $(\sigma, V, P, M, G) \xrightarrow[p]{} (\sigma', V', P', M', G')$ between process configurations in $\confset_p$ %
from which we induce the transition relation between machine states.

\smallskip

\noindent{\it Process Relation.}
The formal definition of  $\xrightarrow[p]{}$ is in Figure \ref{program_sem}.  Below, we explain these inference rules.

\noindent{\bf {Read}}. A process $p$ can read from  $M$ by observing a 
message $m=(x,v, (f,t], K)$ if $V(x) \leq t$ (i.e.,  $\proc$ must not be aware of a later message for  $\xvar$). 
In case of a relaxed read $\rd(\rlx, x, v)$, 
 the process view of $x$ is updated to $t$, while 
 for an acquire read $\rd(\ra, x, v)$, the process view is updated 
to  $V[x \mapsto t] \sqcup K$. The global memory $M$, the set of promises $P$, and the global view $G$ remain the same.

\noindent{\bf {Write}}. A process can add a fresh message 
to the memory ($\mathsf{MEMORY:NEW}$) or fulfil an outstanding promise 
($\mathsf{MEMORY : FULFILL}$). The execution of a  write ($\wt(\rlx, x, v)$)  results in  a message $m$ with location 
$x$ along with a timestamp interval $(-, t]$. Then,  the process view of location $x$ is  updated to $t$. 
 In case of a release write ($\wt(\ra, x, v)$) the updated process view is also attached 
to $m$, and ensures that the process does not have an outstanding promise 
 on location $x$.  ($\mathsf{MEMORY : FULFILL}$) allows 
 to split a promise interval or lower its view before fulfilment.

\noindent{\bf {Update}}. When a process performs a RMW, it first reads a message 
$m=(x, v, (f,t], K)$ and then writes an update message   
with $\from$
 timestamp equal to $t$; that is, a message 
of the form $m'=(x, v', (t, t'], K')$. This forbids any other write 
to be placed between $m$ and $m'$. The access modes of the reads and writes in the update 
follow what has been described for the read and write above. 

\noindent{\bf Promise, Reservation and Cancellation.} 
 A process can non-deterministically \emph{promise} future writes which 
 are not release writes. 
   This is done by adding a message $m$ to the memory $M$ 
 s.t. $m \#M$ and to the set of promises $P$. Later, a relaxed write instruction can fulfil an existing promise.  Recall that the execution of  a release write requires that the set of promises to be empty and thus it can not be used to fulfil a promise. 
  In the  reserve step,  the process reserves a timestamp interval to be used 
  for a later  RMW instruction reading from a certain message without fixing the value it will write.  
  A reservation is added both to the memory and the promise set. The process can drop the reservation from both sets using the cancel step in non-deterministic manner.   

\noindent{\bf SC fences.}The process view $V$ is merged with the global view $G$, resulting 
in $V \sqcup G$ as the updated process view and global view.

\smallskip

\noindent
{\it Machine Relation.}
We are ready now to define the induced transition relation between machine states.  
  For machine states $\MS=((J, R), VS, PS, M, G)$ and $ \MS'=((J', R'), VS', PS', M', G')$, we write
  $\MS \xrightarrow[p]{} \MS' $ iff   $(1)$ ${\MS {\downarrow} p}  \xrightarrow[p]{} {\MS {\downarrow} p}$ and $(J(p'),VS(p'), PS(p')) = (J'(p'),VS'(p'), PS'(p'))$ for all $p' \neq p$.

\smallskip

\noindent{\bf Consistency.}
According to Lee et al. \cite{promising2}, there is one final requirement on machine states called \emph{consistency},
which roughly states that, from every encountered machine state encountered,
all the messages promised by a process $\proc$ can be {\em certified} (i.e., made fulfillable) by executing $\proc$ on its own from  a certain
future memory (called capped memory), i.e.,  extension of the memory with additional
reservation. Before  defining consistency, we  need  to introduce capped memory.

\smallskip

\noindent
\emph{Cap View, Cap Message and Capped Memory.}  The last element of a memory $M$ with respect to a location $x$, denoted by $\overline{m}_{M, x}$, is 
an element from $M(x)$ with the highest timestamp among all elements of $M(x)$ and is defined as 
$\overline{m}_{M,x} =  \max_{m\in M(x)} m.\too$.  
The \emph{cap view} of a memory $M$, denoted by $\widehat{V}_M$, is the view which assigns to each 
location $x$, the $\too$ timestamp in  the message $\overline{m}_{\widetilde{M},x}$. That is, 
 $\widehat{V}_M = \lambda x. \overline{m}_{\widetilde{M}, x}.\too$.  Recall that  $\widetilde{M}$ denote  the subset of $M$ containing 
only messages (no reservations).
The \emph{cap message} of a memory $M$ with respect to a location $x$, 
is given by the message 
$\widehat{m}_{M,x} = (x, \overline{m}_{\widetilde{M}, x}.{\tt{val}}, (\overline{m}_{M, x}.\too, \overline{m}_{M, x}.\too + 1], \widehat{V}_{M})$. 

 Then, the capped memory of a memory $M$, wrt. a set of promises $P$, denoted by $\widehat{M}_P$, is an extension of $M$, defined as: $(1)$ for every $m_1,m_2 \in M$ with $m_1.\var = m_2.\var,~ m_1.\too <
m_2.\too$, and there is no message $m' \in M(m_1.\var)$ such
that $m_1.\too < m'.\too < m_2.\too$, we include a reservation
$(m_1.\var, (m_1.\too, m_2.\from])$ in $\widehat{M}_P$, and $(2)$ we include a cap message $\widehat{m}_{M,x}$ in $\widehat{M}_P$ for every variable $x$ unless $\overline{m}_{M,x}$ is a reservation in $P$.

\smallskip

\noindent
\emph{Consistency of machine states.}
A machine state $\MS=((J, R), VS, PS, M, G)$ is \emph{consistent}
if  every process $p\in\procset$ can certify/fulfil all its promises  from  the capped memory $\widehat{M}_{PS(p)}$, i.e., 
$ ((J, R), VS, PS, \widehat{M}_{PS(p)}, G) \rtstep{\xrightarrow[p]{}} ((J', R'), VS', \emptyset, M', G')$.

\medskip

\noindent
\textbf{The  Reachability Problem in $\ps$.}
A   run of $\prog$ is a sequence of the form: 
  $\MS_{0} \rtstep{\xrightarrow[p_{i_1}]{}} 
   \MS_{1} \rtstep{\xrightarrow[p_{i_2}]{}} 
   \MS_{2} \rtstep{\xrightarrow[p_{i_3}]{}} 
   \ldots  
   \xrightarrow[p_{i_n}]{*}
   \MS_{n}$
where $\MS_0=\MS_{\rm init}$ is the initial machine state and $\MS_1,\ldots,\MS_n$ are consistent machine states.
In this case, the machine states $\MS_0,\ldots,\MS_n$ are said to be  reachable from $\MS_{\rm init}$.

Given an instruction label function $J: \procset \rightarrow \mathbb{L}$ that maps each process $\proc \in \procset$ to an instruction label in $ \mathbb{L}_{\proc}$, 
the \emph{reachability} problem asks
whether there exists a machine state of the form $((J,R),V,P,M,G)$ that is reachable from $\MS_{\rm init}$.
In the case of a positive answer to this problem, we say that $J$ is  reachable in $\prog$ in $\ps$.

\subsection{Examples}

\tikzset{background rectangle/.style={draw=black,sharp corners,fill = yellow!5}}
\begin{wrapfigure}{r}{0.5\textwidth}
\begin{tikzpicture}[codeblock/.style={line width=0.2pt, inner xsep=0pt, inner ysep=0pt}, show background rectangle]
\node[codeblock] (init) at (current bounding box.north east) {
\footnotesize
\begin{tabular}[h]{l||l}
\begin{lstlisting}[style=examples,tabsize=4]
$r1=x
if($r1 != 2){ 
	z=1
	$r1=z
	assume($r1=3)
	z=2
}
else{
	z=2 //
}
\end{lstlisting}
&
\begin{lstlisting}[style=examples,tabsize=4]
z=3
$r2=z
assume($r2=2)
x=2
\end{lstlisting}
\end{tabular}
};
\end{tikzpicture}
\caption{\footnotesize The annotated behaviour is not reachable.}
\label{eg1}
\end{wrapfigure}

In the following, we describe some examples to demonstrate $\ps$. For readability, instead of referring to 
reachable instruction labels, we consider possible 
program outcomes represented using the program comment annotation ``\textcolor{sangria}{//}''.   
	All writes and reads are relaxed in both examples below.

\begin{example}
The annotated program outcome in Figure \ref{eg1} is not allowed by $\ps$. 
	
\noindent We list the execution steps of $\ps$ showing that the annotated behaviour is not possible. We give a proof by contradiction. 
Assume that the annotated behaviour is possible. The only way for this is that the first process $p_1$ (whose code on the left side) to execute the $\kwelse$ branch. For this, it needs to read 2 from x. 
This can be provided only by the second  process   using the write x=2. For this to happen, 
$p_2$ first executes the write z=3 by adding a message (z, 3, $(r,s], \bot)$ to the memory. Next,   
$p_2$ has to read a message 
of the form (z, 2, $(f,t], \bot)$  which can only be generated 
by $p_1$ as a promise.  
\vspace{.2cm}

Note that $p_1$ 
can promise the write $z=2$ in its $\kwif \dots \kwthen$ branch. To certify this promise, $p_1$ starts from 
the capped memory, and first executes the write z=1 in the $\kwif \dots \kwthen$ branch. To do this, 
it can split the promise interval $(f,t]$ and add a message (z, 1, $(f, t'], \bot)$ while modifying 
 (z, 2, $(f,t], \bot)$ in the memory to (z, 2, $(t',t], \bot)$. Note that since 
  we work from the capped memory, there are no available intervals 
  in $[0, max(t,s)]$, and the only way to add a message for the write z=1 of $p_1$,   
  in such a way that $p_1$ can read the 3 written by $p_2$, and also to fulfil its promise, 
  is to split the promise interval.  Next, $p_1$ reads (z, 3, $(r,s], \bot)$ to 
 go past the $\keyworr{assume}$(z=3) statement. This imposes $f<t' \leq r < s$. However, 
   since $p_2$ 
 wrote 3 to z before reading the promise (z, 2, $(t',t], \bot)$, we also need $r < s \leq f <t'$ which contradicts 
 $f < r$. Hence, the annotated behaviour is not reachable, since $p_1$ fails the certification.  
 \end{example}

\begin{example}
In Figure \ref{eg2}, we present an example having a run realising the program outcome  
which has unboundedly many reservations and subsequent cancellations.  

\tikzset{background rectangle/.style={draw=black,sharp corners,fill = yellow!5}}
\begin{figure}[h]
  \begin{tikzpicture}[codeblock/.style={line width=0.3pt, inner xsep=0pt, inner ysep=0pt}, show background rectangle, sharp corners]
\node[codeblock] (init) at (current bounding box.north east) {
\begin{tabular}{c||c||c}
\begin{lstlisting}[style=examples,tabsize=2]
$r1=z
if($r1 = 2)
{
	x=2 //
} 
else{
	do{
		$r4 = FADD(y,1)
	}while (w=0)
	x=2 
}
\end{lstlisting}
&
\begin{lstlisting}[style=examples,tabsize=2]
w=1
$r2=x
assume($r2 == 2)
z=2
\end{lstlisting}
&
\begin{lstlisting}[style=examples,tabsize=2]
do
y=$r3
while(w=0)
\end{lstlisting}
\end{tabular}
};
\end{tikzpicture}
\caption{\footnotesize The annotated behaviour is reachable.}
\label{eg2}
\end{figure}
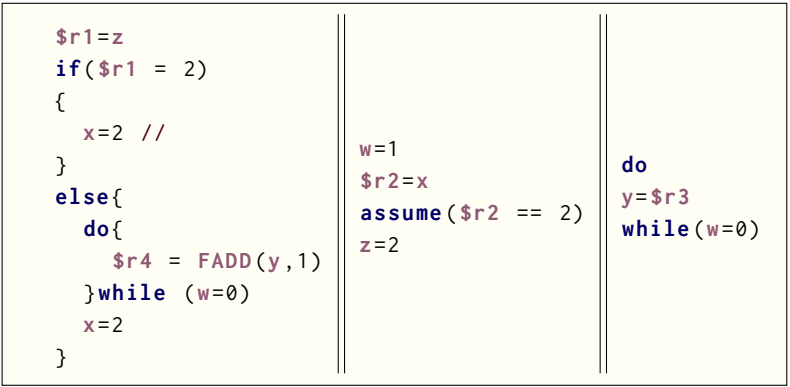

	\noindent We list the execution steps of $\ps$ leading to the annotated behaviour. Items prefixed with ``C'' represent certification steps.
 \begin{itemize}
 \item[(1)] Process 2 writes 1 to $w$.
	 \item[(2)] 
 Process 3  writes arbitrarily many messages $(y, 0, (f_1, t_1], \bot), (y, 0, (f_2, t_2], \bot) \dots (y, 0, (f_k, t_k], \bot)$ such that $t_1 < f_2 < t_2 < f_3 \dots < f_k<t_k$, until it reads the value 1 from $w$.  The number of messages written depends on the number of iterations of $\kwwhile$. 
	\item[(3)] Process 1  promises  $(x, 2, (f, t], \bot)$ corresponding to the 	 write $x=2$ in the $\kwelse$ branch.  
\item[(4)] Process 1 makes arbitrarily many reservations $(y, (t_1, t'_1]),
(y, (t_2, t'_2]), \dots, (y, (t_{k-1}, t'_{k-1}])
$ such that $t'_1< f_2< t'_2 < f_3 \dots t'_{k-1} < f_k<t_k$ and $(y, (t_k, t_{k+1}])$.   
\item[(C1)] Starting from the capped memory, process 1 
cancels the reservations one by one, while executing the 
  $\fadd$ instructions, thereby adding messages   
 $(y, 1, (t_i, t'_i], \bot)$ to the memory. 
 \item[(C2)] Process 1 fulfils its promise. 
  \item[(5)] Process 2 	 reads the message $(x, 2, (f,t], \bot)$ 
  and adds the message $(z, 2, (f'', t''], \bot)$ for the write $z=2$. 
  \item[(6)]Process 1 reads $(z, 2, (f'', t''], \bot)$
  and fulfils  $(x, 2, (f, t], \bot)$ reaching the program outcome. 
  \end{itemize}
\end{example}

\section{Undecidability of Consistent Reachability in $\ps$}
\label{sec:undec}

In this section, we show that  reachability is undecidable 
for $\ps$ even for finite-state programs.  The proof is by a reduction from Post's Correspondence Problem (PCP) \cite{post}. Our proof works with the fragment of $\ps$ having only relaxed ($\rlx$) memory accesses and crucially uses unboundedly many promises to ensure that a process cannot skip any writes made by another process. 
It also works even when we restrict our analysis to executions that can be split into a bounded number of contexts,
where within each context, only one process is active.  We need just 3 context switches. 
Our undecidability result is also \emph{tight} in the sense that
the reachability problem becomes decidable when we restrict ourselves to machine states where the number of promises is bounded. Given our proof (Theorem \ref{thm:undec}) where undecidability is obtained 
with the  $\rlx$ fragment of $\ps$,  a natural question is the 
decidability status of the $\ra$ fragment of $\ps$. 
This is known to be undecidable from \cite{pldi2019} even in the absence of promises.    Let us call  the fragment of $\ps$ with only $\rlx$ memory accesses $\psr$.

\begin{theorem}
\label{undecidability}
 The  reachability  problem for concurrent programs over a finite data domain is undecidable under $\ps$. 
 In fact, the undecidability still holds for the $\psr$ fragment.  
 \label{thm:undec}
\end{theorem}

{
\begin{figure*}[h]
\centering
\colorbox{black!5}{\footnotesize
\begin{tabular}{|@{}c@{}|@{}c@{}|@{}c@{}|@{}c@{}|}
\hline\hline
Process $\proc_1$ & Process $\proc_2$& $\keyword{Module^{\proc_1}_{v_i}}$ & $\keyword{Module^{\proc_2}_{u_i}}$ \\
\hline\hline
$\begin{array}[t]{l}
\lcomme ~\textcolor{blush}{generation~ mode}~ \rcomme 
\\  \kwif\ \mathit{validate} = 0\ \kwthen
\\\quad \kwwhile\ \mathit{term} = 0\ \kwdo
\\\qquad \mathit{index} = 1
\\\qquad \keyword{Module^{\proc_1}_{u_1}}
\\\qquad \mathit{index} = \#
\\\qquad \ldots
\\\qquad \mathit{index} = n
\\\qquad \keyword{Module^{\proc_1}_{u_n}}
\\\qquad \mathit{index} = \#
\\\quad \kwdone
\\\quad \mathit{index} = \S
\\\lcomme ~\textcolor{blush}{validation~ mode}~ \rcomme 
\\\kwelse
\\\quad \reg' = \mathit{index}'
\\\quad \assume(\reg' \in [1,n])
\\\quad \kwwhile\ \reg' \neq \S \ \kwdo
\\\qquad \kwif\ \reg' = 1\ \kwthen
\\\qquad \quad\keyword{Module^{\proc_1}_{v_1}}
\\\qquad \kwelse\ \kwif\ \reg' = 2\ \kwthen
\\\qquad \quad\keyword{Module^{\proc_1}_{v_2}}
\\\qquad \ldots
\\\qquad \kwelse\ \kwif\ \reg' = n\ \kwthen
\\\qquad \quad\keyword{Module^{\proc_1}_{v_n}}
\\\qquad \kwendif
\\\qquad \assume(\mathit{index}' = \#)
\\\qquad \reg' = \mathit{index}'
\\\qquad \assume(\mathit{index}' \neq \#)
\\\quad \kwdone
\\\quad \mathit{index} = \S
\\\quad \textcolor{red}{\assume(\mathit{true})}\pouto
\\\kwendif
\end{array}$
&
$\begin{array}[t]{l}
  \mathit{term} = 1;
\\\reg = \mathit{index};
\\\assume(\reg \in [1,n])
\\\kwwhile\ \reg \neq \S \ \kwdo
\\\quad \kwif\ \reg = 1\ \kwthen
\\\quad \quad\keyword{Module^{\proc_2}_{u_1}}
\\\quad \kwelse\ \kwif\ \reg = 2\ \kwthen
\\\quad \quad\keyword{Module^{\proc_2}_{u_2}}
\\\quad \ldots
\\\quad \kwelse\ \kwif\ \reg = n\ \kwthen
\\\quad \quad\keyword{Module^{\proc_2}_{u_n}}
\\\quad \kwendif
\\\quad \assume(\mathit{index} = \#)
\\\quad \reg = \mathit{index}
\\\quad \assume(\mathit{\reg} \neq \#)
\\\kwdone
\\\mathit{validate} = 1
\\\mathit{index}' =\S
\\\textcolor{red}{\assume(\mathit{true})}\poutt
\end{array}$
&
$\begin{array}[t]{l}
  \assume(y = v_i [1])
\\\assume(y = \#)
\\\assume(y = v_i [2])
\\\ldots
\\\assume(y = v_i [|v_i|])
\\\assume(y = \#)
\\x =v_i[1]
\\x = \#
\\x =v_i[2]
\\\ldots
\\x = v_i [|v_i|]
\\\mathit{index} = i
\\\mathit{index} = \#
\\~
\\\hline\hline
\multicolumn{1}{c}{ \keyword{Module^{\proc_1}_{u_i}} }
\\\hline\hline
x = u_i [1]
\\x = \#
\\x = u_i [2]
\\\ldots
\\x = u_i [|u_i|]
\\x = \#
\end{array}$
&
$\begin{array}[t]{l}
  \assume(x = u_i [1])
\\\assume(x = \#)
\\\assume(x = u_i [2])
\\\ldots
\\\assume(x = u_i [|u_i|])
\\\assume(x = \#)
\\y =u_i[1]
\\y = \#
\\y =u_i[2]
\\\ldots
\\y = u_i [|u_i|]
\\\mathit{index}' = i
\\\mathit{index}' = \#
\end{array}$
\\\hline\hline
\end{tabular}}
	\caption{Simulation of the PCP problem using two processes.}  
\label{tab:prog:unde}
\end{figure*}

The rest of this section is devoted to the proof of Theorem \ref{thm:undec}. The 
undecidability is obtained by a reduction from Post's Correspondence Problem (PCP) \cite{post}. 
 A PCP instance consists of two  sequences $u_1, \ldots, u_n$ and $v_1, \ldots, v_n$ of non-empty words
over some  alphabet $\Sigma$. Checking whether there exists a sequence of indices $j_1, \dots, j_k \in \{1, \dots, n\}$
s.t. $u_{j_1} \dots u_{j_k}=v_{j_1} \dots v_{j_k}$ is undecidable. 

We construct a concurrent program with two processes $p_1$ and $p_2$ (see Figure \ref{tab:prog:unde}),
six memory locations $\varset=\{x, y, \mathit{validate}, \mathit{index}, \mathit{index}',\mathit{term}\}$,
and two registers $\{\reg,\reg'\}$.
The finite data domain of $\prog$ is defined as $\mathsf{Val}=\Sigma \cup \{0,1, \dots, n\} \cup \{\S,\#\}$,
where $\S$ and $\#$ are two  special symbols (not in $\Sigma \cup \{0,1,\dots,n\}$).
All the locations and registers are initialized to zero. We show that reaching the instructions annotated  by
$\pouto$ and $\poutt$ in $p_1, p_2$ is possible iff the PCP instance has a solution.  We give below an overview of  the execution steps leading to the annotated instructions.  
\begin{itemize}
\item[(1)]To begin, process $p_2$ writes 1 to the location $term$.
\item[(2)]Process $p_1$ 
promises to write letters of $u_i$ (one by one) to location $x$, and the respective indices $i$ to the location $index$. 
The number of made promises  is arbitrary, since it depends on the length of the PCP solution. Observe that the sequence of promises made to the variable $index$ corresponds to the guessed  solution of the PCP problem. 
\item[(C1)] Using the $\kwif$ branch, $p_1$ certifies its promise before switching out of context.  Note that fulfilment of promises is yet to be done.  
\item[(3)]  Process $p_2$ reads from  the sequences of promises written to $x$ and $index$ and copies them (one by one) to variables $y$ and $index'$ respectively, and reaches $\poutt$.	
\item[(4)] The $\kwelse$ branch in $p_1$ is enabled at this point, where $p_1$ reads the sequence of indices  from $index'$, and 
each time it reads an index $i$ from $index'$, it checks that it can read the sequence of letters of $v_i$ from $y$.  
\item[(C1)] $p_1$ copies (one by one) the sequence of observed values from $y$ and $index'$  back to $x$ and $index$ respectively. To fulfil the promises, 
it is crucial that the sequence of read values from $index'$ (resp. $y$) is the same as the sequence of written values to $index$ (resp. $x$).
Since $y$ holds a sequence $v_{i_1}\dots v_{i_k}$, the promises are fulfilled iff this sequence 
is same as the promised sequence $u_{i_1} \dots u_{i_k}$. This happens only when $i_1, \dots, i_k$ is a PCP solution.
 \item[(5)] At the end of promise fulfilment, $p_1$ reaches $\pouto$. 
\end{itemize}

Let us now give more details about the code  of the two processes  given in Figure~\ref{tab:prog:unde}.
Depending on the value of the $\mathit{validate}$ flag read,
process $p_1$ can run in generation mode ($\kwthen$ branch) or validation mode ($\kwelse$ branch).
In generation mode,
 $p_1$ writes in sequential manner the sequence of indices  (alternated with the special symbol $\#$) of a potential solution of the PCP problem to the location $\mathit{index}$ and  writes, letter by letter, the sequence of letters of the word $u_i$ to location $x$ each time $p_1$ sets the location $\mathit{index}$ to $i$ (using the $\keyword{Module^{\proc_1}_{u_i}}$ procedure). 
In validation mode, $p_1$ reads from locations $\mathit{index}'$ and $y$
and writes back what it has read, to the locations $\mathit{index}$ and $x$, respectively  (using the $\keyword{Module^{\proc_1}_{v_i}}$).
The second process proceeds in a similar manner as the $\kwelse$ branch of the first process:
It reads from locations $\mathit{index}$ and $x$ and writes the values read to $\mathit{index}'$ and $y$, respectively  (using the $\keyword{Module^{\proc_2}_{u_i}}$).
We will show that a solution of the PCP problem exists iff  we can reach the annotations $\pouto, \poutt$ respectively in processes $p_1, p_2$. 

Assume that a solution of the PCP problem exists.
This means that there is a sequence of indices $i_1,i_2,\ldots,i_k$ such that $v_{i_1} v_{i_2} \cdots v_{i_k}= u_{i_1} u_{i_2} \cdots u_{i_k}$.
Let $w= u_{i_1} u_{i_2} \cdots u_{i_k}$.
Let us show that the pair of annotations $\pouto, \poutt$ are reachable in $\prog$.
For that aim, consider the following   run of the program $\prog$:
$p_2$ starts first by setting the location $\mathit{term}$ to $1$.
Then, $p_1$ will use the $\kwthen$ branch of its conditional statement and make the two following sequences of promises
$(\mathit{index},i_1,(1,2]), (\mathit{index},i_2,(2,3]), \ldots, (\mathit{index},i_k,(k,k+1])$ and
$(x,w[1],(1,2]), (x,w[2],(2,3]), \ldots, (x,w[|w|],(|w|,|w|+1])$.
Observe that $p_1$ can certify such sequences of promises by iterating its iterative statement in the $\kwthen$ branch of its alternative statements.
Once these promises are performed, $p_2$ reads these two sequences and writes them back to the locations $\mathit{index}'$ and $y$, respectively.
$p_2$ then sets the location $validate$ to $1$.
Now $p_1$ can resume its execution by reading the location $validate$ written by the second process
and enter its $\kwelse$ branch of its alternative statement.
Then, $p_1$ will iteratively read the values written by $p_2$ on the location $\mathit{index}'$ and $y$ and write them back to the locations $\mathit{index}$ and $x$, respectively.
By doing this $p_1$ fulfils also the sequence of promises that has been issued.

Now assume that we can reach the pair of annotations 
$\pouto, \poutt$. 
In order for $p_1$ to reach $\pouto$, it must execute the $\kwelse$ branch of its conditional statement.
Let us assume it does so.
Then, $p_1$ will read the sequence of indices $i_1,i_2,\ldots,i_k$ written by the process $p_2$ on the location $\mathit{index}'$.
Let us assume that the process $p_2$ writes the sequence of indices $j_1,j_2,\ldots, j_m$ on the location $\mathit{index}'$ (by reading the sequence of promises made by $p_1$).
Each time that the process $p_1$ reads an index from the location $\mathit{index}'$, it writes it back on the location $\mathit{index}$.
The process $p_1$ (resp.\ $p_2$)  alternates between writing/reading an index in $\{1,\dots,n\}$  and the special symbol $\#$ in order to make sure that each written index is at most read once.
In similar manner, the process $p_2$ reads the sequence of indices $j_1,j_2, \ldots, j_m$ written by the process $p_1$ on the location $\mathit{index}$ and it writes it back on the locations $\mathit{index}'$.
This implies that the sequence $j_1,j_2, \ldots, j_m$ is a subsequence of $i_1,i_2, \ldots,i_k$ (since the process $p_2$  can miss reading some written indices by the process $p_1$) and also that the sequence $i_1,i_2, \ldots,i_k$  is  a subsequence of $j_1,j_2, \ldots, j_m$ (since $p_1$ can miss reading some written index by the process $p_2$).
Thus, we have that the sequences $i_1,i_2, \ldots,i_k$ and  $j_1,j_2, \ldots, j_m$ are the same.
Every time the process $p_1$ (resp.\ $p_2$) reads an index $i$ from the location $\mathit{index}'$ (resp.\ $\mathit{index}$),
it (1) tries to read in sequential manner the sequence of letters appearing in $v_i$ (resp.\ $u_i$) (alternated with the special symbol $\#$) from the location $y$ (resp.\ $x$),
and (2) writes the same sequence of letters to the location $x$ (resp.\ $y$).
Using a similar argument as in the case of indices, we can deduce that if $p_1$ (resp.\ $p_2$) writes the words $v_{i_1} v_{i_2} \cdots v_{i_k}$  (resp.\ $u_{j_1} u_{j_2} \cdots u_{j_m}$), letter by letter (with an alternation with the symbol$\#$), to the location $x$ (resp.\ $y$), then  $v_{i_1} v_{i_2} \cdots v_{i_k}$  (resp.\ $u_{j_1} u_{j_2} \cdots u_{j_m}$) is a subsequence of $u_{j_1} u_{j_2} \cdots u_{j_m}$ (resp.\ $v_{i_1} v_{i_2} \cdots v_{i_k}$).
Thus, if the pair of annotations $\pouto, \poutt$ are reachable then there exist two sequences $i_1,i_2, \ldots,i_k$ and $j_1,j_2, \ldots, j_m$, written, respectively, by $p_1$ and $p_2$ such that $i_1,i_2, \ldots,i_k$ is equal to $j_1,j_2, \ldots, j_m$, and 
$v_{i_1} v_{i_2} \cdots v_{i_k}$ is equal to $u_{j_1} u_{j_2} \cdots u_{j_m}$.
Observe that sequence of indices $i_1,i_2,\ldots,i_k$ is non-empty due to the assume statement $\assume(\reg' \in [1,n] )$.

\section{Decidable Fragments of $\ps$}
\label{sec:dec}
Since keeping $\ra$ memory accesses renders the reachability problem undecidable \cite{pldi2019} and so does having unboundedly many promises when having $\rlx$ memory accesses (Theorem \ref{thm:undec}), we address in this section the decidability problem for  $\psr$ with a bounded number of promises in any reachable configuration. Observe that bounding the number of promises in any reachable machine state  does not imply that the total number of promises made during that run is bounded. 
 Let  $\bps$ represent the restriction of $\psr$ to boundedly many promises where the  number of promises in  each reachable machine state is smaller or equal to a given constant.  In the following, we show the decidability of the reachability problem for 
 $\bps$. For establishing this result,  we 
 introduce an alternate memory model for concurrent programs which we call $\lhc$ (for 
 ``lossy higher order words''). We present the operational semantics of $\lhc$, and show that 
 $\psr$ is operationally equivalent to $\lhc$. 
 Then, under the bounded  promise assumption, we show how $\lhc$ is used to decide the  reachability problem for $\bps$.

 \subsection{Introduction to $\lhc$}
 Given an alphabet $A$, a simple word over $A$ is an element of $A^*$, while 
 a higher order word is an element of $(A^*)^*$ (i.e., word of words).  A \emph{state} of $\lhc$ maintains a collection of higher order words, one per location, along with the  states  of all processes. 
    The higher order word $\ch_x$ corresponding to the location $x$ is 
 a word of simple words, representing the sub memory $M(x)$ 
 in $\psr$.   Each simple word in $\ch_x$ 
  is an ordered sequence of ``memory types'', that is, 
  messages or promises  in the memory corresponding to $x$, 
 maintained in the order of their $\too$ timestamps in the memory. 
Unlike $\psr$, the $\lhc$ does not store 
timestamps in the messages and promises; instead, it takes advantage of the word order which induces a natural ordering amongst these without explicit use of timestamps. The key information to encode in each memory type occurring in $\ch_x$ is: 
(1) whether it is a message ($\msgg$) or a promise ($\prm$), 
(2) which process ($p$) 
added it to the memory, and the value ($\val$) it holds,  (3) the set $S$ (called pointer set) of processes that are aware of this message/promise (processes which point to this message/promise), and (4) whether 
the time interval to the right has been reserved by some process.  

\smallskip

\noindent{\bf {Memory Types}.} A \emph{memory type} is an element of  
$\Sigma=\{\msgg, \prm\} \times \mathsf{Val} \times \procset \times 2^{\procset}$ $\cup \Gamma= \{\msgg, \prm\} \times \mathsf{Val} \times \procset \times 2^{\procset} \times \procset$.  
  The first component represents a message ($\msgg$) or a promise ($\prm$)  in the memory $M$ of $\psr$,  the second component the value in the message/promise, the third component is the process 
which adds the message/promise to the memory and the fourth component is a 
\emph{pointer set}, which contains all processes whose local view agree with the $\too$ time stamp of the message/promise. In the case 
of $\Gamma$, we have a fifth component which holds the id of the  process that has  
reserved the  time slot to the right of this message/promise.

For a memory type $m=(r,v,p,S)$ (or $m=(r,v,p,S,q)$), we use  $m.value$ to denote $v$. 
For a memory type $m=(r,v,p,S)$ (resp.  $m=(r,v,p,S,q)$) and a process $h \in \procset$, we use $add(m,h)$ to denote the memory type $m=(r,v,p,S \cup \{h\})$ (resp.  $m=(r,v,p,S \cup \{h\},q)$). We use also $delete(m,h)$ to denote the memory type $m=(r,v,p,S \setminus \{h\})$ (resp.  $m=(r,v,p,S \setminus \{h\},q)$). This corresponds to the addition/deletion of the process $h$ to/from the set of pointers of the memory type $m$.
\smallskip

\noindent{\bf{Simple Words}.} A simple word is a word  $\in \Sigma^* \# (\Sigma \cup \Gamma)$,  
  and each $\ch_x$ is a word 
 $\in (\Sigma^* \# (\Sigma \cup \Gamma))^+$. $\#$ is a special symbol not in $\Sigma \cup \Gamma$, which separates the last symbol from the rest 
 of the simple word. Consecutive symbols of $\Sigma$ 
 in a simple word represent adjacent messages/promises 
 in the memory of $\psr$, and are hence unavailable 
 for a RMW.   The special symbol $\#$ segregates these 
 from the last symbol of $\Sigma \cup \Gamma$ in a simple word. $\#$ does 
 not correspond to any element from the memory; its job is simply 
 to demarcate the messages/promises which are not available for RMW 
 from the last symbol of the simple word. 
   If the last symbol in a simple word is 
in  $\Sigma$, then it is available for a RMW; if the last symbol is in $\Gamma$, then it is not available for a RMW since the  next message adjacent to this symbol is a reservation. The last symbol from $\Sigma \cup \Gamma$ in a simple word $\Sigma^* \# (\Sigma \cup \Gamma)$ thus represents a message/promise (combined with or not  a reservation) in the memory which is adjacent 
to the messages represented by the symbols immediately preceding $\#$ (if any). 

\vspace{0.3cm}
\begin{figure}[h]
\includegraphics[scale=.25]{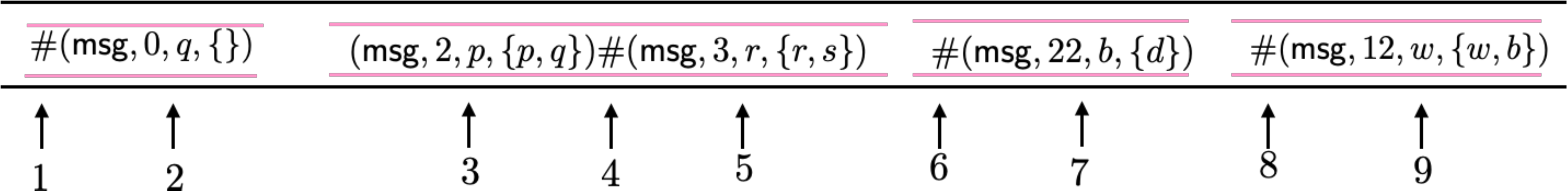}

\caption{A higher order word $\ch$.}	
\label{fig:ch1}
\end{figure}
\vspace{0.3cm}

\noindent{\bf{Higher order words}}. 
A \emph{higher order word}  is a sequence of simple words.    Figure \ref{fig:ch1} depicts a higher order word with four simple words. We use a left to right order
in both simple words and higher order words. Furthermore,  we extend in  the straightforward manner  the classical  word indexation strategy  to higher order words. For example, the symbol at the third position of the higher order word $\ch$ given in  Figure \ref{fig:ch1} is $\ch[3]=(\msgg,2,p,\{p,q\})$.  A higher order word $\ch$ is {\em well-formed} iff for every $\proc \in \procset$, there is a unique  
 position $i$ in $\ch$ having $p$ in its pointer set; that is, $\ch[i]$ is of the form  $(-, -, -, S) \in \Sigma$ or 
$(-, -, -, S, -) \in  \Gamma$  
 s.t. $p \in S$. Observe that the higher order word given in Figure \ref{fig:ch1}  is well-formed.  We will use $\ptr(p,\ch)$ to denote the unique position $i$ in $\ch$ having $p$ in its pointer set. Next, we  assume that all the manipulated  higher order words are well-formed.

As already mentioned, for each $x \in \varset$, we have a 
higher order word $\ch_x$.  The higher order word $\ch_x$ 
represents the entire space $[0, \infty)$ of available timestamps. Each simple word 
in $\ch_x$ represents a timestamp interval $(f, t]$, with consecutive 
  simple words representing disjoint timestamp intervals (while preserving order). The memory types 
in each simple word take up  \emph{adjacent} timestamp intervals, spanning 
the timestamp interval of the simple word. This adjacency of timestamp intervals 
 within simple words is mainly used  in RMW steps and reservations. 
    The memory type in $\Sigma$ occurring at the end of a simple word  denotes  a message/promise which is available for a RMW operation. The memory type in $\Gamma$ occurring at the end of  a simple word denotes a message/promise followed by a reservation and therefore it is not available for a RMW operation. 
        The memory types  at positions other than the rightmost 
 in a simple word, represent messages/promises  which are not available for  RMW. Figure \ref{fig:ch2} presents a mapping from a memory of   $\psr$ to a collection of  higher order words (one per location) in $\lhc$.
 
  Given a higher order word $\ch$, a position $i \in \{1,\ldots, |\ch|\}$, and  $p \in \procset$ , we use $add(\ch,p,i)$ (resp. $delete(\ch,p)$) to denote the higher order word $\ch[1,i-1] \cdot add(\ch[i],p) \cdot \ch[i+1,|\ch|]$ (resp.  $\ch[1,i-1] \cdot delete(\ch[\ptr(p,\ch)],p) \cdot \ch[i+1,|\ch|]$). This corresponds to the addition/deletion of  $p$ to/from the set of pointers of $\ch[i]$/$\ch[\ptr(p,\ch)]$. We use $move(\ch,p,i)$ to denote $add(delete(\ch,p),p,i)$.

\vspace{0.3cm}
\begin{figure}[h]
\centering
\newtcbox{\colorboxouline}[1][]{boxsep=0.5pt,left=0.1pt,right=0.1pt,top=1pt,bottom=1pt,colframe=magenta,colback=white,boxrule=0pt,toprule=1pt,bottomrule=1pt,sharp corners,#1}

\tikzstyle{simprect}=[fill={rgb,255: red,191; green,191; blue,191}, draw=black, shape=rectangle, minimum width=0.5cm, minimum height=0.5cm, tikzit draw=black, tikzit fill={rgb,255: red,191; green,191; blue,191}, tikzit shape=rectangle]
\tikzstyle{simprect2}=[fill={rgb,255: red,207; green,152; blue,202}, draw=black, shape=rectangle, tikzit draw=black, tikzit fill={rgb,255: red,207; green,152; blue,202}, tikzit shape=rectangle, minimum width=0.5cm, minimum height=0.5cm]
\tikzstyle{int1}=[fill={rgb,255: red,191; green,191; blue,191}, draw=black, shape=rectangle, tikzit fill={rgb,255: red,191; green,191; blue,191}, tikzit draw=black, tikzit shape=rectangle, minimum width=0.6cm, minimum height=0.4cm]

\tikzstyle{simplearrow}=[->,  line width=0.8pt]
\scalebox{0.6}{
\begin{tikzpicture}
	\begin{pgfonlayer}{nodelayer}
		\node [style=none] (0) at (0, 0.5) {};
		\node [style=none] (1) at (0, 5.5) {};
		\node [style=none] (2) at (8, 0.5) {};
		\node [style=none] (3) at (-0.5, 4.75) {\Large Locs};
		\node [style=none] (4) at (5, 0) {\Large Timestamp};
		\node [style=none] (5) at (-0.5, 3) {\Large $M(y)$};
		\node [style=none] (6) at (-0.5, 1.5) {\Large $M(x)$};
		\node [style=none] (7) at (0, 4) {};
		\node [style=none] (8) at (0, 3) {};
		\node [style=none] (9) at (0, 2) {};
		\node [style=none] (10) at (0, 1) {};
		\node [style=none] (11) at (3, 1) {};
		\node [style=none] (12) at (3, 2) {};
		\node [style=none] (13) at (1.5, 3) {};
		\node [style=none] (14) at (1.5, 4) {};
		\node [style=none] (15) at (4.5, 1) {};
		\node [style=none] (16) at (4.5, 2) {};
		\node [style=none] (17) at (6.5, 2) {};
		\node [style=none] (18) at (6.5, 1) {};
		\node [style=none] (19) at (3.5, 4) {};
		\node [style=none] (20) at (3.5, 3) {};
		\node [style=none] (21) at (5, 3) {};
		\node [style=none] (22) at (5, 4) {};
		\node [style=none] (23) at (6.5, 4) {};
		\node [style=none] (24) at (6.5, 3) {};
		\node [style=simprect] (33) at (1.5, 5) {};
		\node [style=simprect2] (34) at (1.5, 4.25) {};
		\node [style=none] (35) at (3.5, 5) {\Large promises/messages};
		\node [style=none] (36) at (3, 4.25) {\Large reservations};
		\node [style=none] (37) at (11.5,3) {\colorboxouline{\Large $\#(\_,v_4,\_,\_)$}};
		\node [style=none] (38) at (16,3) {\colorboxouline{\Large $\#(\_,v_6,\_,\_,\_)$}};
		\node [style=none] (39) at (11.75,1.5) {\colorboxouline{\Large $(\_,v_4,\_,\_)(\_,v_3,\_,\_)\#(\_,v_1,\_,\_)$}};
		\node [style=none] (40) at (17,1.5) {\colorboxouline{\Large $(\_,v_2,\_,\_)\#(\_,v_5,\_,\_)$}};
		\node [style=none] (41) at (14,0.5) {\Large $\mathsf{HW}_x$};
		\node [style=none] (42) at (14,4) {\Large $\mathsf{HW}_y$};

	\end{pgfonlayer}
	\begin{pgfonlayer}{edgelayer}
		\draw[line width=1pt] (8.75,3.5) -- (19,3.5);
		\draw[line width=1pt] (8.75,2.5) -- (19,2.5);
		\draw[line width=1pt] (8.75,1) -- (19,1);
		\draw[line width=1pt] (8.75,2) -- (19,2);
		\filldraw[] [style=simplearrow] (0.center) to (1.center);
		\filldraw[] [style=simplearrow] (0.center) to (2.center);
		\filldraw[fill=gray!40] (0,1) rectangle (4.5,2) node[pos=.5] {\Large$(\_,v_4,\_)(\_,v_3,\_)(\_,v_1,\_)$};
		\filldraw[fill=gray!40] (5,1) rectangle (8,2) node[pos=.5] {\Large$(\_,v_2,\_)(\_,v_5,\_)$};
		\filldraw[fill=gray!40] (0,2.5) rectangle (2,3.5) node[pos=.5] {\Large$(\_,v_4,\_)$};
		\filldraw[fill=gray!40] (3.5,2.5) rectangle (5.5,3.5) node[pos=.5] {\Large$(\_,v_6,\_)$};
		\filldraw[fill=violet!40, draw=green!40!black] (5.5,2.5) rectangle (7,3.5) node[pos=.5] {\Large$(\_,\_)$};;
	\end{pgfonlayer}
\end{tikzpicture}
}
\caption{A mapping from  memories $M(x), M(y)$
 to higher order words $\ch_x, \ch_y$, respectively.}	
\label{fig:ch2}
\vspace{0.2cm}
\end{figure}
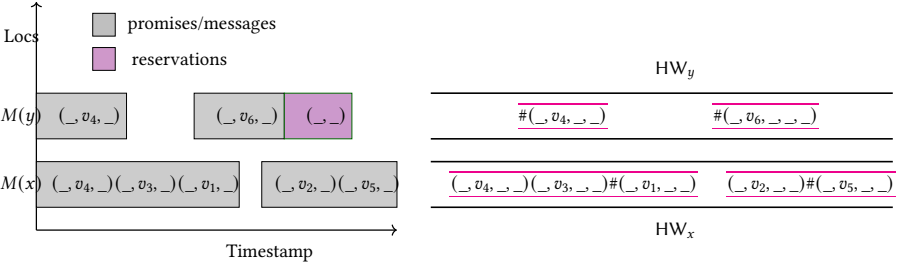

\smallskip

\noindent{\bf {Initializing higher order words}.} For each location $x \in \varset$, the initial higher order word $\ch^{\rm init}_x$ is defined  as \includegraphics[scale=.30]{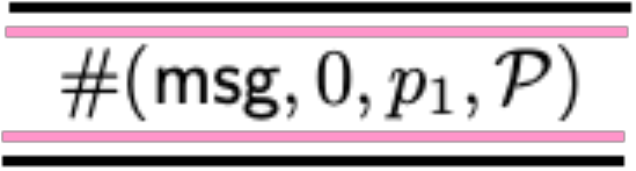}, where  $\procset$ is the set of all processes 
and $p_1$ is some process in $\procset$. The set of all higher order words $\ch_x^{\rm init}$ for all locations $x$ 
represents the initial memory of $\psr$ where all locations have value 0, and all processes 
are aware of the initial message. 
  
\smallskip

\noindent{\bf Simulating Reads, Writes, RMWs in $\lhc$.}
In the following, we informally describe how to handle $\psr$ instructions 
in $\lhc$. Since we only have 
the $\rlx$ access mode, we denote 
Reads, Writes and RMWs as $\wt(x,v)$, $\rd(x,v)$ and $\upd(x, v_r, v_w)$, dropping the access modes. 

\paragraph{Reads} A $\rd(x,v)$  step  by a process $p$  (reading  $v$ from  $x$) is handled as follows in $\lhc$.

There exists an index $j \geq \ptr(p,\ch_x)$ in $\ch_x$ such that $\ch_x[j]$  is of the form  $(-, v, -, S')$ or $(-, v, -, S',-)$. This corresponds to the existence of a memory type holding the value $v$ in $\ch_x$ and this symbol is on the right of the current view/pointer of the process $p$.

Add $p$ to the set of pointers  $S'$ and remove it from  its previous position. 

\paragraph{Writes} A $\wt(x,v)$ step  by a process $p$ (writing the value $v$ to the location $x$) in $\psr$ is done by  adding a new message with a timestamp 
higher than the local view of $p$ for $x$: the timestamp interval of this new message can be adjacent 
to the timestamp of the local view of $p$, or much ahead. These two possibilities 
are captured in $\lhc$ as follows. 

(1) Add the simple word \includegraphics[scale=.22]{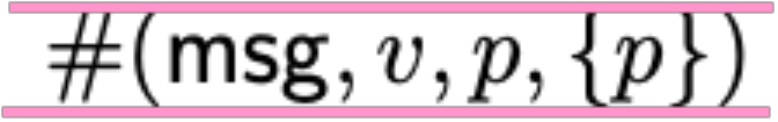} to $\ch_x$ to the right of $\ptr(p,\ch_x)$, or 

(2) there is  a symbol $\alpha \in \Sigma$ and two words  $w$ and $w'$ such that $\ch_x=w \cdot \# \cdot\alpha\cdot w'$. Then, update the higher order word $\ch_x$ to  $ w\cdot \alpha\cdot \# \cdot(\msgg, v, p, \{p\})\cdot w'$.

Finally, remove $p$ from its  previous pointer set.

\paragraph{(RMW)} Capturing RMWs is similar to the execution of a read followed by a  write. 
 In $\psr$, a process $p$ performing RMW reads from a message with a timestamp interval 
$(,t]$ and adds a message to the memory with timestamp interval $(t,-]$. This is handled as follows in $\lhc$, and shows 
the need for the higher order words. Consider a $\upd(x, v_r, v_w)$ step by  $p$. Then, 

there is a simple word \includegraphics[scale=.22]{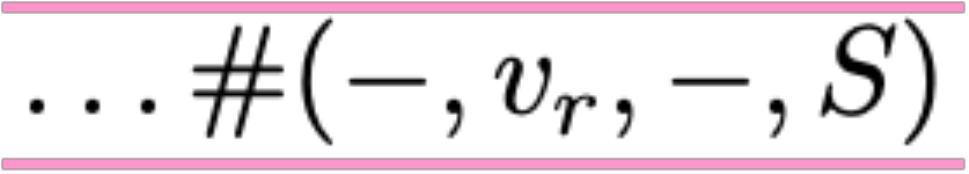} in $\ch_x$ having 
$(-, v_r, -, S)$ as the last memory type in it, and the position of the memory type  $(-, v_r, -, S)$ is on the right of  the current pointer of $p$ in $\ch_x$.

$p$ is removed from its pointer set, 

$\#(-, v_r, -, S)$ is replaced with $(-, v_r, -, S\backslash\{p\})\#$ and 
$(-, v_w, p, \{p\})$ is appended, resulting in extending 
\includegraphics[scale=.22]{rmw.pdf} to  
 \includegraphics[scale=.22]{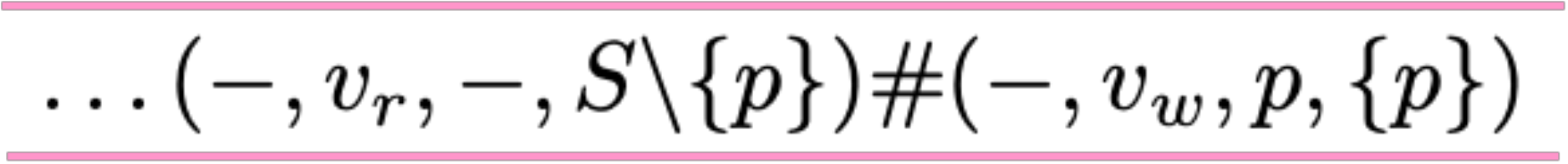}.

\begin{example}
We illustrate the read, write and RMW in $\lhc$ on an example. Figure \ref{fig:sim} depicts a run in $\psr$ and the corresponding run in $\lhc$. The run of $\psr$ shows how the memory evolves, and the corresponding run 
in $\lhc$ faithfully simulates this using higher order words $\ch_x$ and $\ch_y$.

\begin{center}
\begin{tabular}[t]{c||c}
\begin{lstlisting}[style=examples,tabsize=3]
x:=1
y:=2
x:=3
\end{lstlisting}
&
\begin{lstlisting}[style=examples,tabsize=3]
x:=5
$r1:=x //3
$r2:= FADD(y,1) //2
\end{lstlisting}
\end{tabular}
 \end{center}
 
 \begin{figure}[ht]
 \includegraphics[scale=.14]{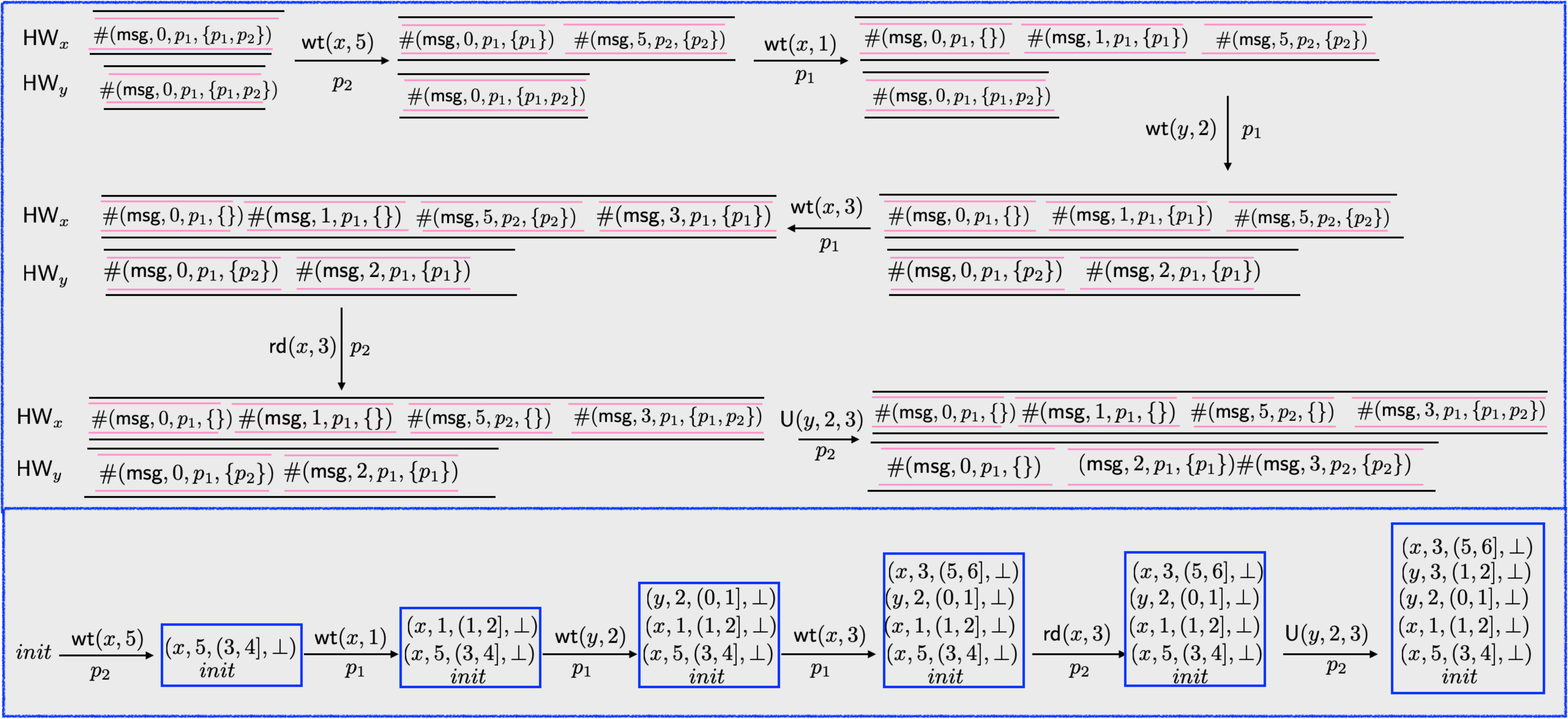}	
 \caption{Below, a run in $\ps$ showing the changes to memory,  and  above, the corresponding run in $\lhc$. Observe that $init$ stands for the initial memory.}
 \label{fig:sim}
 \end{figure}
 \end{example}

\noindent{\bf Promises  in $\lhc$.}
Next, we discuss how to handle promises. 
\paragraph{Promises} Handling promises made by a process $p$ in $\psr$ is similar to handling $\wt(x,v)$: we add the simple word \includegraphics[scale=.22]{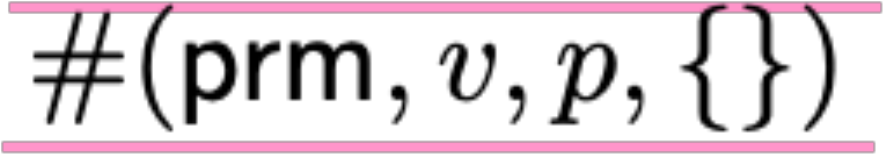} in $\ch_x$  
to the right of the position $\ptr(p, \ch_x)$, or 
append $(\prm, v, p, \{\})$ at the end of  a simple word with  a position larger than $\ptr(p,\ch_x)$. 
 Other than tagging the symbol as a promise ($\prm$), the pointer set is empty. 
  
\smallskip

\noindent{\bf Reservations and Cancellations in $\lhc$.} Next, we come to one of the new features of $\ps$ over the first version, namely, reservations and cancellations. In $\psr$, a process $p$ makes a reservation by adding the pair $(x, (f,t])$ 
to the memory, given that there is a message/promise in the memory with timestamp interval $(-,f]$.  In $\lhc$ this is captured by ``tagging'' 
the rightmost memory type (message/promise) in a simple word with the name of the process that makes the reservation. This requires us to consider the memory types from  $\Gamma=\{\msgg, \prm\} \times \mathsf{Val} \times \procset \times 2^{\procset}\times \procset$
 where the last component stores the process which made the reservation. 
 Such a memory type always appears at the end of a simple word, and represents that the next timestamp interval adjacent to it has been reserved.  Observe that we can not add new memory types to the right of  a memory type  of the form 
 $(\msgg, v, p, S, q)$. 
    Thus, reservations are handled as follows.
  \begin{enumerate} 
\item[(Res)] 
Assume the rightmost symbol 
in a simple word as $(\msgg, v, p, S)$. To capture the 
reservation by $q$, $(\msgg, v, p, S)$ is replaced with $(\msgg, v, p, S,q)$. 
\item[(Can)] A cancellation is done by removing the last component 
$q$ from $(\msgg, v, p, S,q)$ resulting 
in $(\msgg, v, p, S)$. 
\end{enumerate}

\noindent{\bf Empty Memory Types, Redundant simple words.}
When a process $p$  
reads from a message, the pointer of $p$ is updated, and moves forward. As a result, we may have memory types of the form $(\msgg, v, p, \{\})$  as well as $(\msgg, v, p, \{\},q)$  
representing those messages in the memory whose pointer set is empty.  
 Call such symbols of $\Sigma \cup \Gamma$ \emph{empty memory types}. It is  then possible to lose an  \emph{empty memory type} of $\Sigma$ from a simple word if it is not at the rightmost position. This will  not have any consequence  with respect to the    reachability problem, since processes can non-deterministically skip reading some messages in the memory.   Likewise, a simple word of the form $w \# m \in \Sigma^* \# (\Sigma \cup \Gamma)$ where all symbols 
  in $w$ are empty memory types from $\Sigma$ and $m$ is an empty memory type from $\Sigma \cup \Gamma$ 
   can be lost entirely. Such simple words are called  \emph{redundant simple words}.  Given this, what cannot be lost from $\ch_x$? The following:  
\begin{itemize}[leftmargin=*]
\item memory types  $(\prm, -,-,-)$ or  	$(\prm, -,-,-,-)$
 representing promises. This is due to the fact promises should be fulfilled and therefore can not be lost.
 \item non \emph{empty memory types}: the pointer set of these contain at least one process. Since losing any of these  memory types will result in losing the  pointer/view of at least one of the processes. 
\item Only rightmost memory type (right next to $\#$) in a  simple word. Losing only this memory type will result in a non well-defined higher order word. 
 \end{itemize}

\smallskip

\noindent{\bf Certification and Fulfilment.}
In $\psr$, certification, for a process $\proc$, happens from the capped memory, where  
intermediate time slots (other than reserved ones) are blocked, and any new message can be added 
only at the maximal timestamp. 
This is handled in $\lhc$ by one of the following:
\begin{itemize}[leftmargin=*]
	\item  addition 
of new memory types  is only allowed only at the right end of any $\ch_x$,
\item If the rightmost memory type $m$ in $\ch_x$ is of the form  $(-, v, -, -, q)$ with $q \neq p$ (i.e., tagged by a reservation for  $q$), then  a  simple word $\#(\msgg, v, q,\{\})$  is appended at the  end of $\ch_x$.   
\end{itemize}
 Memory is altered in $\psr$ during  certification phase to check for promise fulfilment, 
and at the end of the certification phase, 
 we resume from the memory 
which was there before.   To capture this in $\lhc$, we work on a duplicate of $(\ch_x)_{x \in \varset}$ 
 in the certification phase. Notice that the duplication 
allows losing some of empty memory types and redundant simple words  non deterministically (as described in the previous paragraph). 
 This copy of $\ch_x$ is then modified during certification, and 
is discarded once we finish the certification phase. 

The fulfilment of a promise by  $p$ using the   rule $\stackrel{L}{\hookleftarrow}$ (see rule $\mathsf{(MEMORY: FULFILL)}$ in Figure \ref{program_sem}) will be handled in a similar manner as  using the rule  $\stackrel{A}{\hookleftarrow}$ (since we are only dealing   with the fragment of $\ps$ restricted to $\rlx$). This will result in replacing a memory type of the form $(\prm, v, p, S)$  (resp. $(\prm, v, p, S,q)$) by $(\msgg, v, p, S)$ (resp. $(\msgg, v, p, S,q)$) if this memory type is in a position which is on the right of the current pointer of the process $p$. Then, the process $p$ is added to the pointer set $S$  while removing it from  the previous  pointer set it belongs to.

The fulfilment of a promise by a process $p$ in $\ps$ using the   rule $\stackrel{S}{\hookleftarrow}$ (see rule $\mathsf{(MEMORY: FULFILL)}$ in Figure \ref{program_sem})  results in  splitting the intervals of the promise, when adding a new message $\msgnew{\xvar}{v'}{f}{t}{\bot}$  to the memory. 
 To capture this, we allow insertion
of a memory type right before the promise whose interval 
is split. This will result in replacing a memory type of the form $(\prm, v, p, S)$ (resp. $\#(\prm, v, p, S,q)$) by $(\msgg, v', p, \{p\}) (\prm, v, p, S)$ (resp. $(\msgg, v', p, \{p\}) \# (\prm, v, p, S,q)$) if this memory type is in a position which is on the right of the current pointer of the process $p$. Then, the process $p$ is removed from   the previous  pointer set it belongs to. We may also need to update the position of the separator $\#$ so that it is just before the last symbol of a simple word.

\smallskip

\noindent{\bf SC fences}. SC-fences are handled by adding a dummy process $g$ 
to $\procset$. Whenever a process $p$ performs a SC fence, 
 $g, p$ are added to the  same pointer set, by moving 
 $g$ ($p$) to the pointer set of $p$ ($g$) depending on which is 
 more to the right.

\begin{example}
Figure \ref{fig:sim3} illustrates a run in $\lhc$ on a program where  promises are necessary to reach the annotated part $\textcolor{sangria}{//}$.   
To reach the annotated part in P1, the execution proceeds as follows. C1, C2 represent 
two certification phases.
\begin{enumerate}
	\item[(1)] P1 promises the write of 42 to $x$, by a message $(x, 42, (f,t], \bot)$.
	\item[(C1)] To certify, P1 begins from the capped memory, 
	and enters the else branch. It begins a duplicate of the higher order words, and works on them in this phase. 
	\begin{itemize}
	\item 	 Since all positions in $(0,t]$ are blocked, 
	P1 splits the interval $(f,t]$ to write 41 to 
	$x$, and modifies the memory to $(x, 42, (t',t], \bot)$, 
	$(x, 41, (f,t'], \bot)$.
	\item P1 fulfils its promise
	\end{itemize}
\item[(2)] P2 reads 42 from $x$ and writes 42 to $z$
\item[(3)] P1 reads 42 from $z$
\item[(4)] P1 fulfils its promise, and reaches the annotated part. 	
	\end{enumerate}

\begin{figure}[ht]
\includegraphics[scale=0.13]{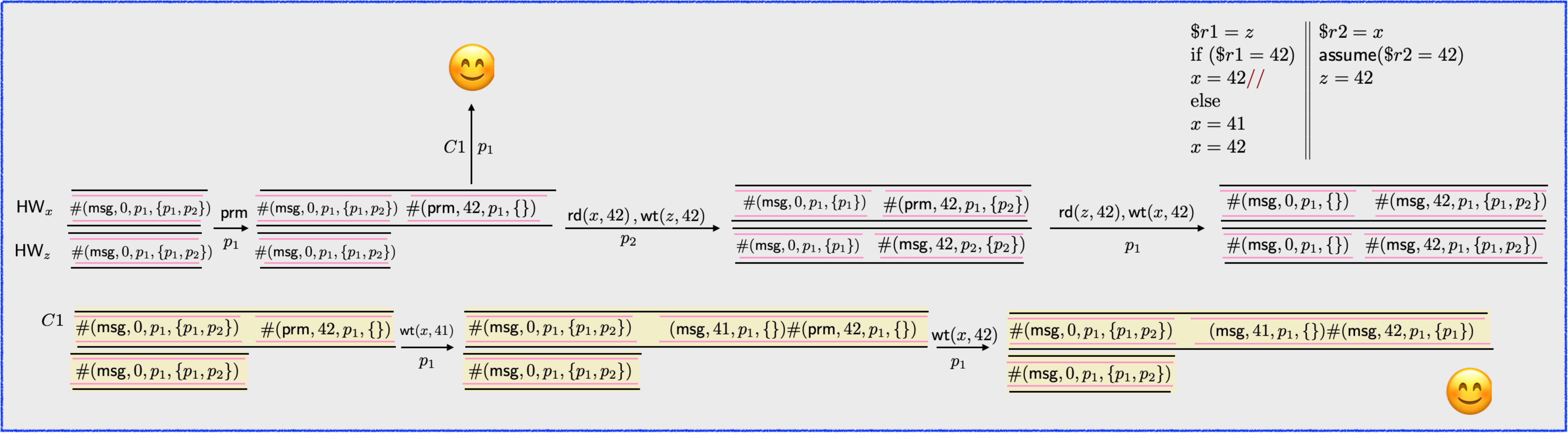}	
\caption{Run in $\lhc$. The certification phase works on the duplicates of $\ch_x, \ch_z$ denoted in yellow.}
\label{fig:sim3}
\end{figure}
\label{eg:sim3}
\end{example}

\subsection{Formal Model of  $\lhc$}
\label{sec:formal}
In the following, we formally define $\lhc$ and state the equivalence  of  the  reachability problem in $\psr$ 
and $\lhc$.

\smallskip

\noindent{\bf Insertion into higher order words.}
A higher order word $\ch$ can be extended in position  $1 \leq j\leq |\ch|$ with  a memory type $m$ of the form $(r, v, p, \{p\})$  in a number of ways:
\smallskip

\noindent
$\bullet$ {\it Insertion as a new simple word.}
$\ch \underset{j}{\stackrel{N}{\hookleftarrow}} m$ is defined only if 
$\ch[j-1]= \#$ (i.e., the position $j$ is the end of a simple word). Let $\ch'$ be the higher order word defined as $delete(\ch,p)$  (i.e., removing $p$ from its  previous set of pointers). Then, the extended higher order   $\ch \underset{j}{\stackrel{N}{\hookleftarrow}} m$ is defined as   $\ch'[1,j] \cdot  \# m \cdot \ch'[j+1,|\ch|]$  (i.e.,  inserting the new simple word just after the position $j$).

\smallskip

\noindent
$\bullet$ {\it Insertion at the end of a simple word.}
$\ch \underset{j}{\stackrel{E}{\hookleftarrow}} m$ is defined only if 
$\ch[j-1]= \#$ (i.e., the position $j$ is the end of a simple word) and $\ch[j] \in \Sigma$ (i.e., the last memory type in the simple word should be free from reservations). Let $\ch'$ be the higher order word defined as $delete(\ch,p)$.  Then,  the extended higher order   $\ch \underset{j}{\stackrel{E}{\hookleftarrow}} m$ is defined as   $w_1 \cdot  m' \cdot  \# m \cdot w_2$  with $\ch'= w_1 \cdot \# m'  \cdot w_2 $, and $m' \in \Sigma$, and $|w_1 \cdot \# m'|= j$ (i.e., inserting the new memory type just after the position $j$).

\smallskip

\noindent
$\bullet$ {\it Splitting a promise.}
$\ch \underset{j}{\stackrel{SP}{\hookleftarrow}} m$ is defined only if 
$\ch[j]$ is of the form $(\prm,-,p,-)$ or $(\prm,-,p,-,-)$ (i.e., the memory type at position $j$ is a promise). Let $\ch'$ be the higher order word defined as $delete(\ch,p)$.  Then,  the extended higher order   $\ch \underset{j}{\stackrel{SP}{\hookleftarrow}} m$ is defined as   $(1)$ $\ch'[1,j-2] \cdot  m \cdot \# m'    \cdot \ch'[j+1,|\ch|]$  if $\ch'[j]=m'$ and  $\ch'[j-1]=\#$, or $(2)$ $\ch'[1,j-1] \cdot  m \cdot  m'    \cdot  \ch'[j+1,|\ch|]$  if $\ch'[j]=m'$ and  $\ch'[j-1] \neq \#$. Observe that in both cases we are inserting  
the new memory type $m$ just before the position $j$.

\smallskip

\noindent
$\bullet$ {\it Fulfilment of  a promise.}
$\ch \underset{j}{\stackrel{FP}{\hookleftarrow}} m$ is defined only if 
$\ch[j]$ is of the form $(\prm,v,p,S)$ or $(\prm,v,p,S,q)$. Let $\ch'$ be the higher order word defined as $delete(\ch,p)$.  Then,  the extended higher order   $\ch \underset{j}{\stackrel{FP}{\hookleftarrow}} m$ is defined as  $\ch'[1,j-1] \cdot  m'    \cdot \ch'[j+1,|\ch'|]$  with   $m'=(\msgg,v,p,S\cup \{p\})$ if $\ch[j]=(\prm,v,p,S)$ and   $m'=(\msgg,v,p,S\cup \{p\},q)$ if $\ch[j]=(\prm,v,p,S,q)$.

\smallskip

\noindent
$\bullet$ {\it Splitting a reservation.}
$\ch \underset{j}{\stackrel{SR}{\hookleftarrow}} m$ is defined only if 
$\ch[j]$ is of the form $(r',v',q,S,p)$. Let $\ch'$ be the higher order word defined as $delete(\ch,p)$.  Then,  the extended higher order   $\ch \underset{j}{\stackrel{SR}{\hookleftarrow}} m$ is defined as    $\ch'[1,j-2] \cdot    (r',v',q,S)\cdot \#  (r,v,p,\{p\},p) \cdot \ch'[j+1,|\ch|]$. Observe that the new message $ (r,v,p,\{p\},p)$ is added to the right of the position $j$ which corresponds to the slot that has been reserved by $p$. This special splitting rule will be used during the certification phase. This will allow the process $p$ to  use   the reserved slots. Recall that it is not allowed to add memory types in the middle of the higher order words (other than the reserved ones) during the certification phase.
\smallskip

\noindent{\bf Making/Canceling a reservation.}
A higher order word $\ch$ can also be   modified through making/cancelling  a reservation at a position $1 \leq j\leq |\ch|$ by a process $p$. Thus, we define the operation ${Make}(\ch,p,j)$ (resp. ${Cancel}(\ch,p,j)$) that reserves (resp. cancels) a time slot at the position $j$.  ${Make}(\ch,p,j)$ (resp. ${Cancel}(\ch,p,j)$) is only defined if $\ch[j]$ is of the form $(r,v,q,S)$ (resp. $(r,v,q,S,p)$) and $\ch[j-1]=\#$. Then,  the extended higher order ${Make}(\ch,p,j)$ (resp. ${Cancel}(\ch,p,j)$) is defined as  $\ch[1,j-1] \cdot    (r,v,q,S,p)\cdot  \ch[j+1,|\ch|]$ (resp. $\ch[1,j-1] \cdot    (r,v,q,S)\cdot  \ch[j+1,|\ch|]$).

\smallskip

\noindent{\bf{Process configuration in $\lhc$}.} A configuration  of  $\proc \in \procset$ in $\lhc$     consists of a pair  $(\sigma, {\chh})$ where $(1)$  $\sigma$  is the process state maintaining the 
instruction label and the  register values  (see Subsection \ref{sec:ps}), and ${\chh}$ is  a mapping from the set of locations to   higher order words. 
 The transition relations $\xrightarrow[\proc]{\nor}$ and $\xrightarrow[\proc]{\cert}$ between process configuration is given in Figure  \ref{ps-program_sem}. The transition relation $\xrightarrow[\proc]{\cert}$ is used only in the certification phase while $\xrightarrow[\proc]{\nor}$ is used to simulate the standard phase of $\psr$.
A read operation in both phases (standard and certification) is handled by reading a value from a memory type which is on the right of the current pointer of $\proc$. A write operation, in the standard phase, can result in the insertion, on the right of the current pointer of $p$, of a new memory type at the end of a simple word or as a new simple word.  The memory type resulting from a  write in the certification phase is only allowed to be  inserted at the end of the higher order word or at the reserved slots (using the rule splitting a reservation). Write can also be used to fulfil a promise or to split a promise (i.e., partial fulfilment) during the both phases. Making/canceling a reservation will result in tagging/untagging a  memory type at the end of a simple word on the right of the current pointer of $p$. The case of RMW  is similar to a read followed by a write operations (whose resulting memory type should be inserted to the right of the read memory type). Finally, a promise can only be made during the standard phase and the resulting memory type will be  inserted  at the end of a simple word or as a new word on the right of the current pointer of $p$.

\tikzset{background rectangle/.style={fill=none,rounded corners,thick}}
\begin{figure}[t]
\resizebox{1.0\textwidth}{!}{
  \begin{tikzpicture}[codeblock/.style={line width=0.5pt, inner xsep=0pt, inner ysep=0pt}, show background rectangle]
\node[codeblock] (init) at (current bounding box.north west) {
$
\begin{array}{cc}
\rowcolor{red!10}
\displaystyle\frac{\sigma \xrightarrow[p]{\rd(x,v)} \sigma', ~~ i \geq \ptr(p,{\chh}(x)), ~~ v={\chh}(x)[i].value, ~~ {\chh'}={\chh}[x \mapsto move({\chh}(x),p,i)]}
{(\sigma,{\chh}) \xrightarrow[\proc]{a}
(\sigma',{\chh'})} 
&  \begin{array}{c}{\texttt{Read}}\\{a\in \{\cert,\nor\}}\end{array}
\\
\rowcolor{blue!5}
\displaystyle\frac{\sigma \xrightarrow[p]{\wt(x,v)} \sigma', ~~ i > \ptr(p,{\chh}(x)), ~~ {\chh'}={\chh}[x \mapsto ({\chh}(x)\underset{i}{\stackrel{K}{\hookleftarrow}} (\msgg,v,p,\{p\}))]}
{(\sigma,{\chh}) \xrightarrow[\proc]{a}
(\sigma',{\chh'})} 
&  \begin{array}{c}{\texttt{(Partial) fulfilment} (write)}\\ {a\in \{\cert,\nor\}, K \in \{SP,FP\}}\end{array}
\\
\rowcolor{green!5}
\displaystyle\frac{\sigma \xrightarrow[p]{\wt(x,v)} \sigma', ~~ i \geq \ptr(p,{\chh}(x)), ~~ {\chh'}={\chh}[x \mapsto ({\chh}(x)\underset{i}{\stackrel{K}{\hookleftarrow}} (\msgg,v,p,\{p\}))]}
{(\sigma,{\chh}) \xrightarrow[\proc]{\nor}
(\sigma',{\chh'})} 
&  \begin{array}{c}{\texttt{Standard write}}\\ {K \in \{N,E\}}\end{array}
\\
\rowcolor{red!10}
\displaystyle\frac{\sigma \xrightarrow[p]{\wt(x,v)} \sigma', ~~ i = |{\chh}(x)|, ~~ {\chh'}={\chh}[x \mapsto ({\chh}(x)\underset{i}{\stackrel{K}{\hookleftarrow}} (\msgg,v,p,\{p\}))]}
{(\sigma,{\chh}) \xrightarrow[\proc]{\cert}
(\sigma',{\chh'})} 
&  \begin{array}{c}{\texttt{Certification write}}\\ {K \in \{N,E\}}\end{array}
\\
\rowcolor{blue!5}
\displaystyle\frac{\sigma \xrightarrow[p]{\wt(x,v)} \sigma', ~~ i \geq \ptr(p,{\chh}(x)), ~~ {\chh'}={\chh}[x \mapsto ({\chh}(x)\underset{i}{\stackrel{SR}{\hookleftarrow}} (\msgg,v,p,\{p\}))]}
{(\sigma,{\chh}) \xrightarrow[\proc]{\cert}
(\sigma',{\chh'})} 
&  \begin{array}{c}{\texttt{Splitting a reservation (write)}}\\ {}\end{array}
\\
\rowcolor{green!10}
\displaystyle\frac{i \geq \ptr(p,{\chh}(x)), ~~ {\chh'}={\chh}[x \mapsto {Make}({\chh}(x),p,i)]}
{(\sigma,{\chh}) \xrightarrow[\proc]{\nor}
(\sigma,{\chh'})} 
&  \begin{array}{c}{\texttt{Making a reservation}}\\ {}\end{array}
\\
\rowcolor{red!5}
\displaystyle\frac{i \geq \ptr(p,{\chh}(x)), ~~ {\chh'}={\chh}[x \mapsto {Cancel}({\chh}(x),p,i)]}
{(\sigma,{\chh}) \xrightarrow[\proc]{a}
(\sigma,{\chh'})} 
&  \begin{array}{c}{\texttt{Cancelling a reservation}}\\ {a\in \{\cert,\nor\}}\end{array}
\\
\rowcolor{blue!5}
\displaystyle\frac{\sigma \xrightarrow[p]{\upd(x,v_r,w_r)} \sigma', ~~ i \geq \ptr(p,{\chh}(x)), ~~ v_r={\chh}(x)[i].value, ~~ {\chh'}={\chh}[x \mapsto ({\chh}(x)\underset{i}{\stackrel{E}{\hookleftarrow}} (\msgg,w_r,p,\{p\}))]}
{(\sigma,{\chh}) \xrightarrow[\proc]{\nor}
(\sigma',{\chh'})} 
&  \begin{array}{c}{\texttt{Standard update}}\\ {}\end{array}
\\
\rowcolor{green!10}
\displaystyle\frac{\sigma \xrightarrow[p]{\upd(x,v_r,w_r)} \sigma', ~~ i=|{\chh}(x)|, ~~ v_r={\chh}(x)[i].value, ~~ {\chh'}={\chh}[x \mapsto ({\chh}(x)\underset{i}{\stackrel{E}{\hookleftarrow}} (\msgg,w_r,p,\{p\}))]}
{(\sigma,{\chh}) \xrightarrow[\proc]{\cert}
(\sigma',{\chh'})} 
&  \begin{array}{c}{\texttt{Certification Update}}\\ {}\end{array}
\\
\rowcolor{red!5}
\displaystyle\frac{\sigma \xrightarrow[p]{\upd(x,v_r,w_r)} \sigma', ~~ i \geq \ptr(p,{\chh}(x)), ~~ v_r={\chh}(x)[i].value, ~~ {\chh'}={\chh}[x \mapsto ({\chh}(x)\underset{i+1}{\stackrel{K}{\hookleftarrow}} (\msgg,w_r,p,\{p\}))]}
{(\sigma,{\chh}) \xrightarrow[\proc]{a}
(\sigma',{\chh'})} 
&  \begin{array}{c}{\texttt{(Partial) fulfilment (update)}}\\ {a\in \{\cert,\nor\}, K \in \{SP,FP\}}\end{array}
\\
\rowcolor{blue!5}
\displaystyle\frac{\sigma \xrightarrow[p]{\upd(x,v_r,w_r)} \sigma', ~~ i \geq \ptr(p,{\chh}(x)), ~~ v_r={\chh}(x)[i].value, ~~ {\chh'}={\chh}[x \mapsto ({\chh}(x)\underset{i}{\stackrel{SR}{\hookleftarrow}} (\msgg,w_r,p,\{p\}))]}
{(\sigma,{\chh}) \xrightarrow[\proc]{\cert}
(\sigma',{\chh'})} 
&  \begin{array}{c}{\texttt{Splitting a reservation (update)}}\\ {}\end{array}
\\
\rowcolor{green!10}
\displaystyle\frac{i \geq \ptr(p,{\chh}(x)), ~~ {\chh'}={\chh}[x \mapsto ({\chh}(x)\underset{i}{\stackrel{E}{\hookleftarrow}} (\prm,v,p,\{\}))]}
{(\sigma,{\chh}) \xrightarrow[\proc]{\nor}
(\sigma,{\chh'})} 
&  \begin{array}{c}{\texttt{Promise}}\\\end{array}
\\
\rowcolor{red!5}
\displaystyle\frac{\sigma \xrightarrow[p]{(\fence)} \sigma',   ~~ i_x= max(\ptr(p,{\chh}(x)),\ptr(g,{\chh}(x))),  ~~ {\chh'}={\chh}[x \mapsto move({\chh}(x),p,i_x)]_{x \in \varset}[x \mapsto move({\chh}(x),g,i_x)]_{x \in \varset}}
{(\sigma,{\chh}) \xrightarrow[\proc]{a}
(\sigma',{\chh'})} 
&  \begin{array}{c}{\texttt{SC-fence}}\\ {a \in \{\nor,\cert\}}\end{array}
\\
\end{array}$
};
\end{tikzpicture}
}
\caption{\footnotesize $\lhc$ inference rules at the process level, defining the transition $(\sigma,{\chh})\xrightarrow[\proc]{a} (\sigma',{\chh'})$
where $\proc \in \procset$ and $a \in \{\nor,\cert\}$ is the current mode. $\sigma=(J,R)$ and $\sigma'=(J',R')$ represent 
local process states.}
 \label{ps-program_sem}
\vspace{-0.5cm}
\end{figure}

\smallskip

\noindent{\bf{Losses in $\lhc$.}}
Let $\ch$ and $\ch'$ be two higher order words in  $(\Sigma^* \# (\Sigma \cup \Gamma))^+$. Let us assume that $\ch= u_1 \#a_1  u_2  \# a_2 \dots u_k \# a_k$  
and $\ch'= v_1 \# b_1 v_2  	\# b_2 \dots v_m \# b_m$, with 
$u_i, v_i \in \Sigma^*$ and $a_i, b_j \in \Sigma \cup \Gamma$. We extend the subword relation $\sqsubseteq$ to higher order word as follows: 
$\ch \sqsubseteq \ch'$ iff there is a strictly  increasing function 
$f: \{1, \dots, k\} \rightarrow \{1, \dots, m\}$ s.t.  
$(1)$ $u_i \sqsubseteq v_{f(i)}$ for all $1 \leq i \leq k$, $(2)$ 
$a_i=b_{f(i)}$, and $(3)$ we have the same number of  memory types of the form $(\prm,-,-,-)$ or $(\prm,-,-,-,-)$ in $\ch$ and $\ch'$.
The  relation  $\sqsubseteq$ corresponds to the  loss of  some special empty memory types and redundant simple words (as explained earlier). 
The  relation  $\sqsubseteq$ is extended to mapping from locations to higher order words as follows: ${\chh} \sqsubseteq {\chh'}$ iff  ${\chh}(x) \sqsubseteq {\chh'}(x)$ for all $x \in \varset$.

\smallskip

\noindent{\bf{$\lhc$ states}.}
A $\lhc$ state $\lst$  is a tuple $(({\sf J}, {\sf R}), \chh)$
where
${\sf J} : \procset \mapsto \mathbb{L}$ maps each process $p$ to the label of the next instruction to be executed, 
${\sf R} : \regset \rightarrow \mathsf{Val}$ maps each register to its current value, and $\chh$ is a mapping from locations to higher order words.  The initial $\lhc$ state $\lst_{\rm init}$ is defined as 
$(({\sf J}_{\rm init}, {\sf R}_{\rm init}), \chh_{\rm init})$  where:
(1) ${\sf J}_{\rm init}(p)$ is the label of the  initial   instruction of $\proc$; (2) ${\sf R}_{\rm init}(\reg)=0$ for every register $\reg \in \regset$; and $(3)$ $\chh_{\rm init}(x)=\ch^{\rm init}_x$ for all $x \in \varset$. 

Now we are ready to define the induced transition relation between $\lhc$ states.  For two $\lhc$ states $\lst=(({\sf J}, {\sf R}), \chh)$ and $\lst'=(({\sf J}', {\sf R}'), \chh')$ and $a \in \{\nor,\cert\}$, we write
  $\lst \xrightarrow[p]{a} \lst' $ iff one of the following cases holds: $(1)$ $(({\sf J}(p), {\sf R}), \chh)  \xrightarrow[p]{a} (({\sf J}'(p), {\sf R}'), \chh')$ and ${\sf J}(p')={\sf J}'(p')$ for all $p' \neq p$, or $(2)$ $({\sf J}, {\sf R})=({\sf J}', {\sf R}')$ and ${\chh} \sqsubseteq {\chh'}$.

\smallskip

\noindent{\bf{Two phases $\lhc$ states}}. A two-phases state of $\lhc$ is $\Ss=(\pi, p, \lst_{\nor}, \lst_{\cert})$
where $\pi \in \{\cert, \nor\}$ is a flag describing whether 
the $\lhc$ is in ``standard'' phase or ``certification'' phase, $p$ is the process which evolves in one of these phases, 
while $\lst_{\nor}$, $\lst_{\cert}$ are two $\lhc$  states (one for each phase). 
When the $\lhc$ is in the standard phase, then $\lst_{\nor}$ 
evolves, and when the $\lhc$ is in certification phase, 
$\lst_{\cert}$ evolves. A two-phases $\lhc$ state is said to be initial if it is of the form  $(\nor, p, \lst_{\rm init}, \lst_{\rm init})$, where 
$p \in \procset$ is any process. The transition relation $\rightarrow$
 between two-phases  $\lhc$ states  is defined as follows: Given  $\Ss=(\pi, p, \lst_{\nor}, \lst_{\cert})$ and $\Ss'=(\pi', p', \lst'_{\nor}, \lst'_{\cert})$, we have $\Ss \rightarrow \Ss'$ iff one of the following cases hold:
 
 \begin{itemize}[leftmargin=*]
 \item {\bf During  the standard phase.} $\pi=\pi'=\nor$, $p=p'$, $\lst_{\cert}=\lst'_{\cert}$ and $\lst_{\nor} \xrightarrow[p]{\nor}\lst'_{\nor}$. This corresponds to a simulation of a standard step  of the process $p$.
 
  \item {\bf During the certification phase.} $\pi=\pi'=\cert$, $p=p'$, $\lst_{\nor}=\lst'_{\nor}$ and $\lst_{\cert} \xrightarrow[p]{\cert}\lst'_{\cert}$. This corresponds to a simulation of a certification step  of the process $p$.
 
    \item {\bf From the standard  phase to the certification phase.} $\pi=\nor$, $\pi'=\cert$, $p=p'$, $\lst_{\nor}=\lst'_{\nor}= (({\sf J}, {\sf R}), \chh)$,  and $\lst'_{\cert}$ is of the form $(({\sf J}, {\sf R}), \chh')$ where for every $x \in \varset$, $\chh'(x)=\chh(x) \# (\msgg,v,q,\{\})$ if  $\chh(x)$ is of the form $w \cdot \# (-,v,-,-,q)$ with $q \neq p$, and $\chh'(x)=\chh(x)$ otherwise. This corresponds to the copying of the standard $\lhc$  state to the certification $\lhc$ state in order to check if the set of promises made by the process $p$ can be fulfilled. The  higher order word $\chh'(x)$ (at the beginning of the certification phase) is almost the same as $\chh(x)$ (at  the end of the standard phase) except when the rightmost memory type $(-,v,-,-,q)$ of $\chh(x)$  is tagged by a reservation of a process $q \neq p$. In that case, we append the memory type $ (\msgg,v,q,\{\})$  at the end of $\chh(x)$ to obtain $\chh'(x)$. Note that this is in accordance 
     to the definition of capping memory before going into certification: to cite, (item 2 in capped memory of  \cite{promising2}), a cap message 
     is added for each location unless it is a reservation made by the process going in for certification.  
            It is easy to see that this transition rule can be implemented by a sequence of transitions 
    which copies one symbol at a time, from $\chh$ to $\chh'$.

   \item {\bf From the certification phase to standard phase.} $\pi=\cert$, $\pi'=\nor$,  $\lst_{\nor}=\lst'_{\nor}$,   $\lst_{\cert}=\lst'_{\cert}$, and  $\lst_{\cert}$ is of the form $(({\sf J}, {\sf R}), \chh)$ with $\chh(x)$  does not contain any memory type of the form $(\prm,-,p,-)$/$(\prm,-,p,-,-)$ for all $x \in \varset$ (i.e., all  promises made by  $p$ are  fulfilled).

 \end{itemize}

\smallskip
 
  \noindent{\bf {The   Reachability Problem in $\lhc$}}. 
Given an instruction label function $J: \procset \rightarrow \mathbb{L}$ that maps each  $\proc \in \procset$ to a label in $ \mathbb{L}_{\proc}$, 
the \emph{reachability} problem in $\lhc$ asks
whether there exists a two phases $\lhc$ state $\Ss$ of the form  $(\nor, -, ((J,R), \chh), ((J',R'), \chh'))$ s.t. $(1)$  $\chh(x)$   and $\chh'(x) $ do not contain any memory type of the form $(\prm,-,p,-)$/$(\prm,-,p,-,-)$ for all $x \in \varset$, and $(2)$ $\Ss$ is reachable in $\lhc$  (i.e., $\Ss_0 \rtstep{\xrightarrow[]{}} \Ss'$
where $\Ss_0$ is an initial two-phases $\lhc$ states).
In the case of a positive answer to this problem, we say that $J$ is  reachable in $\prog$ in $\lhc$.

\begin{theorem}
\label{equivalence}
An  instruction label function $J$ is  reachable in a program $\prog$ in $\lhc$ iff  $J$ is   reachable in $\prog$ in $\psr$.
\label{thm:eqv}
\end{theorem}

\subsection{Decidability of $\lhc$ with Bounded Promises}
\label{sec:decproof}
The equivalence of  the reachability 
in $\lhc$ and $\psr$, coupled with Theorem \ref{thm:undec} shows that  reachability is undecidable in $\lhc$. 
To recover decidability, we look at  $\lhc$ with only bounded number of the promise memory type in any higher order word. Let K-$\lhc$ denote $\lhc$ with a number of promises  bounded by $K$.  (Observe that K-$\lhc$ corresponds to $\bps$.) 
 
\begin{theorem}
 The  reachability  problem  is  decidable for K-$\lhc$. 
   \label{decidability-qs}	\end{theorem}
  As a corollary of Theorem \ref{decidability-qs}, the decidability 
  of  reachability follows for $\bps$.  
   The proof makes use of the framework of \emph{Well-Structured Transition Systems} (WSTS) \cite{wsts2,wsts1}, and follows from lemmas \ref{lem:wqo} to \ref{computing-pre}.

\smallskip

\noindent{\bf Well-Structured Transition Systems (WSTS).}
We recall the main ingredients of WSTS.
For more details, the reader is referred to \citet{wsts1,wsts2}.
\smallskip

\noindent{\it Well-quasi Orders.}
Given a (possibly infinite set) $C$,
a quasi-order on $C$ is a reflexive and transitive relation ${\preceq} \subseteq C \times C$.
An infinite sequence $c_1, c_2, \dots$ in $C$ is said to be saturating if there exists
indices $i < j$ s.t. $c_i \preceq c_j$. A quasi-order $\preceq$ is said to be a well-quasi order (wqo)
on $C$ if every infinite sequence in $C$ is saturating.  Given a quasi-order $\preceq$ on $C$, 
the \emph{embedding order} $\sqsubseteq$ on $C^*$ (i.e., the set of finite words over $C$) is defined as
$a_1a_2 \dots a_m \sqsubseteq b_1 b_2 \dots b_n$ if there exists a strictly increasing 
function $g: \{1,2,\dots,m\} \rightarrow \{1, 2, \dots, n\}$ s.t. for all $1 \leq i \leq m$, 
$a_i \preceq b_{g(i)}$. It is well-known  that if 
$\preceq$ is a wqo on $C$, then the embedding order $\sqsubseteq$ is also a wqo on $C^*$ \cite{higman}.

\smallskip\noindent{\it Upward Closure.}
Given a wqo $\preceq$ on a set $C$, a set $U \subseteq C$ is  upward closed if for every $a \in U$ and $b \in C$, with $a \preceq b$,
we have $b \in U$. The upward closure of a set $U \subseteq C$ is 
$\upclos{U} =\{b \in C \mid \exists a \in U, a \preceq b\}$.
It is known that every upward closed set $U$ can be characterized by a finite \emph{minor}.
A minor $M \subseteq U$ is s.t.\
(i) for each $a \in U$, there is a $b \in M$ s.t.\ $b \preceq a$, and
(ii) for all $a, b \in M$ s.t. $a \preceq b$, we have $a=b$.
For an upward closed set $U$, let $\mini$ be the function that returns the minor of $U$.

\smallskip\noindent{\it Well-Structured Transition Systems (WSTS)}. 
Let $\Tt$ be a transition system with (possibly infinite) set of states $C$, initial states $C_{\init}$ and  transition relation $\rightsquigarrow \subseteq C \times C$. 
 Let $\preceq$ be a well-quasi ordering on $C$.
We define the set of predecessors of a subset $U \subseteq C$ of states as ${\mathtt{Pre}}(U)=\{c \in C \mid \exists c' \in U.\; c \rightsquigarrow c'\}$.
For a state $c$, we denote the set $\mathtt{min}(\mathtt{Pre}(\upclos{\{c\}}) \cup \upclos{\{c\}})$ as $\mathtt{minpre}(c)$. 
$\Tt$  is called well-structured if $\rightsquigarrow$ is \emph{monotonic} w.r.t.\ $\preceq$ : that is, given  $c_1, c_2$ and $c_3$ in $C$, if $c_1 \rightsquigarrow c_2$ and $c_1 \preceq c_3$, then there exists a state $c_4$ s.t. $c_3 \stackrel{*}{\rightsquigarrow} c_4$ and $c_2 \preceq c_4$. 

Given a finite set of states $C_{\tar} \subseteq C$,
the \emph{coverability} problem asks if there is a state $c' \in  \upclos{C_{\tar}}$ reachable in $\Tt$.
The following conditions are sufficient for the decidability of this problem:
(i) for every two states $c_1, c_2 \in C$, it is decidable if $c_1 \preceq c_2$,
(ii) for every $c \in C$, we can check if $\upclos{\{c\}} \cap C_{\init} \neq \emptyset$, and
(iii) for each $c \in C$, the set $\tt{minpre}(c)$ is finite and computable.

The algorithm for checking WSTS coverability is based on a backward analysis.
The sequence $(U_i)_{i \geq 0}$ with
$U_0= \min(C_{\tar})$ and
$U_{i+1}=\min({\tt{Pre}}(\upclos{U_i}) \cup \upclos{U_i})$
reaches a fixpoint and is computable \cite{wsts2,wsts1}.

\smallskip

\noindent{\bf{$\lhc$ with bounded promises is a WSTS}.}
We will show that the K-$\lhc$ transition system 
is a well-structured transition system. 
Let $C$ denote the set of two-phases   K-$\lhc$ states of $\prog$. Given an instruction label function $J: \procset \rightarrow \mathbb{L}$, let $C_{\tar}$ be a finite subset of $C$   of the form  $(\nor, -, ((J,R), \chh), ((J',R'), \chh'))$ such that for every $x \in \varset$, we have: $(1)$  $\chh(x)$ and   $\chh'(x)$  do not contain any memory type of the form $(\prm,-,p,-)$/$(\prm,-,p,-,-)$, and $(2)$ $|\chh(x)|, |\chh'(x)|  \leq |\procset|$. We define the well-quasi ordering $\sqsubseteq$ on $C$ in a way that the upward closure of $C_{\tar}$ consists of all two-phases K-$\lhc$  states of the form $(\nor, -, ((J,R), \chh), ((J',R'), \chh'))$ such that for every $x \in \varset$,   $\chh(x)$ and   $\chh'(x)$  do not contain any memory type of the form $(\prm,-,p,-)$/$(\prm,-,p,-,-)$.
Then,  the coverability of $C_{\tar}$ 
 is equivalent to the reachability of $J$  in K-$\lhc$.

In the following, we define the well-quasi ordering $\sqsubseteq$ on on $C$ (Lemma \ref{lem:wqo}). Then, we show the monotonicity of the K-$\lhc$ transition relation $\rightarrow$ w.r.t.\  $\sqsubseteq$ (Lemma \ref{lem:mon}). Finally, we show how to compute the set of predecessors of a given two-phases $K$-{$\lhc$} state (Lemma \ref{computing-pre}).  Observe that the first and second sufficient conditions for the decidability of the coverability problem, namely comparing two states and checking whether an upward closure set contains the initial state, are trivial (the second condition can be reduced whether a minimal state is equal to the initial state).

The ordering $\sqsubseteq$ defined on mapping from locations to higher order words 
can be extended to two phases  K-$\lhc$ states by component wise extension: 
$(\pi, p, ((J_1,R_1), \chh_1), ((J_2,R_2), \chh_2)) \sqsubseteq (\pi', p', ((J'_1,R'_1), \chh'_1), ((J'_2,R'_2), \chh'_2))$  holds iff   
$\pi'=\pi$, $p'=p$, $(J_1,R_1)=(J'_1,R'_1)$, $(J_2,R_2)=(J'_2,R'_2)$, $\chh_1 \sqsubseteq \chh'_1$, and  $\chh_2 \sqsubseteq \chh'_2$. Since  the embedded ordering $\sqsubseteq$  is a wqo on higher order words when   the number of   promises is bounded \cite{higman}, we obtain the following lemma. 
 \begin{lemma}
The relation $\sqsubseteq$ 
is a well-quasi ordering on the two phases K-$\lhc$ states.
\label{lem:wqo}
\end{lemma}

Consider now a   two-phases K-$\lhc$  state $\Ss$ of the form $(\nor, -, ((J,R), \chh), ((J',R'), \chh'))$ such that for every $x \in \varset$,   $\chh(x)$ and   $\chh'(x)$  do not contain any memory type of the form $(\prm,-,p,-)$/$(\prm,-,p,-,-)$, then it is easy to see that $\Ss  \in \upclos{C_{\tar}}$.
This implies that:

\begin{lemma}
The coverability of $C_{\tar}$ is equivalent
to the reachability of   $J$ in K-$\lhc$.
\label{lem:cover-equivalence}
\end{lemma}

\noindent{\bf{Monotonicity}}. 
The following lemma shows the monotonicity of the K-$\lhc$ transition relation $\rightarrow$ w.r.t.\ $\sqsubseteq$.
This allows the backward algorithm for coverability to work with only upward closed sets,
since the set of predecessors of an upward closed set is also upward closed \cite{wsts2,wsts1}.

\begin{lemma}
The transition relation $\rightarrow$ is monotonic w.r.t.\ $\sqsubseteq$.
\label{lem:mon}
\end{lemma}

\noindent
{\bf Computing the set of predecessors.}
The last sufficient condition for the decidability of the coverability problem in $K$-$\lhc$ is stated by the following lemma
\begin{lemma}
\label{computing-pre}
For each two-phases K-$\lhc$ state $c$, the set $\tt{minpre}(c)$ is effectively computable.
\end{lemma}

Next, we state  that  the reachability problem for K-$\lhc$ (even for $K=0$) is highly non-trivial (i.e., non-primitive recursive). The proof is done by reduction from the reachability problem for lossy channel systems, in a similar to the case of TSO \cite{ABBM10} where we insert $\fence$ instructions everywhere in the process that simulates the lossy channel process (in order to ensure that no promises can be made by that process). 
 
The proof is done by reduction from the reachability problem for lossy channel systems (LCS). We construct a 
concurrent program with 2 processes, the first process $p_1$ keeps track of the finite state control of the LCS, while the 
second process $p_2$ simulates the lossy channel. Two shared variables $x_c, y_c$ are used to simulate the lossy channel $c$. 
$p_1$ writes to $x_c$ on each transition that writes to $c$ in the LCS. $p_2$ reads from $x_c$ and writes to $y_c$. 
A read from the channel $c$ in the LCS is simulated by $p_1$ reading from $y_c$, thereby simulating the lossiness of $c$ 
($p_2$ can skip some messages of $x_c$, and $p_1$ can also skip some messages of $y_c$). Every two instructions 
of $p_1, p_2$ have a $\fence$ to ensure no promises can be made (and fulfilled). 

\begin{theorem}
The reachability problem for K-$\lhc$  is non-primitive recursive.
\label{thm:npr}
\end{theorem}

\section{Source to Source Translation}
\label{sec:c2c}
We consider a 
parametric under-approximation in the spirit of context bounding \cite{cb2}, \cite{DBLP:conf/cav/TorreMP09}, \cite{DBLP:journals/fmsd/LalR09}, \cite{demsky}, \cite{MQ07}, \cite{DBLP:conf/tacas/QadeerR05}, \cite{pldi2019}, \cite{cb3}. The bounding concept chosen for concurrent programs depends on aspects related to the 
interactions between the processes. In the case of SC programs, context bounding has been shown experimentally to have extensive behaviour coverage for bug detection \cite{MQ07}, \cite{DBLP:conf/tacas/QadeerR05}. 
A context in the SC setting is a computation segment where only one process is active. The concept of context bounding has been extended for  weak memory models. For instance, in TSO, the notion of  context is extended to one where all updates to the main memory are done only from the buffer of  the active thread \cite{cb2}. 
 In the case of POWER \cite{cb3},  context  was extended to consider propagation actions performed by the active process. In the case of $\ps$-$\ra$ without promises and reservations \cite{pldi2019}, context bounding was extended to view bounding, using the notion of view switching messages. The notion of bounding appropriate for a model depends on its underlying  complexity. From a theoretical point of view, we have already seen that $\ps$ is very complex, and bounding contexts  is not sufficient. Our bounding notion for $\ps$ is based on its various features which includes relaxed as well as RA memory accesses, promises and certification. Since $\ps$ subsumes RA, we recall 
 the bounding notion used in RA first, using \emph{view altering} messages.  

 \noindent{\textit{View Altering Reads}}. A read from the memory
is view altering if it changes the view of the process reading it.The message which is reads from in turn is called a view altering message. The under approximate analysis 
  for RA \cite{pldi2019} considered view bounded runs, where the number of view altering reads is bounded. 

\noindent{\textit{Essential Events}}. An essential event in a run $\rho$ of a concurrent program  
under $\ps$ is either a promise, a reservation or a
 view altering read by some process in the run.

\noindent{\textit{Bounded Context}}. A context is an uninterrupted sequence of actions by a single process. In a run having $K$ contexts, 
the execution switches from one process to another $K-1$ times. A $K$ bounded context run is one where the number of context switches are bounded by $K \in \mathbb{N}$. The $K$ bounded context reachability problem in SC checks for the existence 
of a $K$ bounded context run reaching some chosen instruction. A SC program is called a $K$ bounded context program if all 
its runs are $K$ bounded context.  
Now we define the notion of bounding for $\ps$. 
    
    \noindent{\bf {The Bounded Consistent Reachability Problem}}. 
Consider a run $\rho$ of a concurrent program under $\ps$,   $\MS_{0} \rtstep{\xrightarrow[p_{i_1}]{}} 
   \MS_{1} \rtstep{\xrightarrow[p_{i_2}]{}} 
   \MS_{2} \rtstep{\xrightarrow[p_{i_3}]{}} 
   \ldots  
   \rtstep{\xrightarrow[p_{i_n}]{}}
   \MS_{n}$.  
       A run $\rho$ of a concurrent program $\prog$ under $\ps$ is called \emph{$K$ bounded} iff the 
   number of essential events in $\rho$ is $\leq K$. 
      The $K$ bounded reachability problem for $\ps$ checks for the existence 
      of a run $\rho$ of $\prog$ which 
   is $K$-bounded. Assuming $\prog$ has $n$ processes,  
 we  propose an algorithm that reduces the $K$ bounded reachability problem 
 to a $K+n$ bounded context reachability problem under SC.

\noindent{\bf{Translation Overview}}. 
Let $\prog$ be a concurrent program under $\ps$  with set of processes $\procset$ and locations $\varset$. 
Our algorithm relies on a source to source translation of $\prog$ to a bounded context SC program $\sem{\prog}$, as shown in Figure \ref{transl} and operates on the same data domain.   The translation 
 adds a new process (\textsc{Main}) that initializes the global variables of $\sem{\prog }$. 
 The translation of a process $\proc\in\procset$ adds local variables, which are initialized by the function $\textsc{InitProc}$. 

\begin{figure}
\small
\tikzset{background rectangle/.style={fill=black!5,rounded corners,draw=black}}
 \resizebox{4\textwidth/5}{!}{
 \begin{tikzpicture}[codeblock/.style={line width=0.5pt, inner xsep=0pt, inner ysep=0pt}, show background rectangle]
\node[codeblock] (init) at (current bounding box.north west) {
{
$\arraycolsep=0.7pt\def\arraystretch{1.4}
\begin{array}{rl}
\sem{Prog} &\coloneqq (\langle \text{global vars}\rangle; \langle \textsc{Main} \rangle ; (\sem{\texttt{proc } p \texttt{ reg } \reg^* i^*})^* \\
\sem{ \texttt{proc } p \texttt{ reg } \reg^*\; i^* } &\coloneqq \texttt{proc } p \texttt{ reg } \reg^*
 \langle \text{local vars} \rangle \langle\textsc{InitProc}\rangle
 \langle\textsc{CSO}\rangle^{p, \lambda_0} (\sem{ i}^p)^* \\
\sem{ \lambda\space:\space i }^p &\coloneqq
 \lambda\space: \langle\textsc{CSI}\rangle; \sem{s}^p; \langle\textsc{CSO}\rangle^{p, \lambda} \\
\sem{ \kwif\ \mathit{exp} \ \kwthen\ i^* \ \kwelse\ i^* }^p &\coloneqq \kwif\ \mathit{exp} \ \kwthen\ (\sem{ i}^p)^* \ \kwelse (\sem{ i}^p)^*\\
\sem{ \kwwhile\ \mathit{exp} \ \kwdo\ i^* }^p &\coloneqq \kwwhile\ \mathit{exp}\ \kwdo\ (\sem{ i}^p)^* \\
\sem{ \assume(\mathit{exp})}^p &\coloneqq \assume(\mathit{exp}) \\
\sem{ \reg = \mathit{exp}}^p &\coloneqq \reg = \mathit{exp} \\
\sem{ x = \reg }^{p}_{o \in \{\rlx, \ra\}}  &\coloneqq \text{ see write Pseudocode }\\
\sem{ \reg = x }^{p}_{o \in \{\rlx, \ra\}} &\coloneqq \text{ see read Pseudocode }  \\
\end{array}$
}
};
\end{tikzpicture}
}
\caption{Source-to-source translation map}
\label{transl}
\end{figure}

 This is followed by the code block 
$\langle CSO \rangle^{p, \lambda_0}$ (Context Switch Out) 
that optionally enables the process to switch out of context. For each instruction $i$ appearing in the code of $p$, the map $\sem{i}^p$ 
transforms it into a sequence of instructions as follows : the code block 
$\langle CSI \rangle$ (Context Switch In) checks if the process is active in the current context; then it transforms each statement $s$ of instruction $i$ 
into a sequence of instructions following the map $\sem{s}^p$, and finally executes the code block $\langle CSO \rangle^{p, \lambda}$. $\langle CSO \rangle^{p, \lambda}$ facilitates two things: when the process is at an instruction label $\lambda$, 
(1) 
 allows $p$ to make promises/reservations after $\lambda$, s.t. the control is back at $\lambda$ after certification;
 (2) it ensures that the machine state is consistent when $p$ switches out of context.
Translation of $\assume$, $\kwif$ and $\kwwhile$ statements 
keep the same statement. Translation of read and write statements are described later. Translation of RMW statements are omitted for ease of presentation.

\tikzstyle{startstop} = [rectangle, rounded corners, minimum width=1cm, minimum height=0.5cm,text centered, draw=black, fill=blue!30]
\tikzstyle{processN} = [rectangle, minimum width=1cm, minimum height=0.5cm, text centered, draw=black, fill=violet!30]
\tikzstyle{processCC} = [rectangle, minimum width=1cm, minimum height=0.3cm, text centered, draw=black, fill=orange!30]
\tikzstyle{goal} = [rectangle, minimum width=1cm, minimum height=0.3cm, text centered, draw=black, fill=red!30, text width=2cm]
\tikzstyle{arrow} = [thick,->,>=stealth]
\tikzset{
   block filldraw/.style={
       dotted, thick, fill=blue!20,fill opacity=0.25,  draw=blue}
}
\begin{figure}[h]
\centering
\begin{tikzpicture}[node distance=2cm]
\node (start) [startstop] {init};
\node (p1n) [processN, right=0.3cm of start] {$p_1$ \texttt{n}};
\node (p1cc) [processCC, right of=p1n] {$p_1$ \texttt{cc}};
\node (empty) [right of=p1cc] {$\cdots$};
\node (pjn) [processN, right=1cm of empty] {$p_{j-1}$ \texttt{n}};
\node (pjcc) [processCC, right=1cm of pjn] {$p_{j-1}$ \texttt{cc}};
\node (pkn) [goal, right=1cm of pjcc] {$p_j$ \texttt{n} {\scriptsize \texttt{ASSERT(false)}}};
\path[->]
	(start) edge (p1n) 
	(p1n) edge node [above] {{\scriptsize\textsc{CSO}}$^{p_1}$} (p1cc) 
	(p1cc) edge  node [above] {{\scriptsize\textsc{CSO}}$^{p_1}$} (empty)
	(empty) edge node [above] {{\scriptsize\textsc{CSO}}$^{p_{j-2}}$} (pjn)
	(pjn) edge node [above] {{\scriptsize\textsc{CSO}}$^{p_{j-1}}$} (pjcc)
	(pjcc) edge node [above] {{\scriptsize\textsc{CSO}}$^{p_{j-1}}$} (pkn);
\draw[block filldraw] ([xshift=-5pt, yshift=-18pt] p1n.south west) rectangle ([xshift=27pt, yshift=10pt] p1cc.north east) ;
\draw [
   thick,
   decoration={
       brace,
       raise=0.1cm
   },
   decorate
] (p1cc.south east) -- (p1cc.south west); 
\node [below= 0.2cm of p1cc] {\scriptsize$\leq\texttt{certDepth}$}; 
\node [above left=0.27cm and 0.05cm of p1cc] {\scriptsize {\color{blue} one context}};
\end{tikzpicture}
\caption{Control flow: In each context, a process runs first in normal mode \texttt{n} and then in consistency check mode \texttt{cc}. The transitions between these modes is facilitated by the \textsc{CSO} code block of the respective process. We check for assertion failures for $K+n$ context-bounded executions ($j\leq K+n$).}
\label{fig:scrun}
\end{figure}
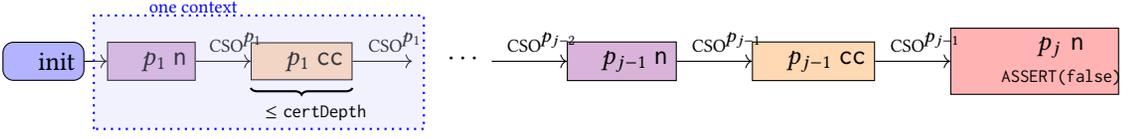

The set of promises a process makes has to be constrained   with respect to the set of promises that it can certify, since  processes can generate arbitrarily many promises/reservations, while, in reality only a few of them will be certifiable. 
To address this, in the translation, processes run in two modes : 
a `normal' mode and a `check' (\textit{consistency check}) mode. In the normal mode, a process does not make any promises or reservations. In the check mode, the process may make promises and reservations and 
 subsequently certify them before switching out of context. In any context, a process first enters the normal mode, and then, before exiting the context it enters the check mode. The  check mode is used by the process to (1) make new promises/reservations and (2) certify consistency of the machine state.
 We also add an optional parameter, called \textit{certification depth} (\texttt{certDepth}), which constrains the number of steps a process may take in the check mode to certify its promises. 
 Figure \ref{fig:scrun} shows the structure of a translated run under SC.

To reduce the $\ps$ run into a bounded context SC run, we use the bound on the number of essential events. 
From the run $\rho$ in $\ps$, we construct a $K$ bounded 
run $\rho'$ in $\ps$ where the processes run in the order of generation 
of essential events. So, the process which generates the first essential event is run first, till that event happens, 
then the second process   which generates the second essential event 
is run, and so on. This continues till $K+n$ contexts :  the $K$ bounds the number of essential events, and  the $n$ is to ensure all processes are run to completion.  
The bound on the number of essential events gives a bound on the number of 
timestamps that need to be maintained. 
As observed in \cite{pldi2019}, one view altering read requires two timestamps; additionally, each promise/reservation requires one timestamp. Since we have $K$ such essential events, $2K$ time stamps suffice. We choose $\mathsf{Time}=\{0,1,2, \dots, 2K\}$ as the set of timestamps.

\noindent {\bf{Data Structures}}.
 We mention the significant ones. 
The \gmess data structure represents a message generated as a write or a promise and has 4 fields (i) $\mathit{var}$, the address of the memory location written to;
(ii) the timestamp $t$ in  the view associated with the message;
(iii) $v$, the value written; and
(iv) $\mathit{flag}$, that keeps track of whether it is a message or a promise; and, in case of a promise, which process it belongs to.
The \goldf{\textsf{View}} data structure stores, for each memory location $x$, (i) a timestamp $t \in \mathsf{Time}$,
(ii) a value $\mathit{v}$  written to $x$,
(iii) a Boolean $l \in \{\true, \false\}$ representing whether $t$ is an exact timestamp (which can be used for essential events) or an abstract timestamp (which corresponds to non-essential events). 

\noindent{\bf {Global Variables}}. 
\label{para:globvars}
The \gmem is an array of size $K$ holding elements 
of type \gmess. This array  
is populated with the view switching messages, promises and reservations generated by the program. We maintain counters for 
(1) 
the number of elements  in \gmem;
(2) 
the number of context switches that have occurred; and
(3) 
the number of essential events that have occurred. 

\noindent{\bf{Local Variables}}. 
In addition to its local registers, each process has local variables including
\begin{itemize}
\item a local variable $\gview$, which stores  a local instance of the view function (this is of type 
  \goldf{\textsf{View}}),
\item
 $\mathit{active}$: a boolean variable which is set when the
process is running in the current context, and
 \item $\mathit{checkMode}$: a boolean  denoting whether the process is in the certification phase. We implement the certification phase as a function call, and hence store the process state and return address, while entering it. 
\end{itemize}

\noindent {\bf{Subroutines}}.  We use certain helper subroutines as follows:
\begin{itemize}
\item genMessage is a subroutine which generates an instance of the  \gmess data structure;  
\item saveState($p$) is a subroutine  which saves the values of the global variables and the local states (instruction labels and local variables)  of process $p$. This is used when switching into check mode.
\item loadState($p$) is a subroutine which loads the 
the values of global variables and local states of $p$  which was saved using saveState($p$). This is use when switching out of check mode.
\end{itemize}

\subsection{Translation Maps}

  \begin{wrapfigure}{r}{0.45\textwidth}       
\footnotesize
      \begin{algorithm}[H]               
        \DontPrintSemicolon
        \SetCustomAlgoRuledWidth{0.45\textwidth}   
        \caption{$\mathsf{CSO}$}
  \tcc{nondeterministically enter check mode and exit context}
  \If{nondet()}{ 
    \uIf{$\neg$checkMode}{
      \tcc{enter consistency check}
      \If{not in context}{
        enter context \;
      }
      checkMode $\leftarrow \texttt{true}$ \;
      save localstate \;
      returnAddr $\leftarrow \lambda$ \;
    }
    \Else{
      \tcc{consistency check successful!}
      ensure all Promises for process are certified \;
      \tcc{for next context}
      mark all Promises as uncertified \;
      $\mathit{checkMode} \leftarrow \texttt{false}$ \;
      load localstate \; 
      goto $\mathit{returnAddr}$ \;
      exit context \;
    }
  }
  \label{alg:cso}
  \end{algorithm}
\end{wrapfigure}

In what follows we illustrate how the translation simulates a run under $\ps$. At the outset, recall that each process alternates, in its execution, between two modes: a \emph{normal} mode (\texttt{n}  in Figure \ref{fig:scrun}) at the beginning of each context and
the \emph{check} mode 
at the end of the current context (\texttt{cc} in Figure \ref{fig:scrun}), where it may make new promises and certify them before switching out of context. 

\noindent \textbf{Context Switch Out ($CSO^{p, \lambda}$).}
We describe the \textsc{CSO} module  (Algorithm \ref{alg:cso} provides its pseudocode). \textsc{CSO}$^{p, \lambda}$ is placed after each instruction $\lambda$ in the original program and serves as an entry and exit point for the consistency check phase of the process.  When in normal mode (\texttt{n})  after some instruction $\lambda$, \textsc{CSO} non-deterministically guesses whether the process should exit the context at this point, and  sets the \textit{checkMode} flag to true and subsequently, saves its local state and  the return address (to mark where to resume execution from, in the next context). 
 The process then continues its execution in the consistency check mode 
 (\texttt{cc}) from the current instruction label ($\lambda$) itself. Now the process may generate new promises (see Algorithm \ref{alg:write}) and certify these as well as earlier made promises. In order to conclude the check mode phase, the process will enter the \textsc{CSO} block at some different instruction label $\lambda'$. Now since the \textit{checkMode} flag is true, the process enters the else branch, verifies that there are no outstanding promises of $p$ to be certified. Since the promises are not yet fulfilled, when $p$ switches out of context, it has to mark all its promises uncertified. 

 When the context is back to $p$ again, this will be used 
 to fulfil the promises or to certify them again before the context switches out of $p$ again.  
  Then it exits the check mode phase, setting \textit{checkMode} to false. Finally it loads the saved state, and returns to the instruction label $\lambda$ (where it entered check mode) and exits the context.

\noindent \textbf{Write Statements}.
We now discuss the translation of a write instruction $\llbracket x\coloneqq\reg\rrbracket_o$, where $o \in\{\rlx,\ra\}$ of a process $\proc$, the intuitive pseudocode for which is given in Algorithm \ref{alg:write}.

  This is the general psuedo code for both kinds of memory accesses, with specific details 
pertaining to the particular access mode omitted.

Let us first consider execution in the normal mode (i.e., $\mathit{checkMode}$ is false).
First, the process updates its local state  with the value that it will write.
Then, the process non-deterministically chooses one of three possibilities for the write, it either
(i) does not assign a fresh timestamp (non-essential event),
(ii) assigns a fresh timestamp and adds it to memory, or
(iii) fulfils some outstanding promise.

\begin{wrapfigure}{r}{0.45\textwidth}     \footnotesize  
      \begin{algorithm}[H]               
        \DontPrintSemicolon
        \SetCustomAlgoRuledWidth{0.45\textwidth}   
        \caption{$\texttt{Write}$}
           update localstate with write \;
  \uIf(\tcc*[f]{(i) no fresh timestamp}){nondet()}{ 
    \uIf{checkMode}{
      \tcc{since write is not a promise}
      certify message with reservation or splitting
    }
  }
  \uElseIf(\tcc*[f]{(ii) fresh timestamp}){nondet()}{
    generate a view; generate a message \;
    \eIf{checkMode}{
      insert message into Memory as Promise and certify \;
    }{
      insert message into Memory as concrete message \;
    }
  }
  \Else(\tcc*[f]{(iii) fulfill old promise}){
    get Promise from Memory \;
    check variable, value and view match \;
    \eIf{checkMode}{
      mark message as certified \;
    }{
      mark message as fulfilled \;
    }
    replace message into Memory \;
  }
  \label{alg:write}
  \end{algorithm}
\end{wrapfigure}

Let us now consider a write executing  when $\mathit{checkMode}$ is true, and highlight differences with the normal mode. 
In case (i),  non essential events exclude promises and reservations. 
Then, while in certification phase, since we use a capped memory, 
the process can make a write if either (1) the write interval can be generated through splitting insertion or (2) the write can be certified with the help of a reservation.  
Basically the writes we make either split an existing interval (and add this to the left of a promise), or forms a part of a reservation. 
\setlength\intextsep{0pt}

\begin{figure}{r}
       \footnotesize
      \begin{algorithm}[H]               
        \DontPrintSemicolon
        \SetCustomAlgoRuledWidth{0.45\textwidth}   
        \caption{$\texttt{Read}$}
  \uIf(\tcc*[f]{local read}){nondet()}{ 
    check local state is valid \;
    update local state with read \;
  }
  \Else(\tcc*[f]{nonlocal (view-switching) read}){
    check that local state allows read \;
    get message from Memory \;
    check variable, value, view  are allowed \;
    update local state with message view \;
  }
  \label{alg:read}      
  \end{algorithm}
\end{figure}

Thus, 
the time stamp of a neighbour is used. 
In case (ii) when a fresh time stamp is used, the write is made as a promise, and then certified before switching out of context. 
The analogue of case (iii)  is the certification of promises for the current context; promise fulfilment happens only in the normal mode.  
To help a process decide the value of a promise,  
we use the fact that CBMC allows us to assign a non-deterministic value of a variable. On top of that, we have implemented an optimization that checks the set of possible values to be written in the future.

\noindent \textbf{Read Statements.}
The translation of a read instruction $\llbracket\reg\coloneqq x\rrbracket_o$, $o \in \{\rlx,\ra\}$ of process $\proc$ is given in Algorithm \ref{alg:read}. 
The process first guesses, whether it will read from a view altering message in the memory of from its local view. If it is the latter, the process must first verify  whether it can read from the local view ; 
for instance, reading from the local view may not be possible after execution of a  \texttt{fence} instruction when the timestamp of a variable $x$ gets incremented from the local view $t$ to $t' > t$.  In the case of a view altering read, we first check that we have not reached the context switching/essential event bound. Then the new \gmess is fetched from \gmem and we check the view (timestamps) in the acquired \gmess satisfy the conditions imposed by the access type $\in \{\ra, \rlx\}$. Finally, the process updates its view with that of the new message and increments the counters for the context switches and the essential events. Theorem \ref{thm:s2s} proves the correctness  
of our translation. 

\begin{theorem}
Given a program $\prog$ under $\ps$, and $K \in \mathbb{N}$, the source to source translation 
constructs a program $\sem{prog}$ whose size is polynomial in $\prog$ and $K$ such that,  for every $K$-bounded run of $\prog$ under $\ps$ reaching a set of instruction labels, there is a $K+n$-bounded context run of $\sem{prog}$ 
under SC that reaches the same set of instruction labels.   	
	\label{thm:s2s}
\end{theorem}
\section{Implementation and Experimental Results}
\label{sec:eval}
In order to check the efficiency of the source-to-source translation, we implement a prototype tool, \tool{} which is the first tool to handle \ps. \tool~ takes as input a C program and a bound $K$ and translates it to a program $\mathit{Prog}'$ to be run under SC. We use CBMC version 5.10 as  backend to verify $\mathit{Prog}'$. CBMC takes as input $L$, the loop unrolling parameter for bounded model checking of $\mathit{Prog}'$. 
We supply the bound on \textit{Essential Events}, $K$ as a parameter to \tool. \tool~ then considers the subset of executions respecting the bounds $K$ and $L$ provided as input. If it returns \textit{unsafe}, then the program has an unsafe execution. Conversely, if it returns \textit{safe} then none of the executions within the subset violate any assertion. $K$ may be iteratively incremented to increase the number of executions explored. We provide a functionality with which the user optionally selects a subset of processes for which promises and reservations will be enabled. While in the extreme cases we can run \tool~ in the promise-full (all processes can promise) and promise-free modes, \textit{partial promises} (allowing subsets of processes to promise) turns out to be an effective  technique.

We now report the results of experiments we have performed with \tool. We have two objectives: (1) studying the performance of \tool~ on benchmarks which are unsafe only if promises are enabled and (2) comparing \tool~ with other model checkers when operating in the promise-free mode (since they can not handle promises).
In the first case, we show that \tool~ is able to uncover bugs in examples with low interaction (reads and writes) with the shared memory. When this interaction increases, however, \tool~ does not scale, owing to the huge non-determinism in  \ps. However, with partial promises, \tool~ is once again able to uncover bugs in reasonable amounts of time. 
In the second case, our observations highlight the ability to detect hard to find bugs with small $K$ for unsafe benchmarks, and scalability by altering $K$ as discussed earlier in case of safe benchmarks. We compare \tool with three state-of-the-art stateless model checking tools, $\cdsc$ \cite{cdsc}, $\genmc$ \cite{genmc} and $\rcmc$ \cite{rcmc} that support the promise-free subset of the \ps semantics. 
In the tables that follow we 
provide the value of $K$ (for \tool{} only) and the value of $L$ (for all tools).
We do not consider compilation time for any tool while reporting the results. For \tool, the time reported is the time taken by the CBMC backend for analysis. The timeout used is 1 hour for all benchmarks. All experiments are conducted on a machine  with a 3.00 GHz Intel Core i5-3330 CPU and 8GB RAM running a Ubuntu 16 64-bit operating system. We denote timeout by `TO', and memory limit exceeded by `MLE'.
\subsection{Experimenting with Promises}
In this section we check the efficiency of the source-to-source translation  in handling  promises for \ps (which is the most difficult part due to the non-determinism).

 We first test \tool{} on litmus-tests adapted from \cite{promising,promising2,thinair,jmm}. These examples are small programs that serve as barebones thin-air tests for the C11 memory model. Consistency tests based on the Java Memory Model are proposed in \cite{jmm}. These were also experimented on in \cite{mrder} with the MRDer tool. Like MRDer, \tool{} is able to verify most of these tests within 1 minute which shows its ability to handle typical programming idioms of \ps.

\setlength\intextsep{-6pt}
\captionsetup[table]{font=scriptsize}
\begin{table}[t]
\vspace{-0.3cm}
\scriptsize
\centering
\resizebox{0.25\columnwidth}{!}{
\begin{tabular}{cccc}
\hline
\textbf{testcase} & $K$ & \textbf{\tool}
\\
\hline\hline
ARM\_weak  & 4 & 0.765s  \\ 
Upd-Stuck & 4 & 1.252s \\
split & 4 &  25.737s  \\ \hline
LBd & 3 & 1.481s \\
LBfd & 3 & 1.512s \\ \hline
CYC & 5 & 1.967s \\
Coh-CYC & 5 & 42.67s \\ \hline\hline
Pugh2  &  3 & 13.725s  \\ 
Pugh3  & 3 & 12.920s  \\ 
Pugh8  & 3 & 1.67s  \\ \hline
Pugh5  & 5 & 4.811s  \\ 
Pugh10 & 5 & 3.868s \\
Pugh13 & 5 & 3.345s \\ \hline
\end{tabular}}
\caption{Litmus Tests}
\label{tab:litmus}
\end{table}

\begin{table}[t]
\small
\centering
\resizebox{0.27\columnwidth}{!}{
\begin{tabular}{cccc}
\hline
\textbf{testcase} & $K$ & \textbf{\tool}
\\
\hline\hline
fib\_local\_3  & 4 & 0.742s  \\ 
fib\_local\_4  & 4 & 0.761s \\ \hline
fib\_local\_cas\_3 & 4 & 1.132s  \\ 
fib\_local\_cas\_4 & 4 &  1.147s  \\ \hline
\end{tabular}}
\caption{Performance of \tool~ on cases with local computation}
\label{tab:prom1}
\end{table}
\begin{table}[t]
\resizebox{0.3\columnwidth}{!}{
\begin{tabular}{cccc}
\hline
\textbf{testcase} & $K$ & \textbf{\tool}[1p]
\\
\hline\hline
fib\_global\_2  & 4 & 55.972s \\
fib\_global\_3  & 4 & 2m4s  \\ 
fib\_global\_4 & 4 &  4m20s  \\ \hline
exp\_global\_1  & 4 & 19m37s \\
exp\_global\_2  & 4 & 41m12s  \\  \hline
\end{tabular}}
\caption{Performance of \tool~ on cases with global computation}
\label{tab:prom2}
\end{table}

In Table \ref{tab:prom1} we consider unsafe examples in which a process is required to generate a promise (speculative write) with value as the $i^{\mathit{th}}$ fibonacci number (\texttt{Fibonacci}-based benchmarks for SV-COMP 2019 \cite{beyer2019automatic}).  This promise is certified using computations local to the process. Thus though the parameter $i$ increases the interaction of the promising process with the memory remains constant. The $\cas$ variant requires the process to make use of reservations. We note that \tool~ uncovers the bugs effectively in all these cases.

Now we consider the case where promises  require some interaction between processes. We consider an example  adapted from the \texttt{Fibonacci}-based benchmarks for SV-COMP 2019 \cite{beyer2019automatic}, where two processes compute the $i^{\text{th}}$ fibonacci number in a distributed fashion. Unlike the previous case, here, the amount of interaction increases with $i$. Here however, our tool times out.

\textit{How do we recover tractable analysis in this case?} We tackle this problem by a modular approach of allowing partial-promises, i.e. subsets of processes are allowed to generate promises/reservations. In the experiments, we allowed only a single process to do so. The results obtained are in Table \ref{tab:prom2}, where \tool[1p] denotes that only one process is permitted to perform promises. We then repeat our experiments on two other unsafe benchmarks - \texttt{ExponentialBug} from Fig. 2 of \cite{huang}  and have similar observations. With this modular approach \tool~ uncovers the bug. To summarize,  we note that the source to source approach performs well on programs requiring limited global memory interaction. When this interaction increases, \tool~ times out, owing to the huge non-determinism of \ps. However, the modular approach of partial-promises enables us to recover effectiveness.

\subsection{Comparing Performance with Other Tools}

In this section we compare performance of \tool~ in promise-free mode with $\cdsc$ (\cite{cdsc}), $\genmc$ (\cite{genmc}) and $\rcmc$ (\cite{rcmc}) on safe and unsafe benchmarks. We provide a subset of the experimental results, the remaining can be found in the full version. The results of this section indicate that the source-to-source translation with essential event bounding is effective at uncovering hard to find bugs in non-trivial programs. We will observe that in most examples discussed below, we had $K \leq 10$. Additionally, the bound $K$ allows incremental verification of safe programs in cases where the other tools timeout. 

\paragraph{Parameterized Benchmarks}
\label{para:exp2}

In Table~\ref{tab:parab} we compare the performance of these tools on two parametrized benchmarks: $\texttt{ExponentialBug}$ (from Fig. 2 of \cite{huang}) and $\texttt{Fibonacci}$ 
(from SV-COMP 2019). In $\texttt{ExponentialBug}(N)$
$N$ represents the number of times a process writes to a variable. We note that in $\texttt{ExponentialBug}(N)$ the number of executions grows as $N!$, while the  processes have to follow a specific interleaving to uncover the hard to find bug. 
In $\texttt{Fibonacci}(N)$, two processes compute the value of the $n^{th}$ fibonacci number in a distributed fashion. 
Our tool performs better than the other tools on the $\texttt{ExponentialBug}$ and competes well on $\texttt{Fibonacci}$ for larger values of the parameter. These results show the ability of our tool to uncover bugs with a small value of $K$.

\begin{table}[t]
\resizebox{0.7\columnwidth}{!}{
\begin{tabular}{ccccccc}
\hline
\textbf{benchmark} & $L$ & $K$ & \textbf{\tool} & \textbf{CDSChecker} & \textbf{GenMC} & \textbf{RCMC} \\ \hline\hline
exponential\_10\_unsafe      & 10 & 10 & 1.854s          & 1.921s          & 0.367s & 3m41s           \\
exponential\_25\_unsafe      & 25 & 10 & 3.532s         & 7.239s          & 3.736s & TO                 \\
exponential\_50\_unsafe      & 50 & 10 & 6.128s         &  36.361s            & 39.920s  & TO                 \\
\hline
fibonacci\_2\_unsafe                 & 2  & 20 & 2.746s          & 2.332s     & 0.084s     & 0.086s          \\
fibonacci\_3\_unsafe                 & 3  & 20 & 9.392s         &  46m8s         & 0.462s     & 0.544s          \\
fibonacci\_4\_unsafe                 & 4  & 20 & 34.019s         & TO         & 12.437s     & 18.953s            \\ \hline 
\end{tabular}}
\caption{Comparison  on a set of parameterized benchmarks}
\label{tab:parab}
\end{table}

\paragraph{Concurrent data structures based benchmarks}
\label{para:exp3}
We  compare the tools in Table \ref{tab:ds} on benchmarks based on concurrent data structures. The first of these is a concurrent locking algorithm originating from \citet{hehner}. The second, $\texttt{LinuxLocks(N)}$ is adapted from evaluations of $\cdsc$ \cite{cdsc}. We note that if not completely fenced, it is unsafe. We fence all but one lock access. \texttt{Queue} is a \textit{safe} benchmark adapted from SV-COMP 2018, parameterized by the number of processes. We note the ability of the tool to uncover bugs with a small value of $K$.

\begin{table}[t]
\small
\resizebox{0.7\columnwidth}{!}{
\begin{tabular}{ccccccc}
\hline
\textbf{benchmark} & $L$ & $K$ & \textbf{\tool} & \textbf{CDSChecker} & \textbf{GenMC} & \textbf{RCMC} \\ \hline\hline
hehner2\_unsafe    & 4 & 5 & 7.207s & 0.033s & 0.094s & 0.087s \\ 
hehner3\_unsafe    & 4 & 5 & 28.345s & 0.036s & 2m53s & 1m13s \\ \hline
linuxlocks2\_unsafe    & 2 & 4 & 0.547s & 0.032s & 0.073s & 0.078s \\ 
linuxlocks3\_unsafe    & 2 & 4 & 1.031s & 0.031s & 0.083s & 0.081s \\ \hline
queue\_2\_safe & 4 & 4 & 0.180s & 0.031s & 0.082s & 0.085s \\
queue\_3\_safe & 4 & 4 & 0.347s & 0.037s & 0.090s & 0.092s \\ \hline
\end{tabular}}
\caption{Comparison on concurrent data structures}
\label{tab:ds}
\end{table}

\paragraph{Variations of mutual exclusion protocols}
We now consider safe and unsafe variants of mutual exclusion protocols from  SV-COMP 2018. The fully fenced versions of the protocols are \textit{safe}. We modify these protocols by introducing bugs and comparing the performance of \tool~ for bug detection with the other tools. These benchmarks are parameterized by  the number of processes. 

\begin{table}[t]
\small
\resizebox{0.7\columnwidth}{!}{
\begin{tabular}{ccccccc}
\hline
\textbf{benchmark} & $L$ & $K$ & \textbf{\tool} & \textbf{CDSChecker} & \textbf{GenMC} & \textbf{RCMC} \\ \hline\hline
peterson1U(4)    & 1 & 6 & 1.408s & 0.039s & TO & 9.129s \\ 
peterson1U(8)    & 1 & 6 & 47.786s & TO & TO & TO \\ 
\hline
szymanski1U(4)   & 1 & 2 & 1.015s & 0.043s & MLE & TO \\
szymanski1U(8)   & 1 & 2 & 6.176s & TO & TO & TO \\ 
\hline
\end{tabular}}
\caption{Comparison of performance on mutual exclusion benchmarks with a single unfenced process}
\label{tab:mutex1}
\end{table}

In Table \ref{tab:mutex1}, we  unfence a single process of the \texttt{Peterson} and \texttt{Szymanski} protocols making them \textit{unsafe}. For \tool, the value of $K$ taken is 6 and 2 respectively, asserting that  bugs can be found (even for non-trivial examples) with small $K$. We note that the other tools eventually timeout for larger values of $n$.

\setlength\intextsep{10pt}

In Table~\ref{tab:mutex2} we keep all processes fenced but introduce a bug into the critical section of a process (write a value to a shared variable and read a different value from it). We note that all other tools timeout, while \tool~ is able to detect the bug within one minute, showing that essential event-bounding is an effective technique for bug-finding. Additionally in \texttt{Peterson2C}, we vary the  example by changing the process in which we add the bug. We note that $\cdsc$, can uncover the bug in \texttt{Peterson2C(5)} in around two minutes, while for \texttt{Peterson1C(5)} it timed out. Thus, $\cdsc$ algorithm is sensitive to  changes in the position of the bug due to its DPOR exploration strategy.

\begin{table}[h]
\small
\resizebox{0.7\columnwidth}{!}{
\begin{tabular}{ccccccc}
\hline
\textbf{benchmark} & $L$ & $K$ & \textbf{\tool} & \textbf{CDSChecker} & \textbf{GenMC} & \textbf{RCMC} \\ \hline\hline
peterson1C(3)    & 1 & 2 & 0.487s & 0.053s & 0.083s & 0.087s \\ 
peterson1C(5)    & 1 & 2 & 2.713s & TO & TO & TO \\
peterson1C(7)    & 1 & 2 & 11.008s & TO & TO & TO \\ \hline
peterson2C(3)    & 1 & 2 & 0.481s & 0.032s & 0.099s & 0.091s \\ 
peterson2C(5)    & 1 & 2 & 2.801s & 1m47s & TO & TO \\
peterson2C(7)    & 1 & 2 & 11.030s & TO & TO & TO \\ \hline
\end{tabular}}
\caption{Comparison of performance on completely fenced peterson mutual exclusion benchmarks with a bug introduced in the critical section of a single process}
\label{tab:mutex2}
\vspace{-0.6cm}
\end{table}

We consider in Table~\ref{app-tab:mutex4} completely fenced versions of the mutual exclusion protocols.
In this experiment, we increase the loop unwinding bound and with it, the value of $K$. These examples exhibit the practicality of iterative increments in $K$. The other tools eventually timeout, while \tool~ is able to provide atleast partial guarantees.

\begin{table}[h]
\small
\vspace{0.1cm}
\resizebox{0.7\columnwidth}{!}{
\begin{tabular}{ccccccc}
\hline
\textbf{benchmark} & $L$ & $K$ & \tool & \textbf{CDSChecker} & \textbf{GenMC} & \textbf{RCMC} \\ \hline\hline
peterson(3)    & 1 & 2 & 0.878s & TO & 9.665s & 26.208s \\
peterson(2)    & 1 & 2 & 0.321s & 0.325s & 0.087s & 0.068s \\ \hline
peterson(3)    & 2 & 4 & 1.695s & TO & MLE & TO \\
peterson(2)    & 2 & 4 & 0.539s & 15m22s & 0.039s & 0.428s \\ \hline
peterson(3)    & 4 & 4 & 15.900s & TO & MLE & TO \\
peterson(2)    & 4 & 4 & 3.412s & TO & TO & TO \\ \hline
\end{tabular}}
\caption{Evaluation using safe mutual exclusion protocols}
\label{app-tab:mutex4}
\vspace{-0.6cm}
\end{table}
\section{Conclusion}
In this paper, we investigate decidability 
of the promising semantics, $\ps$ from \citet{promising2}. The release-acquire ($\ra$) fragment of $\ps$ with RMW operations is known to be undecidable \cite{pldi2019}. However, the decidability of the fragment of $\ps$ with only relaxed ($\rlx$) accesses (denoted $\psr$) 
was open. We started with this fragment, and obtained undecidability of the reachability 
 problem, when there is no bound on the number of promises. In the quest for decidability, we 
 considered an underapproximation of $\psr$ 
 where we bound the number of promises in any execution. The fragment of $\psr$ with bounded promises is denoted as $\bps$. 
  We showed that reachability is decidable for $\bps$. 
  Our decidability proof includes the introduction of a new memory model $\lhc$, and 
 proving the equivalence of $\psr$ and $\lhc$. 
 The decidability of $\bps$ is shown using the  theory of well structured transition systems. This also gives non-primitive recursive complexity of $\bps$, with a proof similar to RMW-free fragment of release-acquire \cite{pldi2019}. 
 
 Having explored the decidability landscape of $\ps$ thoroughly, we moved towards practical verification techniques  for $\ps$. Motivated 
 by the success of context bounded reachability in SC \cite{DBLP:conf/tacas/QadeerR05}, and subsequent notions in weak memory models, we
introduced a notion of essential events bounded reachability for $\ps$, which bounds the number  
 of promises and view altering messages in any execution. We provide a source to source translation from a concurrent program   
   under $\ps$ with this bounded notion to 
   a bounded context SC program, and implemented this in a tool \tool{}. \tool{} is the first tool capable of handling the promising framework, $\ps$ from \citet{promising2} and the $\mathsf{PS}$ model from \citet{promising}. \tool{} allows modularity with respect 
   to allowing/disallowing promises on a thread-by-thread basis. We exhibit the efficacy of this modular technique in the face of non-determinism induced by $\ps$. We also compare the performance of \tool{} with existing tools which do not support promises by operating it in the \textit{promise-free} mode (in which no threads are allowed to promise). In this case, we exhibit the effectiveness of the bounding technique in uncovering hard-to find bugs. 

\bibliography{biblio}


\begin{thebibliography}{00}


\ifx \showCODEN    \undefined \def \showCODEN     #1{\unskip}     \fi
\ifx \showDOI      \undefined \def \showDOI       #1{#1}\fi
\ifx \showISBNx    \undefined \def \showISBNx     #1{\unskip}     \fi
\ifx \showISBNxiii \undefined \def \showISBNxiii  #1{\unskip}     \fi
\ifx \showISSN     \undefined \def \showISSN      #1{\unskip}     \fi
\ifx \showLCCN     \undefined \def \showLCCN      #1{\unskip}     \fi
\ifx \shownote     \undefined \def \shownote      #1{#1}          \fi
\ifx \showarticletitle \undefined \def \showarticletitle #1{#1}   \fi
\ifx \showURL      \undefined \def \showURL       {\relax}        \fi
\providecommand\bibfield[2]{#2}
\providecommand\bibinfo[2]{#2}
\providecommand\natexlab[1]{#1}
\providecommand\showeprint[2][]{arXiv:#2}

\bibitem[\protect\citeauthoryear{Abdulla, Arora, Atig, and Krishna}{Abdulla
  et~al\mbox{.}}{2019}]%
        {pldi2019}
\bibfield{author}{\bibinfo{person}{Parosh~Aziz Abdulla}, \bibinfo{person}{Jatin
  Arora}, \bibinfo{person}{Mohamed~Faouzi Atig}, {and}
  \bibinfo{person}{Shankara~Narayanan Krishna}.}
  \bibinfo{year}{2019}\natexlab{}.
\newblock \showarticletitle{Verification of programs under the release-acquire
  semantics}. In \bibinfo{booktitle}{{\em {PLDI} 2019}}.
  \bibinfo{publisher}{{ACM}}, \bibinfo{pages}{1117--1132}.
\newblock


\bibitem[\protect\citeauthoryear{Abdulla, Atig, Bouajjani, Derevenetc,
  Leonardsson, and Meyer}{Abdulla et~al\mbox{.}}{2020}]%
        {Power_Netys20}
\bibfield{author}{\bibinfo{person}{Parosh~Aziz Abdulla},
  \bibinfo{person}{Mohamed~Faouzi Atig}, \bibinfo{person}{Ahmed Bouajjani},
  \bibinfo{person}{Egor Derevenetc}, \bibinfo{person}{Carl Leonardsson}, {and}
  \bibinfo{person}{Roland Meyer}.} \bibinfo{year}{2020}\natexlab{}.
\newblock \showarticletitle{Safety Verification under Power}. In
  \bibinfo{booktitle}{{\em {NETYS} 2020}} {\em (\bibinfo{series}{Lecture Notes
  in Computer Science})}. \bibinfo{publisher}{Springer}.
\newblock
\newblock
\shownote{to appear.}


\bibitem[\protect\citeauthoryear{Abdulla, Atig, Bouajjani, and Ngo}{Abdulla
  et~al\mbox{.}}{2017}]%
        {cb3}
\bibfield{author}{\bibinfo{person}{Parosh~Aziz Abdulla},
  \bibinfo{person}{Mohamed~Faouzi Atig}, \bibinfo{person}{Ahmed Bouajjani},
  {and} \bibinfo{person}{Tuan~Phong Ngo}.} \bibinfo{year}{2017}\natexlab{}.
\newblock \showarticletitle{Context-Bounded Analysis for {POWER}}. In
  \bibinfo{booktitle}{{\em Tools and Algorithms for the Construction and
  Analysis of Systems - 23rd International Conference, {TACAS} 2017, Held as
  Part of the European Joint Conferences on Theory and Practice of Software,
  {ETAPS} 2017, Uppsala, Sweden, April 22-29, 2017, Proceedings, Part {II}}}
  {\em (\bibinfo{series}{Lecture Notes in Computer Science})},
  \bibfield{editor}{\bibinfo{person}{Axel Legay} {and} \bibinfo{person}{Tiziana
  Margaria}} (Eds.), Vol.~\bibinfo{volume}{10206}.
  \bibinfo{publisher}{Springer}, \bibinfo{pages}{56--74}.
\newblock


\bibitem[\protect\citeauthoryear{Abdulla, Atig, Jonsson, and Ngo}{Abdulla
  et~al\mbox{.}}{2018}]%
        {phong}
\bibfield{author}{\bibinfo{person}{Parosh~Aziz Abdulla},
  \bibinfo{person}{Mohamed~Faouzi Atig}, \bibinfo{person}{Bengt Jonsson}, {and}
  \bibinfo{person}{Tuan~Phong Ngo}.} \bibinfo{year}{2018}\natexlab{}.
\newblock \showarticletitle{Optimal stateless model checking under the
  release-acquire semantics}.
\newblock \bibinfo{journal}{{\em Proc. ACM Program. Lang.\/}}
  \bibinfo{volume}{2}, \bibinfo{number}{{OOPSLA}} (\bibinfo{year}{2018}),
  \bibinfo{pages}{135:1--135:29}.
\newblock


\bibitem[\protect\citeauthoryear{Abdulla and Jonsson}{Abdulla and
  Jonsson}{1996}]%
        {wsts2}
\bibfield{author}{\bibinfo{person}{Parosh~Aziz Abdulla} {and}
  \bibinfo{person}{Bengt Jonsson}.} \bibinfo{year}{1996}\natexlab{}.
\newblock \showarticletitle{Verifying Programs with Unreliable Channels}.
\newblock \bibinfo{journal}{{\em Inf. Comput.\/}} \bibinfo{volume}{127},
  \bibinfo{number}{2} (\bibinfo{year}{1996}), \bibinfo{pages}{91--101}.
\newblock
\showDOI{%
\url{https://doi.org/10.1006/inco.1996.0053}}


\bibitem[\protect\citeauthoryear{Atig, Bouajjani, Burckhardt, and
  Musuvathi}{Atig et~al\mbox{.}}{2010}]%
        {ABBM10}
\bibfield{author}{\bibinfo{person}{Mohamed~Faouzi Atig}, \bibinfo{person}{Ahmed
  Bouajjani}, \bibinfo{person}{Sebastian Burckhardt}, {and}
  \bibinfo{person}{Madanlal Musuvathi}.} \bibinfo{year}{2010}\natexlab{}.
\newblock \showarticletitle{On the verification problem for weak memory
  models}. In \bibinfo{booktitle}{{\em Proceedings of the 37th {ACM}
  {SIGPLAN-SIGACT} Symposium on Principles of Programming Languages, {POPL}
  2010, Madrid, Spain, January 17-23, 2010}}. \bibinfo{publisher}{{ACM}},
  \bibinfo{pages}{7--18}.
\newblock


\bibitem[\protect\citeauthoryear{Atig, Bouajjani, and Parlato}{Atig
  et~al\mbox{.}}{2011}]%
        {cb2}
\bibfield{author}{\bibinfo{person}{Mohamed~Faouzi Atig}, \bibinfo{person}{Ahmed
  Bouajjani}, {and} \bibinfo{person}{Gennaro Parlato}.}
  \bibinfo{year}{2011}\natexlab{}.
\newblock \showarticletitle{Getting Rid of Store-Buffers in {TSO} Analysis}. In
  \bibinfo{booktitle}{{\em Computer Aided Verification - 23rd International
  Conference, {CAV} 2011, Snowbird, UT, USA, July 14-20, 2011. Proceedings}}
  {\em (\bibinfo{series}{Lecture Notes in Computer Science})},
  \bibfield{editor}{\bibinfo{person}{Ganesh Gopalakrishnan} {and}
  \bibinfo{person}{Shaz Qadeer}} (Eds.), Vol.~\bibinfo{volume}{6806}.
  \bibinfo{publisher}{Springer}, \bibinfo{pages}{99--115}.
\newblock


\bibitem[\protect\citeauthoryear{Batty, Owens, Sarkar, Sewell, and Weber}{Batty
  et~al\mbox{.}}{2011}]%
        {batty:2011}
\bibfield{author}{\bibinfo{person}{Mark Batty}, \bibinfo{person}{Scott Owens},
  \bibinfo{person}{Susmit Sarkar}, \bibinfo{person}{Peter Sewell}, {and}
  \bibinfo{person}{Tjark Weber}.} \bibinfo{year}{2011}\natexlab{}.
\newblock \showarticletitle{Mathematizing {C++} concurrency}. In
  \bibinfo{booktitle}{{\em {POPL} 2011}},
  \bibfield{editor}{\bibinfo{person}{Thomas Ball} {and} \bibinfo{person}{Mooly
  Sagiv}} (Eds.). \bibinfo{publisher}{{ACM}}, \bibinfo{pages}{55--66}.
\newblock
\showDOI{%
\url{https://doi.org/10.1145/1926385.1926394}}


\bibitem[\protect\citeauthoryear{Beyer}{Beyer}{2019}]%
        {beyer2019automatic}
\bibfield{author}{\bibinfo{person}{Dirk Beyer}.}
  \bibinfo{year}{2019}\natexlab{}.
\newblock \showarticletitle{Automatic verification of C and Java programs:
  SV-COMP 2019}. In \bibinfo{booktitle}{{\em International Conference on Tools
  and Algorithms for the Construction and Analysis of Systems}}. Springer,
  \bibinfo{pages}{133--155}.
\newblock


\bibitem[\protect\citeauthoryear{Chakraborty and Vafeiadis}{Chakraborty and
  Vafeiadis}{2019a}]%
        {weakestmo}
\bibfield{author}{\bibinfo{person}{Soham Chakraborty} {and}
  \bibinfo{person}{Viktor Vafeiadis}.} \bibinfo{year}{2019}\natexlab{a}.
\newblock \showarticletitle{Grounding thin-air reads with event structures}.
\newblock \bibinfo{journal}{{\em {PACMPL}\/}} \bibinfo{volume}{3},
  \bibinfo{number}{{POPL}} (\bibinfo{year}{2019}),
  \bibinfo{pages}{70:1--70:28}.
\newblock
\showDOI{%
\url{https://doi.org/10.1145/3290383}}


\bibitem[\protect\citeauthoryear{Chakraborty and Vafeiadis}{Chakraborty and
  Vafeiadis}{2019b}]%
        {thinair}
\bibfield{author}{\bibinfo{person}{Soham~Sundar Chakraborty} {and}
  \bibinfo{person}{Viktor Vafeiadis}.} \bibinfo{year}{2019}\natexlab{b}.
\newblock \showarticletitle{Grounding thin-air reads with event structures}.
\newblock \bibinfo{journal}{{\em PACMPL\/}}  \bibinfo{volume}{3}
  (\bibinfo{year}{2019}), \bibinfo{pages}{70:1--70:28}.
\newblock


\bibitem[\protect\citeauthoryear{Crary and Sullivan}{Crary and
  Sullivan}{2015}]%
        {crary-sullivan:2015}
\bibfield{author}{\bibinfo{person}{Karl Crary} {and}
  \bibinfo{person}{Michael~J. Sullivan}.} \bibinfo{year}{2015}\natexlab{}.
\newblock \showarticletitle{A Calculus for Relaxed Memory}. In
  \bibinfo{booktitle}{{\em POPL 2015}},
  \bibfield{editor}{\bibinfo{person}{Sriram~K. Rajamani} {and}
  \bibinfo{person}{David Walker}} (Eds.). \bibinfo{publisher}{{ACM}},
  \bibinfo{pages}{623--636}.
\newblock
\showDOI{%
\url{https://doi.org/10.1145/2676726.2676984}}


\bibitem[\protect\citeauthoryear{Emmi, Qadeer, and Rakamaric}{Emmi
  et~al\mbox{.}}{2011}]%
        {DBLP:conf/popl/EmmiQR11}
\bibfield{author}{\bibinfo{person}{Michael Emmi}, \bibinfo{person}{Shaz
  Qadeer}, {and} \bibinfo{person}{Zvonimir Rakamaric}.}
  \bibinfo{year}{2011}\natexlab{}.
\newblock \showarticletitle{Delay-bounded scheduling}. In
  \bibinfo{booktitle}{{\em Proceedings of the 38th {ACM} {SIGPLAN-SIGACT}
  Symposium on Principles of Programming Languages, {POPL} 2011, Austin, TX,
  USA, January 26-28, 2011}}, \bibfield{editor}{\bibinfo{person}{Thomas Ball}
  {and} \bibinfo{person}{Mooly Sagiv}} (Eds.). \bibinfo{publisher}{{ACM}},
  \bibinfo{pages}{411--422}.
\newblock


\bibitem[\protect\citeauthoryear{Finkel and Schnoebelen}{Finkel and
  Schnoebelen}{2001}]%
        {wsts1}
\bibfield{author}{\bibinfo{person}{Alain Finkel} {and}
  \bibinfo{person}{Philippe Schnoebelen}.} \bibinfo{year}{2001}\natexlab{}.
\newblock \showarticletitle{Well-structured transition systems everywhere!}
\newblock \bibinfo{journal}{{\em Theor. Comput. Sci.\/}} \bibinfo{volume}{256},
  \bibinfo{number}{1-2} (\bibinfo{year}{2001}), \bibinfo{pages}{63--92}.
\newblock
\showDOI{%
\url{https://doi.org/10.1016/S0304-3975(00)00102-X}}


\bibitem[\protect\citeauthoryear{Hehner and Shyamasundar}{Hehner and
  Shyamasundar}{1981}]%
        {hehner}
\bibfield{author}{\bibinfo{person}{Eric~C.R. Hehner} {and}
  \bibinfo{person}{R.K. Shyamasundar}.} \bibinfo{year}{1981}\natexlab{}.
\newblock \showarticletitle{An implementation of P and V}.
\newblock \bibinfo{journal}{{\it Inform. Process. Lett.}} \bibinfo{volume}{12},
  \bibinfo{number}{4} (\bibinfo{year}{1981}), \bibinfo{pages}{196 -- 198}.
\newblock
\showISSN{0020-0190}
\showDOI{%
\url{https://doi.org/10.1016/0020-0190(81)90100-9}}


\bibitem[\protect\citeauthoryear{Higman}{Higman}{1952}]%
        {higman}
\bibfield{author}{\bibinfo{person}{Graham Higman}.}
  \bibinfo{year}{1952}\natexlab{}.
\newblock \showarticletitle{Ordering by Divisibility in Abstract Algebras}.
\newblock \bibinfo{journal}{{\em Proceedings of the London Mathematical
  Society\/}} \bibinfo{volume}{s3-2}, \bibinfo{number}{1}
  (\bibinfo{year}{1952}), \bibinfo{pages}{326--336}.
\newblock
\showDOI{%
\url{https://doi.org/10.1112/plms/s3-2.1.326}}
\showeprint{https://londmathsoc.onlinelibrary.wiley.com/doi/pdf/10.1112/plms/s3-2.1.326}


\bibitem[\protect\citeauthoryear{Huang}{Huang}{2015}]%
        {huang}
\bibfield{author}{\bibinfo{person}{Jeff Huang}.}
  \bibinfo{year}{2015}\natexlab{}.
\newblock \showarticletitle{Stateless model checking concurrent programs with
  maximal causality reduction}. In \bibinfo{booktitle}{{\em Proceedings of the
  36th {ACM} {SIGPLAN} Conference on Programming Language Design and
  Implementation, Portland, OR, USA, June 15-17, 2015}},
  \bibfield{editor}{\bibinfo{person}{David Grove} {and} \bibinfo{person}{Steve
  Blackburn}} (Eds.). \bibinfo{publisher}{{ACM}}, \bibinfo{pages}{165--174}.
\newblock


\bibitem[\protect\citeauthoryear{Jeffrey and Riely}{Jeffrey and Riely}{2019}]%
        {jeffrey-riely:2019}
\bibfield{author}{\bibinfo{person}{Alan Jeffrey} {and} \bibinfo{person}{James
  Riely}.} \bibinfo{year}{2019}\natexlab{}.
\newblock \showarticletitle{On Thin Air Reads: Towards an Event Structures
  Model of Relaxed Memory}.
\newblock \bibinfo{journal}{{\em Logical Methods in Computer Science\/}}
  \bibinfo{volume}{15}, \bibinfo{number}{1} (\bibinfo{year}{2019}).
\newblock
\showDOI{%
\url{https://doi.org/10.23638/LMCS-15(1:33)2019}}


\bibitem[\protect\citeauthoryear{Kang, Hur, Lahav, Vafeiadis, and Dreyer}{Kang
  et~al\mbox{.}}{2017}]%
        {promising}
\bibfield{author}{\bibinfo{person}{Jeehoon Kang}, \bibinfo{person}{Chung{-}Kil
  Hur}, \bibinfo{person}{Ori Lahav}, \bibinfo{person}{Viktor Vafeiadis}, {and}
  \bibinfo{person}{Derek Dreyer}.} \bibinfo{year}{2017}\natexlab{}.
\newblock \showarticletitle{A promising semantics for relaxed-memory
  concurrency}. In \bibinfo{booktitle}{{\em POPL 2017}},
  \bibfield{editor}{\bibinfo{person}{Giuseppe Castagna} {and}
  \bibinfo{person}{Andrew~D. Gordon}} (Eds.). \bibinfo{publisher}{{ACM}},
  \bibinfo{pages}{175--189}.
\newblock


\bibitem[\protect\citeauthoryear{Kokologiannakis, Lahav, Sagonas, and
  Vafeiadis}{Kokologiannakis et~al\mbox{.}}{2017}]%
        {rcmc}
\bibfield{author}{\bibinfo{person}{Michalis Kokologiannakis},
  \bibinfo{person}{Ori Lahav}, \bibinfo{person}{Konstantinos Sagonas}, {and}
  \bibinfo{person}{Viktor Vafeiadis}.} \bibinfo{year}{2017}\natexlab{}.
\newblock \showarticletitle{Effective Stateless Model Checking for C/C++
  Concurrency}.
\newblock \bibinfo{journal}{{\em Proc. ACM Program. Lang.\/}}
  \bibinfo{volume}{2}, \bibinfo{number}{POPL}, Article \bibinfo{articleno}{17}
  (\bibinfo{date}{Dec.} \bibinfo{year}{2017}), \bibinfo{numpages}{32}~pages.
\newblock
\showISSN{2475-1421}
\showDOI{%
\url{https://doi.org/10.1145/3158105}}


\bibitem[\protect\citeauthoryear{Kokologiannakis, Raad, and
  Vafeiadis}{Kokologiannakis et~al\mbox{.}}{2019}]%
        {genmc}
\bibfield{author}{\bibinfo{person}{Michalis Kokologiannakis},
  \bibinfo{person}{Azalea Raad}, {and} \bibinfo{person}{Viktor Vafeiadis}.}
  \bibinfo{year}{2019}\natexlab{}.
\newblock \showarticletitle{Model checking for weakly consistent libraries}. In
  \bibinfo{booktitle}{{\em PLDI}}.
\newblock
\showDOI{%
\url{https://doi.org/10.1145/3314221.3314649}}


\bibitem[\protect\citeauthoryear{{La Torre}, Madhusudan, and Parlato}{{La
  Torre} et~al\mbox{.}}{2008}]%
        {madhu3}
\bibfield{author}{\bibinfo{person}{Salvatore {La Torre}}, \bibinfo{person}{P.
  Madhusudan}, {and} \bibinfo{person}{Gennaro Parlato}.}
  \bibinfo{year}{2008}\natexlab{}.
\newblock \showarticletitle{Context-Bounded Analysis of Concurrent Queue
  Systems}. In \bibinfo{booktitle}{{\em Tools and Algorithms for the
  Construction and Analysis of Systems, 14th International Conference, {TACAS}
  2008, Held as Part of the Joint European Conferences on Theory and Practice
  of Software, {ETAPS} 2008, Budapest, Hungary, March 29-April 6, 2008.
  Proceedings}} {\em (\bibinfo{series}{Lecture Notes in Computer Science})},
  \bibfield{editor}{\bibinfo{person}{C.~R. Ramakrishnan} {and}
  \bibinfo{person}{Jakob Rehof}} (Eds.), Vol.~\bibinfo{volume}{4963}.
  \bibinfo{publisher}{Springer}, \bibinfo{pages}{299--314}.
\newblock


\bibitem[\protect\citeauthoryear{{La Torre}, Madhusudan, and Parlato}{{La
  Torre} et~al\mbox{.}}{2009}]%
        {DBLP:conf/cav/TorreMP09}
\bibfield{author}{\bibinfo{person}{Salvatore {La Torre}}, \bibinfo{person}{P.
  Madhusudan}, {and} \bibinfo{person}{Gennaro Parlato}.}
  \bibinfo{year}{2009}\natexlab{}.
\newblock \showarticletitle{Reducing Context-Bounded Concurrent Reachability to
  Sequential Reachability}. In \bibinfo{booktitle}{{\em Computer Aided
  Verification, 21st International Conference, {CAV} 2009, Grenoble, France,
  June 26 - July 2, 2009. Proceedings}} {\em (\bibinfo{series}{Lecture Notes in
  Computer Science})}, \bibfield{editor}{\bibinfo{person}{Ahmed Bouajjani}
  {and} \bibinfo{person}{Oded Maler}} (Eds.), Vol.~\bibinfo{volume}{5643}.
  \bibinfo{publisher}{Springer}, \bibinfo{pages}{477--492}.
\newblock


\bibitem[\protect\citeauthoryear{{La Torre}, Madhusudan, and Parlato}{{La
  Torre} et~al\mbox{.}}{2010}]%
        {DBLP:conf/cav/TorreMP10}
\bibfield{author}{\bibinfo{person}{Salvatore {La Torre}}, \bibinfo{person}{P.
  Madhusudan}, {and} \bibinfo{person}{Gennaro Parlato}.}
  \bibinfo{year}{2010}\natexlab{}.
\newblock \showarticletitle{Model-Checking Parameterized Concurrent Programs
  Using Linear Interfaces}. In \bibinfo{booktitle}{{\em Computer Aided
  Verification, 22nd International Conference, {CAV} 2010, Edinburgh, UK, July
  15-19, 2010. Proceedings}} {\em (\bibinfo{series}{Lecture Notes in Computer
  Science})}, \bibfield{editor}{\bibinfo{person}{Tayssir Touili},
  \bibinfo{person}{Byron Cook}, {and} \bibinfo{person}{Paul~B. Jackson}}
  (Eds.), Vol.~\bibinfo{volume}{6174}. \bibinfo{publisher}{Springer},
  \bibinfo{pages}{629--644}.
\newblock


\bibitem[\protect\citeauthoryear{Lahav and Boker}{Lahav and Boker}{2020}]%
        {DBLP:conf/pldi/LahavB20}
\bibfield{author}{\bibinfo{person}{Ori Lahav} {and} \bibinfo{person}{Udi
  Boker}.} \bibinfo{year}{2020}\natexlab{}.
\newblock \showarticletitle{Decidable verification under a causally consistent
  shared memory}. In \bibinfo{booktitle}{{\em Proceedings of the 41st {ACM}
  {SIGPLAN} International Conference on Programming Language Design and
  Implementation, {PLDI} 2020, London, UK, June 15-20, 2020}},
  \bibfield{editor}{\bibinfo{person}{Alastair~F. Donaldson} {and}
  \bibinfo{person}{Emina Torlak}} (Eds.). \bibinfo{publisher}{{ACM}},
  \bibinfo{pages}{211--226}.
\newblock


\bibitem[\protect\citeauthoryear{Lahav, Vafeiadis, Kang, Hur, and Dreyer}{Lahav
  et~al\mbox{.}}{2017}]%
        {rc11}
\bibfield{author}{\bibinfo{person}{Ori Lahav}, \bibinfo{person}{Viktor
  Vafeiadis}, \bibinfo{person}{Jeehoon Kang}, \bibinfo{person}{Chung{-}Kil
  Hur}, {and} \bibinfo{person}{Derek Dreyer}.} \bibinfo{year}{2017}\natexlab{}.
\newblock \showarticletitle{Repairing sequential consistency in {C/C++11}}. In
  \bibinfo{booktitle}{{\em PLDI 2017}},
  \bibfield{editor}{\bibinfo{person}{Albert Cohen} {and}
  \bibinfo{person}{Martin~T. Vechev}} (Eds.). \bibinfo{publisher}{{ACM}},
  \bibinfo{pages}{618--632}.
\newblock
\showDOI{%
\url{https://doi.org/10.1145/3062341.3062352}}


\bibitem[\protect\citeauthoryear{Lal and Reps}{Lal and Reps}{2009}]%
        {DBLP:journals/fmsd/LalR09}
\bibfield{author}{\bibinfo{person}{Akash Lal} {and} \bibinfo{person}{Thomas~W.
  Reps}.} \bibinfo{year}{2009}\natexlab{}.
\newblock \showarticletitle{Reducing concurrent analysis under a context bound
  to sequential analysis}.
\newblock \bibinfo{journal}{{\em Formal Methods in System Design\/}}
  \bibinfo{volume}{35}, \bibinfo{number}{1} (\bibinfo{year}{2009}),
  \bibinfo{pages}{73--97}.
\newblock


\bibitem[\protect\citeauthoryear{Lee, Cho, Podkopaev, Chakraborty, Hur, Lahav,
  and Vafeiadis}{Lee et~al\mbox{.}}{2020}]%
        {promising2}
\bibfield{author}{\bibinfo{person}{Sung{-}Hwan Lee}, \bibinfo{person}{Minki
  Cho}, \bibinfo{person}{Anton Podkopaev}, \bibinfo{person}{Soham Chakraborty},
  \bibinfo{person}{Chung{-}Kil Hur}, \bibinfo{person}{Ori Lahav}, {and}
  \bibinfo{person}{Viktor Vafeiadis}.} \bibinfo{year}{2020}\natexlab{}.
\newblock \showarticletitle{Promising 2.0: global optimizations in relaxed
  memory concurrency}. In \bibinfo{booktitle}{{\em Proceedings of the 41st
  {ACM} {SIGPLAN} International Conference on Programming Language Design and
  Implementation, {PLDI} 2020, London, UK, June 15-20, 2020}},
  \bibfield{editor}{\bibinfo{person}{Alastair~F. Donaldson} {and}
  \bibinfo{person}{Emina Torlak}} (Eds.). \bibinfo{publisher}{{ACM}},
  \bibinfo{pages}{362--376}.
\newblock


\bibitem[\protect\citeauthoryear{Manson, Pugh, and Adve}{Manson
  et~al\mbox{.}}{2005}]%
        {jmm}
\bibfield{author}{\bibinfo{person}{Jeremy Manson}, \bibinfo{person}{William
  Pugh}, {and} \bibinfo{person}{Sarita~V. Adve}.}
  \bibinfo{year}{2005}\natexlab{}.
\newblock \showarticletitle{The {Java} memory model}. In
  \bibinfo{booktitle}{{\em POPL 2015}}, \bibfield{editor}{\bibinfo{person}{Jens
  Palsberg} {and} \bibinfo{person}{Mart{\'{\i}}n Abadi}} (Eds.).
  \bibinfo{publisher}{{ACM}}, \bibinfo{pages}{378--391}.
\newblock
\showDOI{%
\url{https://doi.org/10.1145/1040305.1040336}}


\bibitem[\protect\citeauthoryear{Musuvathi and Qadeer}{Musuvathi and
  Qadeer}{2007}]%
        {MQ07}
\bibfield{author}{\bibinfo{person}{Madanlal Musuvathi} {and}
  \bibinfo{person}{Shaz Qadeer}.} \bibinfo{year}{2007}\natexlab{}.
\newblock \showarticletitle{Iterative context bounding for systematic testing
  of multithreaded programs}. In \bibinfo{booktitle}{{\em Proceedings of the
  {ACM} {SIGPLAN} 2007 Conference on Programming Language Design and
  Implementation, San Diego, California, USA, June 10-13, 2007}},
  \bibfield{editor}{\bibinfo{person}{Jeanne Ferrante} {and}
  \bibinfo{person}{Kathryn~S. McKinley}} (Eds.). \bibinfo{publisher}{{ACM}},
  \bibinfo{pages}{446--455}.
\newblock


\bibitem[\protect\citeauthoryear{Norris and Demsky}{Norris and Demsky}{2013}]%
        {cdsc}
\bibfield{author}{\bibinfo{person}{Brian Norris} {and} \bibinfo{person}{Brian
  Demsky}.} \bibinfo{year}{2013}\natexlab{}.
\newblock \showarticletitle{CDSchecker: Checking Concurrent Data Structures
  Written with C/C++ Atomics}. In \bibinfo{booktitle}{{\em OOPSLA 2013}}.
  \bibinfo{publisher}{ACM}, \bibinfo{address}{New York, NY, USA},
  \bibinfo{pages}{131--150}.
\newblock
\showISBNx{978-1-4503-2374-1}
\showDOI{%
\url{https://doi.org/10.1145/2509136.2509514}}


\bibitem[\protect\citeauthoryear{Norris and Demsky}{Norris and Demsky}{2016}]%
        {demsky}
\bibfield{author}{\bibinfo{person}{Brian Norris} {and} \bibinfo{person}{Brian
  Demsky}.} \bibinfo{year}{2016}\natexlab{}.
\newblock \showarticletitle{A Practical Approach for Model Checking C/C++11
  Code}.
\newblock \bibinfo{journal}{{\em ACM Trans. Program. Lang. Syst.\/}}
  \bibinfo{volume}{38}, \bibinfo{number}{3}, Article \bibinfo{articleno}{10}
  (\bibinfo{date}{May} \bibinfo{year}{2016}), \bibinfo{numpages}{51}~pages.
\newblock
\showISSN{0164-0925}
\showDOI{%
\url{https://doi.org/10.1145/2806886}}


\bibitem[\protect\citeauthoryear{Paviotti, Cooksey, Paradis, Wright, Owens, and
  Batty}{Paviotti et~al\mbox{.}}{2020}]%
        {mrder}
\bibfield{author}{\bibinfo{person}{Marco Paviotti}, \bibinfo{person}{Simon
  Cooksey}, \bibinfo{person}{Anouk Paradis}, \bibinfo{person}{Daniel Wright},
  \bibinfo{person}{Scott Owens}, {and} \bibinfo{person}{Mark Batty}.}
  \bibinfo{year}{2020}\natexlab{}.
\newblock \showarticletitle{Modular Relaxed Dependencies in Weak Memory
  Concurrency}. In \bibinfo{booktitle}{{\em Programming Languages and Systems -
  29th European Symposium on Programming, {ESOP} 2020, Held as Part of the
  European Joint Conferences on Theory and Practice of Software, {ETAPS} 2020,
  Dublin, Ireland, April 25-30, 2020, Proceedings}}. \bibinfo{pages}{599--625}.
\newblock
\showDOI{%
\url{https://doi.org/10.1007/978-3-030-44914-8\_22}}


\bibitem[\protect\citeauthoryear{Pichon{-}Pharabod and
  Sewell}{Pichon{-}Pharabod and Sewell}{2016}]%
        {bubbly}
\bibfield{author}{\bibinfo{person}{Jean Pichon{-}Pharabod} {and}
  \bibinfo{person}{Peter Sewell}.} \bibinfo{year}{2016}\natexlab{}.
\newblock \showarticletitle{A concurrency semantics for relaxed atomics that
  permits optimisation and avoids thin-air executions}. In
  \bibinfo{booktitle}{{\em POPL 2016}},
  \bibfield{editor}{\bibinfo{person}{Rastislav Bod{\'{\i}}k} {and}
  \bibinfo{person}{Rupak Majumdar}} (Eds.). \bibinfo{publisher}{{ACM}},
  \bibinfo{pages}{622--633}.
\newblock
\showDOI{%
\url{https://doi.org/10.1145/2837614.2837616}}


\bibitem[\protect\citeauthoryear{Post}{Post}{1946}]%
        {post}
\bibfield{author}{\bibinfo{person}{Emil~L. Post}.}
  \bibinfo{year}{1946}\natexlab{}.
\newblock \showarticletitle{A variant of a recursively unsolvable problem}.
\newblock \bibinfo{journal}{{\em Bull. Amer. Math. Soc.\/}}
  \bibinfo{volume}{52} (\bibinfo{year}{1946}), \bibinfo{pages}{264--268}.
\newblock


\bibitem[\protect\citeauthoryear{Qadeer and Rehof}{Qadeer and Rehof}{2005}]%
        {DBLP:conf/tacas/QadeerR05}
\bibfield{author}{\bibinfo{person}{Shaz Qadeer} {and} \bibinfo{person}{Jakob
  Rehof}.} \bibinfo{year}{2005}\natexlab{}.
\newblock \showarticletitle{Context-Bounded Model Checking of Concurrent
  Software}. In \bibinfo{booktitle}{{\em {TACAS} 2005}} {\em
  (\bibinfo{series}{LNCS})}, Vol.~\bibinfo{volume}{3440}.
  \bibinfo{publisher}{Springer}, \bibinfo{pages}{93--107}.
\newblock


\bibitem[\protect\citeauthoryear{Svendsen, Pichon{-}Pharabod, Doko, Lahav, and
  Vafeiadis}{Svendsen et~al\mbox{.}}{2018}]%
        {Svendsen:2018}
\bibfield{author}{\bibinfo{person}{Kasper Svendsen}, \bibinfo{person}{Jean
  Pichon{-}Pharabod}, \bibinfo{person}{Marko Doko}, \bibinfo{person}{Ori
  Lahav}, {and} \bibinfo{person}{Viktor Vafeiadis}.}
  \bibinfo{year}{2018}\natexlab{}.
\newblock \showarticletitle{A Separation Logic for a Promising Semantics}. In
  \bibinfo{booktitle}{{\em 27th European Symposium on Programming, {ESOP}
  2018}} {\em (\bibinfo{series}{LNCS})},
  \bibfield{editor}{\bibinfo{person}{Amal Ahmed}} (Ed.),
  Vol.~\bibinfo{volume}{10801}. \bibinfo{publisher}{Springer},
  \bibinfo{pages}{357--384}.
\newblock
\showDOI{%
\url{https://doi.org/10.1007/978-3-319-89884-1\_13}}


\bibitem[\protect\citeauthoryear{Tomasco, Nguyen, Fischer, {La Torre}, and
  Parlato}{Tomasco et~al\mbox{.}}{2017}]%
        {DBLP:conf/sefm/TomascoN0TP17}
\bibfield{author}{\bibinfo{person}{Ermenegildo Tomasco},
  \bibinfo{person}{Truc~Lam Nguyen}, \bibinfo{person}{Bernd Fischer},
  \bibinfo{person}{Salvatore {La Torre}}, {and} \bibinfo{person}{Gennaro
  Parlato}.} \bibinfo{year}{2017}\natexlab{}.
\newblock \showarticletitle{Using Shared Memory Abstractions to Design Eager
  Sequentializations for Weak Memory Models}. In \bibinfo{booktitle}{{\em
  Software Engineering and Formal Methods - 15th International Conference,
  {SEFM} 2017, Trento, Italy, September 4-8, 2017, Proceedings}} {\em
  (\bibinfo{series}{Lecture Notes in Computer Science})},
  \bibfield{editor}{\bibinfo{person}{Alessandro Cimatti} {and}
  \bibinfo{person}{Marjan Sirjani}} (Eds.), Vol.~\bibinfo{volume}{10469}.
  \bibinfo{publisher}{Springer}, \bibinfo{pages}{185--202}.
\newblock


\bibitem[\protect\citeauthoryear{Zhang and Feng}{Zhang and Feng}{2013}]%
        {zhang-feng:2013}
\bibfield{author}{\bibinfo{person}{Yang Zhang} {and} \bibinfo{person}{Xinyu
  Feng}.} \bibinfo{year}{2013}\natexlab{}.
\newblock \showarticletitle{An Operational Approach to Happens-Before Memory
  Model}. In \bibinfo{booktitle}{{\em Seventh International Symposium on
  Theoretical Aspects of Software Engineering, {TASE} 2013, 1-3 July 2013,
  Birmingham, {UK}}}. \bibinfo{publisher}{{IEEE} Computer Society},
  \bibinfo{pages}{121--128}.
\newblock
\showDOI{%
\url{https://doi.org/10.1109/TASE.2013.24}}


\end{thebibliography}

\newpage
\appendix

\section{Details for Section \ref{sec:dec}}
In this section, we give details of lemmas from 
Section \ref{sec:dec}.

\subsection{Equivalence of $\psr$ and $\lhc$}
To prove Theorem \ref{thm:eqv},  
we show the following: Given a program $\prog$, 
 starting from the initial machine state $\mathcal{MS}_{\init}
 =((J_{\init}, R_{\init}), \mathsf{V}_{\init}, \mathsf{PS}_{\init}, 
 M_{\init}, G_{\init})$
 in $\psr$, we can  reach in $\psr$ the machine state 
$\mathcal{MS}_n=((J_n, R_n), \mathsf{V}_{n}, \mathsf{PS}_{n}, 
 M_{n}, G_{n})$ with $\mathsf{PS}_{n}(p)=\emptyset$ for all $p \in \procset$ iff, 
 starting from an initial $\lhc$ two phases state $\Ss_{\init}=(\nor, p, \lst_{init}, \lst_{init})$,  
 we reach the  state 
 $(\nor, -, ((J_n, R_n),\chh_n), -)$, such that   $\chh_n(x)$  does not contain any memory type of the form $(\prm,-,p,-)$ or $(\prm,-,p,-,-)$ for all $x \in \varset$.   
 The equivalence of the runs follows from the fact that 
 the  sequence of instructions followed in each phase $\nor$ and   $\cert$ are same in 
 both $\psr$ and $\lhc$ 
 ; $\lhc$ allows lossy transitions which does not affect reachability.  Moreover, the $\lhc$ run satisfies the following invariants.

\noindent{\bf{Invariants for $\chh$}}. 
The following invariants hold good for $\chh(x)$ for all $x \in \varset$. 
We then say that $\chh(x)$ is faithful to the sub memory $M(x)$ and the view mapping. 
 \begin{itemize}
\item[(\textbf{Inv1})] For all $x \in \varset$, 
$\chh(x)$ is {\em well-formed} : for each process $\proc \in \procset$, there is a unique  
 position $i$ in $\chh(x)$ having $p$ in its pointer set;
\item[(\textbf{Inv2})]  For all $i > 
\ptr(p,\chh(x))$,  we have $\chh(x)[i] {\notin} \{(\msgg, -, p, -), (\msgg, -, p, -, -)\}$. 
This says that memory types at positions greater than the pointer of $p$  
cannot correspond to messages added by $p$ to $M(x)$.
\end{itemize} 

\begin{lemma}
The higher order words $\chh(x)$ for all $x \in \varset$  appearing in the states of  
a $\lhc$ run  	satisfy invariants  \textbf{Inv1} and 
\textbf{Inv2}.
\label{lem:inv}
\end{lemma}

Lemma \ref{lem:inv} can be proved by inducting on the length of a $\lhc$ run, starting   
from the initial states, using the following. 
\begin{itemize}
\item For each memory type $(\msgg, v, p, S, -)$ or  
$(\msgg, v, p, S)$ 
in $\chh(x)$, there is a message in $M(x)$ which 
was added by process $p$,  having value $v$.
  Similarly, for each memory type $(\prm, v, p, S, -)$ or  
$(\prm, v, p, S)$ 
in $\chh(x)$, there is a promise in $M(x)$ which 
was added by process $p$, having value $v$.
	\item 
 The order between memory types in $\chh(x)$ 
  and the corresponding 
 messages in $M(x)$ are the same. That is, for $i < j$, the messages or promises $m, m' \in M(x)$ 
 corresponding to $\chh(x)[i]$ and $\chh(x)[j]$ are such that  
  $m.\too < m'.\too$. 
 \item the elements in the pointer set of  
a memory type $m$ in $\chh(x)$ are exactly the set of processes whose local view 
is the $\too$ stamp of the element of $M(x)$ corresponding to $m$.  
\end{itemize}
 
The base case is easy : the initial two-phases $\lhc$ state has the same local process states 
as the initial $\ps$ machine state; moreover, the invariants trivially hold, since all process  
pointers are at the same position. 

For the inductive hypothesis, assume that both invariants 
hold in a $\lhc$ run after $i$ steps. To show that they continue 
to hold good after $i+1$ steps, we have to show that for all 
$\lhc$ transitions that can be taken after $i$ steps,  they are preserved. Assume that the two phases $\lhc$ state 
at the end of $i$ steps is $(\nor, p, \lst, \lst')$. The proof for the case when 
we have a state $(\cert, p, \lst, \lst')$ after $i$ steps of the $\lhc$ run is similar.
\begin{itemize}
	\item Assume that we have the transition $\xrightarrow[p]{\rd(x,v)}$. Then $\ptr(p, \chh(x))$ 
	is updated in the resultant state, and so are $(J,R)$, Clearly, the higher order word in the resultant state satisfies both invariants  
	 since the starting state does.
	 \item Assume that we have the transition 
	$\xrightarrow[p]{\wt(x,v)}$. Then 
	we remove $p$ from the pointer set at position $i=\ptr(p, \chh(x))$. A new simple word is added at a position $>i$, or a memory type $(\msgg, v, p, \{p\})$ is added at a position $j>i$, right next to a $\#$, by moving 
	the memory type at $j$ to position $j-2$. 
		 In either case, the 
	resultant higher order word satisfies both invariants, since 
	the starting state does.  
	\item The update rule $\xrightarrow[p]{\upd(x, v_r, v_w)}$ 
	combines the above two cases, by first performing a read and then atomically the write. From the above two cases, the invariants can be seen to hold good in the higher order words in the state obtained after the transition.
	\item Consider the Promise rule. In this case, we do not remove $p$ from its pointer set, and only 
	 add the memory type $(\prm, v, p, \{\})$ ahead of $\ptr(p, \chh(x))$.  Note that 
	\textbf{Inv2} only requires that there are no memory types of the form $(\msgg, v, p, S)$
	or $(\msgg, v, p, S,-)$ ahead of $\ptr(p, \chh(x))$. Clearly, both invariants continue to hold. 
	\item Consider a fulfil rule obtained as a write. 
	In this case, $p$ is deleted from the position $\ptr(p, \chh(x))$; and the memory type $(\prm, v, p, S)$ (or $(\prm, v, p, S.-)$) is replaced with $(\msgg, v, p, S \cup \{p\})$ (or $(\msgg, v, p, S \cup \{p\})$). It is easy to see both invariants holding good.
	\item Consider the  reservation rule. This 
	does not affect the invariants since we only tag the last component of a memory type with the process making the reservation. 
	\item Consider the SC fence rule. If $\ptr(p, \chh(x)) > \ptr(g, \chh(x))$, then, 
	in the resultant word, $p$ is moved to $\ptr(g, \chh(x))$. The case when $\ptr(p, \chh(x)) < \ptr(g, \chh(x))$, 
	is handled by moving $g$ to $\ptr(p, \chh(x))$. Since this is the only change  in the resultant higher 
	order words, clearly, both invariants hold good.  
		\end{itemize} 
 
   Notice that the arguments above hold good for both modes $a \in \{\stdd, \cert\}$.

   To prove Theorem \ref{thm:eqv},  
we show the following: Given a program $\prog$, 
 starting from the initial machine state $\mathcal{MS}_{\init}
 =((J_{\init}, R_{\init}), \mathsf{V}_{\init}, \mathsf{PS}_{\init}, 
 M_{\init}, G_{\init})$
 in $\psr$, we can  reach the machine state 
$\mathcal{MS}_n$=$((J_n, R_n), \mathsf{V}_{n}, \mathsf{PS}_{n}, 
 M_{n}, G_{n})$ with $\mathsf{PS}_{n}(p)=\emptyset$ for all $p \in \procset$ iff, 
 starting from an initial $\lhc$ two phases state $\Ss_{\init}=(\nor, p, \lst_{init}, \lst_{init})$,  
 we reach the  state 
 $(\nor, -, ((J_n, R_n),\chh_n), -)$, such that   $\chh_n(x)$  does not contain any memory type of the form $(\prm,-,p,-)$ or $(\prm,-,p,-,-)$ for all $x \in \varset$.   
 The equivalence of the runs follows from the fact that 
 the  sequence of instructions followed in each phase $\nor$ and   $\cert$ are same in 
 both $\psr$ and $\lhc$ 
 ; $\lhc$ allows lossy transitions which does not affect reachability.  Moreover, the $\lhc$ run satisfies the following invariants.

\noindent{\bf{Invariants for $\chh$}}. 
The following invariants hold good for $\chh(x)$ for all $x \in \varset$. 
We then say that $\chh(x)$ is faithful to the sub memory $M(x)$ and the view mapping. 
 \begin{itemize}
\item[(\textbf{Inv1})] For all $x \in \varset$, 
$\chh(x)$ is {\em well-formed} : for each process $\proc \in \procset$, there is a unique  
 position $i$ in $\chh(x)$ having $p$ in its pointer set;
\item[(\textbf{Inv2})]  For all $i > 
\ptr(p,\chh(x))$,  we have $\chh(x)[i] {\notin} \{(\msgg, -, p, -), (\msgg, -, p, -, -)\}$. 
This says that memory types at positions greater than the pointer of $p$  
cannot correspond to messages added by $p$ to $M(x)$.
\end{itemize} 
\subsection*{All $\chh(x)$ respect Invariants \textbf{Inv1} and 
\textbf{Inv2}}
\begin{lemma}
The higher order words $\chh(x)$ for all $x \in \varset$  appearing in the states of  
a $\lhc$ run  	satisfy invariants  \textbf{Inv1} and 
\textbf{Inv2}.
\label{lem:inv}
\end{lemma}

Lemma \ref{lem:inv} can be proved by inducting on the length of a $\lhc$ run, starting   
from the initial states, using the following. 
\begin{itemize}
\item For each memory type $(\msgg, v, p, S, -)$ or  
$(\msgg, v, p, S)$ 
in $\chh(x)$, there is a message in $M(x)$ which 
was added by process $p$,  having value $v$.
  Similarly, for each memory type $(\prm, v, p, S, -)$ or  
$(\prm, v, p, S)$ 
in $\chh(x)$, there is a promise in $M(x)$ which 
was added by process $p$, having value $v$.
	\item 
 The order between memory types in $\chh(x)$ 
  and the corresponding 
 messages in $M(x)$ are the same. That is, for $i < j$, the messages or promises $m, m' \in M(x)$ 
 corresponding to $\chh(x)[i]$ and $\chh(x)[j]$ are such that  
  $m.\too < m'.\too$. 
 \item the elements in the pointer set of  
a memory type $m$ in $\chh(x)$ are exactly the set of processes whose local view 
is the $\too$ stamp of the element of $M(x)$ corresponding to $m$.  
\end{itemize}
 
The base case is easy : the initial two-phases $\lhc$ state has the same local process states 
as the initial $\ps$ machine state; moreover, the invariants trivially hold, since all process  
pointers are at the same position. 

For the inductive hypothesis, assume that both invariants 
hold in a $\lhc$ run after $i$ steps. To show that they continue 
to hold good after $i+1$ steps, we have to show that for all 
$\lhc$ transitions that can be taken after $i$ steps,  they are preserved. Assume that the two phases $\lhc$ state 
at the end of $i$ steps is $(\nor, p, \lst, \lst')$. The proof for the case when 
we have a state $(\cert, p, \lst, \lst')$ after $i$ steps of the $\lhc$ run is similar.
\begin{itemize}
	\item Assume that we have the transition $\xrightarrow[p]{\rd(x,v)}$. Then $\ptr(p, \chh(x))$ 
	is updated in the resultant state, and so are $(J,R)$, Clearly, the higher order word in the resultant state satisfies both invariants  
	 since the starting state does.
	 \item Assume that we have the transition 
	$\xrightarrow[p]{\wt(x,v)}$. Then 
	we remove $p$ from the pointer set at position $i=\ptr(p, \chh(x))$. A new simple word is added at a position $>i$, or a memory type $(\msgg, v, p, \{p\})$ is added at a position $j>i$, right next to a $\#$, by moving 
	the memory type at $j$ to position $j-2$. 
		 In either case, the 
	resultant higher order word satisfies both invariants, since 
	the starting state does.  
	\item The update rule $\xrightarrow[p]{\upd(x, v_r, v_w)}$ 
	combines the above two cases, by first performing a read and then atomically the write. From the above two cases, the invariants can be seen to hold good in the higher order words in the state obtained after the transition.
	\item Consider the Promise rule. In this case, we do not remove $p$ from its pointer set, and only 
	 add the memory type $(\prm, v, p, \{\})$ ahead of $\ptr(p, \chh(x))$.  Note that 
	\textbf{Inv2} only requires that there are no memory types of the form $(\msgg, v, p, S)$
	or $(\msgg, v, p, S,-)$ ahead of $\ptr(p, \chh(x))$. Clearly, both invariants continue to hold. 
	\item Consider a fulfil rule obtained as a write. 
	In this case, $p$ is deleted from the position $\ptr(p, \chh(x))$; and the memory type $(\prm, v, p, S)$ (or $(\prm, v, p, S.-)$) is replaced with $(\msgg, v, p, S \cup \{p\})$ (or $(\msgg, v, p, S \cup \{p\})$). It is easy to see both invariants holding good.
	\item Consider the  reservation rule. This 
	does not affect the invariants since we only tag the last component of a memory type with the process making the reservation. 
	\item Consider the SC fence rule. If $\ptr(p, \chh(x)) > \ptr(g, \chh(x))$, then, 
	in the resultant word, $p$ is moved to $\ptr(g, \chh(x))$. The case when $\ptr(p, \chh(x)) < \ptr(g, \chh(x))$, 
	is handled by moving $g$ to $\ptr(p, \chh(x))$. Since this is the only change  in the resultant higher 
	order words, clearly, both invariants hold good.  
		\end{itemize} 
 Notice that the arguments above hold good in both modes $a \in \{\stdd, \cert\}$.

\subsection*{Proof of Theorem \ref{thm:eqv}}
To show the equivalence of $\psr$ and $\lhc$ we show that the transitions 
in each phase  of $\psr$ (standard, certification) is handled 
in $\lhc$ by an appropriate state $(\nor, -,-, -)$ 
or $(\cert, -, -, -)$, and conversely. The first direction we consider is from $\psr$ to $\lhc$.

To see the proof, we consider the four kinds of transitions between phases.

\begin{itemize}
	\item  Switching from \emph{certification} phase 
to the \emph{standard} phase is possible in $\psr$ only when the 
promise set of the process in the certification phase 
becomes empty.  
Any process can non deterministically 
begin the standard phase when the certification 
of one process ends successfully. 
These conditions are the simulated in $\lhc$ by
 allowing a transition from a two phases state $(\cert, p, ((J,R), \chh), ((J',R'), \chh'))$ to 
  $(\nor, q, ((J,R), \chh), ((J',R'), \chh'))$ only when there are no memory types 
 $(\prm, -, p, -)$ in $\chh'$.  
\item The switch from \emph{standard} phase to \emph{certification} phase happens in $\psr$ 
from a capped memory.  This is simulated in $\lhc$ as follows. When entering the certification phase,  $\lhc$ duplicates the higher 
order words. 
When the last memory type in any $\chh(x)$ is not tagged by the reservation of a process $q \neq p$,  the duplicated higher order word 
accounts for the capped memory, since we do not allow insertions in between during certification.  
When the last memory type in $\chh(x)$ is tagged by a reservation 
of process $q \neq p$, then we add a new simple word $\# (\msgg, -, q, \{\})$ at the end of the duplicated higher order  
word. This respects the semantics of reservation by a process $q \neq p$. 
Thus, the capped memory during certification of $\psr$ is simulated in $\lhc$ by disallowing insertions inside a higher order word, and making explicit the reservations of a process.

\item Once we are in a phase an continue in that phase,   
the proof in both directions 
is done by showing that each instruction simulated 
in $\psr$ can be simulated by the corresponding rule 
in $\lhc$ preserving the invariants, and conversely.  
\end{itemize}

The first direction from $\psr$ to $\lhc$ is done as follows.  For each transition by a process $p$ on an instruction in $\psr$, we show that we can simulate  the same instruction in $\lhc$.

\begin{enumerate}
\item Consider the read $\rd(x,v)$ rule in $\psr$. 
In $\lhc$, the read rule updates $\ptr(p,\chh(x))$ in such a way that 
$\chh(x)$ is faithful to $M(x)$ and the view $\mathsf{V}$. In case the read operation 
in $\psr$ uses a message whose $\too$ time stamp is not the local view of any process, the corresponding memory type may or may not be 
present in $\chh(x)$ due to lossiness. Considering the case when this memory type is not lost, it is used exactly in the same manner as the respective message in $\psr$. {\bf{Rule 1}} from Figure \ref{ps-program_sem} handles this.  \\

\item Consider the $\wt(x,v)$ rule in $\psr$. In $\lhc$, the write rule either appends  memory types or adds simple words 
to $\chh(x)$ in the \emph{standard} phase, and appends 
the memory type at the end of $\chh_x$ in a \emph{certification} phase
due to the capping of memory. 
$\chh(x)$ is faithful to $M(x)$ and $\mathsf{V}$ in these simulations.  Mapping memory types in $\chh(x)$ to $M(x)$, 
the relative ordering of the new memory type which gets added with respect to existing memory types in $\chh(x)$  is exactly same  as the order  
the newly added message has, with respect to others in $M(x)$ in either phase.  
{\bf{Rules 3,4}} in 
Figure \ref{ps-program_sem} handles this.  

A $\wt(x,v)$ rule can be done in $\psr$ during 
a \emph{certification} phase by splitting a promise, or 
in \emph{standard} phase for the fulfilment of a promise. These cases are handled respectively in $\lhc$ by (1) inserting a new message immediately  preceding a promise  in $\chh_x$, and (2)replacing 
a promise memory type  $(\prm, -, -, -)$ with a message memory type 
$(\msg, -, -, -)$ and updating the pointer of $p$ in each case. 
{\bf{Rule 2}} in  Figure \ref{ps-program_sem} handles these cases.\\

\item Consider the $\upd(x,v_r, v_w)$ rule in $\psr$. In $\lhc$, 
the RMW rule appends memory types to 
a simple word. The memory type corresponding to the 
message $m$ in $M(x)$ on which RMW is done, if available in $\chh(x)$, will be the   rightmost in a simple word (right to a $\#$) in the 
\emph{standard} phase, ahead of $\ptr(p, \chh_x)$, while 
in the  \emph{certification} phase, this will be the rightmost 
symbol in $\chh_x$ due to the implementation capped memory. 
 The memory type which is appended to $\#$ after moving 
 $m$ to the left of $\#$, corresponds to the new addition,  right adjacent 
 to $m$ in $M(x)$.  The append 
operation captures the adjacency of the new message added 
to $M(x)$ with respect to the one on which RMW is performed. 
This results in $\chh(x)$ being faithful to $M(x)$ and  view $\mathsf{V}$. 
{\bf{Rules 8, 9}} in Figure \ref{ps-program_sem} handle these cases. 

An $\upd(x,v_r, v_w)$ rule can be done in $\psr$ during 
a \emph{certification} phase by splitting a promise, or 
in \emph{standard} phase for the fulfilment of a promise. These cases are handled respectively in $\lhc$ by (1) inserting a new message immediately  preceding a promise  in $\chh_x$, and (2)replacing 
a promise memory type  $(\prm, -, -, -)$ with a message memory type 
$(\msg, -, -, -)$ and updating the pointer of $p$ in each case. 
{\bf{Rule 10}} in  Figure \ref{ps-program_sem} handles these cases.\\

\item Next consider the promise rule in $\psr$ by a process $p$. Promises take place only in the \emph{standard} phase. 
The simulation in $\lhc$  is similar to the write rule. A new memory type  $(\prm, v,p, \{\})$ is added to $\chh(x)$
at a position $>\ptr(p, \chh(x))$ with 
  an empty pointer set. This  corresponds to  the fact that 
 the process $p$ which makes the promise has its local view smaller than the $\too$ time stamp of the promise. Promise memory types 
 are not lost from $\chh(x)$.  {\bf{Rule 12}} in Figure \ref{ps-program_sem} handles this. 
 
 Notice that When the promise is fulfilled,  $p$ is added to the pointer set  
 of $(\prm, v, p,S)$ and the $\prm$ memory type 
is replaced with the  $\msgg$ memory type. This corresponds to removing a promise 
from the promise set of $P$. As already explained above, {\bf{rules 2, 10}} 
in Figure \ref{ps-program_sem} handle this. 
Thus, $\chh(x)$ is faithful 
also to the promise set. If there is a promise which cannot be fulfilled in $\psr$, the corresponding promise memory type 
will stay in $\chh(x)$, disallowing to reach a  
 state $(\nor, -, -, -)$ in $\lhc$. \\

 \item  Let us now look at reservations in $\psr$. These are done 
 in the \emph{standard} phase.  
The reserve rule done by a process $p$ 
 reserves a timestamp interval 
  adjacent to an existing message $m$ in $M(x)$. 
  To simulate this in $\lhc$, if the memory type corresponding to $m$ is available 
  in $\chh(x)$, then it will be the rightmost in a simple word 
  of $\chh(x)$. The reservation is done by tagging this memory type as a reservation by $p$, thereby blocking this memory type   
 from participating in any RMW.  {\bf{Rule 6}} in Figure \ref{ps-program_sem} handles this.

Similar to splitting promise intervals in a \emph{certification} phase 
in $\psr$, reservation intervals are also allowed to be split 
 in $\psr$ during certification.This can happen as part of a write 
 or an update in $\psr$. 
   To simulate this in $\lhc$, we allow a process $p$ to make use of its reservation. 
   
  \noindent
$\bullet$ {\it Splitting a reservation.}
$\ch \underset{j}{\stackrel{SR}{\hookleftarrow}} m$ is defined only if 
$\ch[j]$ is of the form $(r',v',q,S,p)$. Let $\ch'$ be the higher order word defined as $\delete(\ch,p)$.  Then,  the extended higher order   $\ch \underset{j}{\stackrel{SR}{\hookleftarrow}} m$ is defined as    $\ch'[1,j-2] \cdot    (r',v',q,S)\cdot \#  (r,v,p,\{p\},p) \cdot \ch'[j+1,|\ch|]$. Observe that the new message $ (r,v,p,\{p\},p)$ is added to the right of the position $j$ which corresponds to the slot that has been reserved by $p$. This special splitting rule will be used during the certification phase. This will allow the process $p$ to  use   the reserved slots. Recall that it is not allowed to add memory types in the middle of the higher order words (other than the reserved ones) during the certification phase.

   This is achieved by removing $p$ from its pointer set and replacing $\#(r', v', 	q, S, p)$
 in $\chh_x$ with $(r',v',q,S)\#(r,v,p, \{p\},p)$.  {\bf{Rules 5, 11}} 
 in Figure \ref{ps-program_sem} handle these. \\
 
\item Cancelling a reservation in $\psr$ frees up the reserved timestamp interval in $M(x)$. 
To simulate this in $\lhc$,  if the corresponding tagged memory type is available in $\chh(x)$, then it is unblocked from doing RMW by removing the reserve tag of $p$ from it. {\bf{Rule   7}} 
in Figure \ref{ps-program_sem} handles this. \\

 \item Finally, SC fence rules in $\psr$ updates the views 
  of the performing process to the most recent one.  
  To simulate this in $\lhc$, a dummy process $g$ simulating the global view is added. We 
    update the pointer sets of $p$ (or $g$) depending on
 which one is ahead.  {\bf{Rule 13}} in Figure \ref{ps-program_sem} 
 handles this. 
    
 \end{enumerate}
Thus, for every run that reaches a consistent state in $\psr$ with local process states $(J,R)$, there is a run in $\lhc$ that reaches
a two phases state $(\nor, -, ((J,R),\lst), -)$ following the same sequence of instructions.  Note that {\bf{rules 1- 13}}  in Figure \ref{ps-program_sem} are mutually non interfering since they apply 
to distinct rules and phases. Thus, for each rule in $\psr$ we have a unique rule in $\lhc$ from Figure \ref{ps-program_sem} which simulates that while the $\psr$ is any of the phases, \emph{standard} or 
\emph{certification}.

The converse argument from $\lhc$ to $\psr$ is similar. 
The crucial argument is the memory types  in each $\chh(x)$ form a  
subset of $M(x)$, which has all the ``necessary'' messages (promises, non empty memory types in non redundant simple words).    
Lossiness of empty memory types/redundant simple words  in $\chh(x)$ can be interpreted as messages 
which are skipped over, or which have already been used in 
$M(x)$. It is easy to see that any sequence of transitions of instructions 
in $\lhc$ can be simulated by exactly the same instruction sequence  
in $\psr$.

\subsection{Proof of Lemma \ref{computing-pre}}
\label{app:pre}
Recall that $\tt{minpre}(c)$ is defined as $\mathtt{min}(\mathtt{Pre}(\upclos{\{c\}}) \cup \upclos{\{c\}})$. 
In the following, we show  the set $\tt{minpre}(c)$ is effectively computable for any two-phases K-$\lhc$ state $c$. 
To do that, we will use  a transducer based approach. Lemma  \ref{computing-pre} is an immediate consequence of Lemma \ref{min-reg}, Lemma \ref{reg-up}, Lemma \ref{product-trans}, and Lemma 
\ref{trans-qs}.

Lemma \ref{reg-up} shows the regularity of $\upclos{\{c\}}$, 
Lemma \ref{trans-qs} and \ref{product-trans} show the regularity of 
$\mathtt{Pre}(\upclos{\{c\}})$, while Lemma \ref{min-reg} shows the effective computability of $\mathtt{min}(\mathtt{Pre}(\upclos{\{c\}}) \cup \upclos{\{c\}})$.
\smallskip

\noindent
{\bf Finite-state automata}. A finite state automaton $A$ is a tuple $A=(\Sigma_1,P,I,E,F)$, where $\Sigma_1$ is the finite input  alphabet, $P$ is a finite set of states,
$I,F\subseteq P$ are subsets of initial and final states, and $E\subseteq P\times\Sigma_1\times P$
is a finite set of transition rules. A word $u=a_1\dots a_n$ is accepted by $A$ if there is a run $p_0 \act{a_1} p_1 \act{a_2} \dots p_{n-1} \act{a_n}p_n$ such that $p_0\in I$, $p_n\in F$ and $(p_{i-1},a_i,p_i)\in E$. We use $L(A)$ to denote the set of words accepted by $A$.

\smallskip
 
 \noindent
 {\bf Regular set of two-phases K-$\lhc$-states} 
We use an encoding of two-phases K-$\lhc$ states as words over a finite alphabet, and 
use this encoding to define a regular set of two-phases K-$\lhc$ states. 
Let $\lst$ denote $((J,R), \chh))$.  Consider a two-phases K-$\lhc$ state $c=(\nor, p, \lst, \lst')$ or 
 $(\cert, p, \lst, \lst')$.   
  Recall that $(J, R)$ gives the local instruction 
   labels of all processes and the local register values. Assuming 
 we have locations $x_1, \dots, x_m$, $\chh=(\ch_{x_i})_{1 \leq i \leq m}$. 
  The state $c$ is encoded 
by the word $w=\nor \$ p \$ J \$R \$_0 \ch_{x_1}\$_1\dots \ch_{x_m}\$_m \ddagger J'\$'_0 \ch'_{x_1} \$' \ch'_{x_2} \dots 
\ch'_{x_m}\$'_m$ or 
$\cert \$ p \$ J \$R \$_0 \ch_{x_1}\$_1\dots \ch_{x_m}\$_m\ddagger J'\$'_0 \ch'_{x_1} \$' \ch'_{x_2} \dots 
\ch'_{x_m}\$'_m$
where 
$J$  defines the local state of each process, and 
the $\ddagger, \$_i, \$'_i$'s act as delimiters between the contents of the higher order words.  
$w$ is denoted $Enc(c)$.  
$w$ is a correct encoding, if, on ``decoding'' $w$, we obtain a 
unique $decode(w)=(\nor, p, \lst, \lst')$ or 
$(\cert, p, \lst, \lst')$
 where,  each $\ch_x \in (\Sigma^* \# (\Sigma \cup \Gamma))^+$  
 appearing in $\lst$ 
  satisfies the invariants $({\bf{Inv1}})$ and $({\bf{Inv2}})$.
      Given a set $R$ of two-phases K-$\lhc$ states, 
 let $Enc(R)$ represent the set of its word encodings.   
We say that a set  $R$ of two-phases K-$\lhc$ states is regular if and only if there is a finite state automaton that accepts $Enc(R)$. 
 
 \begin{lemma}
 Given a regular set $R$ of two-phases K-$\lhc$ states, we can effectively compute $\min{(R)}$. 
\label{min-reg}
 \end{lemma}
 \begin{proof}
 Let $A=(\Sigma_1,P,I,E,F)$ be the finite state automaton that accepts  $Enc(R)$. The main idea 
 to effectively compute $\min{(R)}$ 
 is to bound the size of the words accepted by $A$ that encode   minimal two-phases K-$\lhc$ states. Observe that the cycles in $A$ can be only labeled by the empty memory type.  Otherwise there will be a violation of invariant $({\bf{Inv1}})$. Now consider a word $w$ accepted by $A$. We will first construct another word $w'$ from $w$ such $ decode(w') \sqsubseteq decode(w)$ and the number of $\#e$ where 
 $e$ is an empty memory type from the subset  
  $(\msgg, -, \procset, \{\})$ of $\Sigma$ 
    or $(\msgg, -, \procset, \{\},-)$ of $\Gamma$ 
  occurring in $w'$ is polynomially bounded by the size of $A$. 
    In the following, for convenience, we use macro transitions 
  on $\#a$ rather than two separate transitions 
  on $\#$ followed by a transition for $a$. 
   
  Let us assume that $w$ is accepted by $A$ using the following  run $p_0 \act{\#a_1} p_1 \act{\#a_2} \dots p_{k-1} \act{\#a_k}p_k$. 
    Let $i_1 <i_2 < \cdots< i_b$ be the maximal sequence of indices such that $a_{i_j}$ is an empty memory type 
  $\in (\msgg, -, \procset, \{\})$ or  $(\msgg, -, \procset, \{\},-)$.
   Now if $b > |P| \cdot |\Sigma_1|$, then there are two indices $i_j$ and $i_\ell$ such that $i_j <i_\ell$, $a_{i_j}=a_{i_{\ell}}$ and $p_{i_{j}-1}=p_{i_{\ell}-1}$. Furthermore, all the symbols occurring between $i_j$ and $i_\ell$ are empty memory types 
   (from $({\bf{Inv1}})$). This means that  $p_0 \act{\#a_1} p_1 \act{\#a_2} \dots p_{i_{j}-1} \act{\#a_{i_j}} p_{i_{\ell}} \cdots p_{k-1} \act{\#a_k}p_k$ is an accepting run of  $A$ (accepting the word $w_1$). Furthermore, $ decode(w_1) \sqsubseteq decode(w)$. We can now proceed iteratively on $w_1$ in order to obtain the word $w'$ that is accepted by $A$, $ decode(w') \sqsubseteq decode(w)$, s.t.  the number of $\#e$, with $e$ an empty memory type 
   from $\Sigma \cup \Gamma$ 
   occurring in $w'$ is bounded by $|P| \cdot |\Sigma_1|$. Observe that the 
      number of $\#b$ where $b$ is a non empty memory type from $\Sigma \cup \Gamma$
   occurring in $w'$ is also bounded by $|\procset|$+K+1 : these are 
   either K promise memory types $(\prm, -, -, -)$ or 
   $(\prm, -, -, -,-)$  
      or those of the form 
   $(\msgg, -, -, S)$ or 
   $(\msgg, -, -, S,-)$
    where $S \neq \emptyset$). For the latter, we have a  bound of $|\procset|+1$. This  
   comes  from ${\bf{Inv1}}$ since each process in $\procset \cup \{g\}$ appears in a unique pointer set.  
            Thus, the number of $\#e$ where $e \in \Sigma \cup \Gamma$ 
   occurring in $w'$ is polynomially bounded by the size of $A$.
 
 Now from $w'$ we will construct another word $w''$ accepted by $A$ and such that $ decode(w'') \sqsubseteq decode(w')$ and $|w''|$ is  polynomially bounded by the size of $A$. Let $\rho:= g_0 \act{\#b_1} g_1 \act{\#b_2} \dots g_{t-1} \act{\#b_t}g_t$ be the run of $A$ accepting $w'$. Let $i_1 <i_2 < \cdots< i_r$ be the maximal sequence of indices such that $b_{i_j} \in \Sigma \cup \Gamma$.
  Observe that $r$ is polynomially bounded by the size of $A$ as we have shown previously. Assume $i_0=1$ and $i_{r+1}=t$. Now we can  iteratively remove any cycle between two indices $i_{f}$ and $i_{f+1}$ in $\rho$ that is only labeled by empty memory types  from 
  $\Sigma$
  to obtain $w''$ satisfying the previous conditions.
 \end{proof}

 \begin{lemma}
 Given a regular set $R$ of K-$\lhc$ states, the set $R \uparrow$ is also regular.
 \label{reg-up}
 \end{lemma}

\begin{proof}
 Let $A=(\Sigma_1,P,I,E,F)$ be the finite state automaton that accepts  $Enc(R)$. To construct a finite state automaton $A'$ that accept $Enc(R\uparrow)$, we proceed as follows: The automaton $A'$ is constructed by replacing  each macro transition $(p,ba,p') \in E$ labeled by the letter $a \in \Sigma$, $b \neq \#$ 
 by the following macro-transition $(p, e^* \cdot ba \cdot e^* ,p')$ in $A'$, where $e$ is over the empty memory types of $\Sigma$.
   Furthermore, any macro transition  $(p,\#a,p') \in E$ labeled by the letter $a \in \Sigma \cup \Gamma$ 
   is replaced in $A'$ by the macro-transition $(p,  \#a \cdot (w \#b)^*,
p')$, where $w \in \Sigma^*$ is over the empty memory types of $\Sigma$ and $b \in \Sigma \cup \Gamma$ is an empty memory type in $\Sigma \cup \Gamma$.  We can also have a loop on 
empty memory types of $\Sigma$ on the initial state.  
 Observe that any macro-transition can be easily translated to a sequence of simple transitions  by using extra-intermediary states.
\end{proof}

\smallskip

\noindent{\bf {Rational Transducers}}.  A rational transducer $T$ is a non-deterministic finite state automaton which  outputs words on each transition. 
Formally, a \emph{rational transducer} is a tuple $T=(\Sigma_1,\Sigma_2,Q,I,E,\eta,F)$, where $\Sigma_1,\Sigma_2$ are finite input and output alphabets, $Q$ is a finite set of states,
$I,F\subseteq Q$ are subsets of initial and final states, $E\subseteq Q\times\Sigma_1\times Q$
is a finite set of transition rules, and $\eta:E \rightarrow 2^{\Sigma_2^*}$ is a 
function 
specifying a regular 
language of partial outputs for each 
transition rule (i.e., $\eta(e)$ is a regular language for all $e \in E$).  
The relation defined by $T$ contains 
pairs $(u,v)$ of input and output words,
where $u=a_1\dots a_n$ and $v=v_1\dots v_n$,
for which there is a run
$q_0 \act{a_1 \:|\: v_1} q_1 \act{a_2 \:|\: v_2} \dots q_{n-1} \act{a_n \:|\: v_{n}}q_n$
such that
$q_0\in I$, $q_n\in F$, $(q_{i-1},a_i,q_i)\in E$, 
$v_i\in \eta(q_{i-1},a_i,q_i)$.
 The set of pairs  $(u,v)$ defined by $T$  is denoted $L(T)$.

 \begin{lemma}
 Given a regular language $R$ (described by a finite-state automaton), we can easily compute a finite state automaton $A$ such that $L(A)= \{u \,|\, (u,v) \in L(T) \,\wedge\, v \in R\}$.
 \label{product-trans}
 \end{lemma}
 
 \begin{proof}
 Trivial.
 \end{proof}
 
\begin{lemma}
It is possible to construct a transducer $T$ that accepts any pair $(Enc(s),Enc(s'))$, with $s$ and $s'$ are two two-phases K-$\lhc$-states, such that $s'$ is reachable from $s$ in one step. 
\label{trans-qs}
\end{lemma}

\begin{proof}
 Observe that the class of rational transducers are closed under union and therefore it is sufficient to construct the transducer $T$ for each transition rule. Furthermore, we always assume that the input and output tape of the transducer $T$ satisfy the two invariants  $({\bf{Inv1}})$ and $({\bf{Inv2}})$ (these can be easily specified as a regular language). The proof  is about simulating the rules in the transition system in $\lhc$ as defined in Section \ref{sec:formal}. We reproduce the rules for easy reference. 
 \subsection*{The global transition rules in $\lhc$}
 
 Given  $\Ss=(\pi, p, \lst_{\nor}, \lst_{\cert})$ and $\Ss'=(\pi', p', \lst'_{\nor}, \lst'_{\cert})$, we have $\Ss \rightarrow \Ss'$ iff one of the following cases hold:
  \begin{itemize}
 \item[(a)] {\bf During  the standard phase.} $\pi=\pi'=\nor$, $p=p'$, $\lst_{\cert}=\lst'_{\cert}$ and $\lst_{\nor} \xrightarrow[p]{\nor}\lst'_{\nor}$. This corresponds to a simulation of a standard step  of the process $p$.
 
  \item[(b)] {\bf During the certification phase.} $\pi=\pi'=\cert$, $p=p'$, $\lst_{\nor}=\lst'_{\nor}$ and $\lst_{\cert} \xrightarrow[p]{\cert}\lst'_{\cert}$. This corresponds to a simulation of a certification step  of the process $p$.
 
    \item[(c)] {\bf From the standard  phase to the certification phase.} $\pi=\nor$, $\pi'=\cert$, $p=p'$, $\lst_{\nor}=\lst'_{\nor}= (({\sf J}, {\sf R}), \chh)$,  and $\lst'_{\cert}$ is of the form $(({\sf J}, {\sf R}), \chh')$ where for every $x \in \varset$, $\chh'(x)=\chh(x) \# (\msgg,v,q,\{\})$ if  $\chh(x)$ is of the form $w \cdot \# (-,v,-,-,q)$ with $q \neq p$, and $\chh'(x)=\chh(x)$ otherwise. This corresponds to the copying of the standard $\lhc$  state to the certification $\lhc$ state in order to check if the set of promises made by the process $p$ can be fulfilled. The  higher order word $\chh'(x)$ (at the beginning of the certification phase) is almost the same as $\chh(x)$ (at  the end of the standard phase) except when the rightmost memory type $(-,v,-,-,q)$ of $\chh(x)$  is tagged by a reservation of a process $q \neq p$. In that case, we append the memory type $ (\msgg,v,q,\{\})$  at the end of $\chh(x)$ to obtain $\chh'(x)$. Note that this is in accordance 
     to the definition of capping memory before going into certification: to cite, (item 2 in capped memory of  \cite{promising2}), a cap message 
     is added for each location unless it is a reservation made by the process going in for certification.

  \subsection*{Copying $\chh$ to $\chh'$  symbol by symbol}   
  We can implement copying of $\chh$ to $\chh'$ by copying symbol by symbol 
  as follows. Consider any $\chh(x)$.  Let $\chh=(a_x W_x)_{x \in \varset}$ where $\chh(x)=a_xW_x \in (\Sigma^* \# (\Sigma \cup \Gamma))^*$, $|a_x|=1$. Define the function $\mathsf{copy}$ on the two phases $\lhc$ state $(\nor, p, ((J,R), (a_xW_x)_{x \in \varset}), -)$, and then recursively  to subsequent states 
  until we end up in $(\cert, p, ((J,R),\chh), ((J,R), \chh))$.

      The $\mathsf{copy}$ function is defined recursively as follows. 
  \begin{itemize}
  \item[(Base)] $\mathsf{copy}(\nor, p, ((J,R), (a_xW_x)_{x \in \varset}), -)= 
  (cc, p, ((J,R), (\overline{a}_x W_x)_{x \in \varset}), ((J,R),(a_x)_{x \in \varset}))$. This is copying the  first 
  symbol of each $\chh(x)$. $cc$ is an intermediate phase used only in copying. Notice that the over lined symbol shows the progress of copying, one symbol each time. 
  \item[(Inter)] Next, we copy subsequent symbols.  $\mathsf{copy}(cc, p, ((J,R), (\overline{\alpha} a_x U_x)_{x \in \varset}), ((J, R), (W_x)_{x \in \varset}))$ is defined as  
  $(cc, p, ((J,R), (\overline{\alpha a}_x U_x)_{x \in \varset}), ((J, R), (W_xa_x)_{x \in \varset}))$. 
  \item[(Last)] Finally, when all higher order words have been copied, we move from $cc$ to $\cert$. When a higher word has been completely copied, it has the form $\overline{\alpha}$, where $\alpha \in (\Sigma^* \# \Gamma)^+$. Then we define 
    $\mathsf{copy}(cc, p, ((J,R), (\overline{\alpha}_x)_{x \in \varset}), ((J, R), (W_x)_{x \in \varset}))$  as $(\cert, p, ((J,R), (\alpha_x)_{x \in \varset}), ((J, R), (W_x)_{x \in \varset}))$, by removing the overline, and having the phase $\cert$. 	
 \end{itemize} 
     If the last symbol $a_x$ in $\chh(x)$ is of the form $(-,v,-,-,q)$,      for $q \neq p$, then $\mathsf{copy}$ appends 
     $a_x\#(\msgg, v, q, \{\})$ instead of just $a_x$ in (Inter).

   \item[(d)] {\bf From the certification phase to standard phase.} $\pi=\cert$, $\pi'=\nor$,  $\lst_{\nor}=\lst'_{\nor}$,   $\lst_{\cert}=\lst'_{\cert}$, and  $\lst_{\cert}$ is of the form $(({\sf J}, {\sf R}), \chh)$ with $\chh(x)$  does not contain any memory type of the form $(\prm,-,p,-)$/$(\prm,-,p,-,-)$ for all $x \in \varset$ (i.e., all  promises made by  $p$ are  fulfilled).

 \end{itemize}

\paragraph{Description of the Transducer} 
   We consider 4 cases based on the 4 cases we have 
   in the transition rules (a)-(d) as above. 
\begin{enumerate}
	\item We first consider the case when $s$ and $s'$ have the same phase ($\nor$ or $\cert$).  If the location involved in the instruction is $x_i$, then 
   the transducer copies all $\ch_{x_j}$, $j \neq i$ as is. 
   For $\ch_{x_i}$,    if the phase we have 
 in $Enc(s)$ is $\nor$, then the transducer copies $\ddagger$ 
 as well as all symbols after that in the output, while 
 if the phase we have in $Enc(s)$ is $\cert$, the the transducer copies $\ddagger$ 
 as well as all symbols before that in the output.  This is common 
 to all items below and we will not mention it separately.

\begin{enumerate}
\item 	Consider  a  ${\tt{Read}}$ instruction of the form $\lambda: \$r=x_i$ of the process $p$. Then the transducer will first guess the value $v$ that will be read and update the local states of the processes (as an output). The only change that the transducer will do concerns the $i$-th higher order word $\ch_{x_i}$.
 For each symbol that the transducer reads on the input tape of $\ch_{x_i}$ before $\ddagger$, it outputs the same symbol. Once the symbol pointed by the process $p$ is read on the input tape, the transducer will check the value of each symbol read on the input tape and if it  corresponds to $v$, the transducer will non-deterministically add $p$ to its  pointer set, otherwise it will output the same read symbol (while removing $p$ from its pointer set, which has bee read, if needed).  

\item 	Consider  a  ${\tt{Write}}$ instruction of the form $\lambda: x_i= \$r$ of the process $p$. Then the transducer will first update the local states of the processes (as an  output). The only change that the transducer will do concerns the $i$-th higher order word $\ch_{x_i}$. 
 For each symbol that the transducer reads on the input tape of $\ch_{x_i}$,  it outputs the same symbol. Once the symbol having $p$ in its pointer set is read on the input tape, the transducer will  output the same read symbol (while removing $p$ from the pointer set). When the transducer reads a symbol after $\#$,  it can decide to output the new message corresponding to the write instruction and after that, go on by outputting any read symbol.

\item 	The case of RMW is very similar to the case of a ${\tt{Write}}$ instruction of the process $p$. 
\item The case of a promise rule is similar to the ${\tt{Write}}$. 
The main difference is that  when the transducer reads the symbol pointed by the process $p$  on the input tape, the transducer will  output the same read symbol (without removing $p$ from the pointer set). When the transducer reads a symbol right after $\#$, 
it can decide to output the new promise message, such that the pointer set is empty. After that, it goes on by outputting any read symbol.

\item The case of a reservation rule is similar to RMW. 
\item The case of a cancel rule by a process $p$ is as follows. 
The transducer reads on symbols and outputs the same, till it finds the symbol $(-, -, -, -, p)$. 
 On reading this, it outputs $\epsilon$. After that, it goes on by outputting any read symbol.
\item The case of a fulfil rule is as follows. The transducer 
outputs what it reads till it finds a symbol having $p$ in its pointer set. 
It outputs the same symbol removing $p$ from the pointer set. Then it continues outputting the read symbol till it reads a symbol 
$(\prm, v, p, S)$. It outputs $(\msgg, v, p, S \cup \{p\})$ by adding $p$ 
to the pointer set. After that, it goes on by outputting any read symbol.

\item Consider a  ${\fence}$ instruction. In this case the transducer will output  any read symbol except the ones that have $g$ or ${p}$ in its pointer set.  If $p$ and $g$ are in the same pointer set, then  the transducer will continue outputting any read symbol. If the transducer reads the first encountered symbol that contains only $p$ or $g$ in its pointer set, then the transducer will output the same symbol without the pointer set containing either $g$ or $p$. Once  the transducer reads the second encountered symbol whose pointer set contains only $p$ or $g$ then the transducer will output the same symbol with the pointer set containing both $g$ and $p$. This is done for each $\ch_{x_i}$.
\end{enumerate}
	
\item If the phase in $Enc(s)$ is $\cert$ and that of $Enc(s')$ is $\nor$, 
then the transducer simply replaces $\cert$ by $\nor$, and 
the process $p$ by any process $q$, and copies the rest as is in the output. 
\item If the phase in $Enc(s)$ is $\nor$ and that of $Enc(s')$ is $\cert$, 
then the transducer implements the $\mathsf{copy}$ function described 
above. Each $\mathsf{copy}$ is implemented by a transducer, and 
the final result is obtained by composing all these transducers. 
Note that rational transducers are closed under composition, 
so it is possible to obtain one rational transducer 
 that achieves the effect of all the $\mathsf{copy}$ functions, starting 
 with the $\nor$ phase and ending in the $\cert$ phase. 
  Note that this is easily done, since 
in each step, the transducer progressively marks a symbol before $\ddagger$ with overline, and copies the same at the end.  

\end{enumerate}

\end{proof}

\subsection{Proof of Lemma \ref{lem:mon}}
\label{app:mono}

Consider K-$\lhc$ states $c_1, c_2$ s.t. $c_1 \rightarrow c_2$, 
	and let $c_3$ be a state s.t. $c_1 \sqsubseteq c_3$. 
	We make a case analysis based on the transition chosen. 
	
Let $c_1=(\nor, p, ((J_1,R_1), \chh_1), ((J_2,R_2), \chh_2))$, 
$c_2=(\pi, q, ((J_3,R_3), \chh_3), ((J_4,R_4), \chh_4))$, 
$c_3=(\nor, p, ((J_1,R_1), \chh_5), ((J_2,R_2), \chh_6))$, and  
$c_4=(\pi, q, ((J_7,R_7), \chh_7), (J_8, R_8), \chh_8))$. 
The case when $c_1=(\cert, p, -, -)$ is similar to the case we discuss here.   

\begin{enumerate}
\item  Consider the transition $c_1 \xrightarrow[\proc]{{\lambda:  \$r=x}} c_2$
by a read instruction  $\$r=x$ in process $p$. 
Then $\exists k \leq j$, $k=\ptr(p,\chh_1(x))$, and 
the memory type at $\chh_1(x)[j]$ has the form $(-, v, -, S)$, $v=R(\$r)$. 
$\chh_3(x)$ is obtained by updating $\ptr(p,\chh_1(x))$ to $j$, 
so that $p$ is in the pointer set $S$. Since 
 $c_1 \sqsubseteq c_3$,  there is an increasing function 
 $f$ from the positions of $\chh_1(x)$ 
 to that of $\chh_5(x)$ such that 
 $f(k) \leq f(j)$, $\ptr(p,\chh_5(x))=f(k)$ in $\chh_5(x)$ and 
 the memory type at $f(j)$ has the form 
 $(-,v,-,S'')$, $v=R(\$r)$. Indeed, one can update $\ptr(p,\chh_5(x))$ 
 to $f(j)$, obtaining a state $c_4$ from $c_3$. The local process 
 states of $c_4$ is same as that of $c_2$.    
  All higher order words $\chh_5(y)$, $y \neq x $ of $c_3$ remain unchanged in $c_4$ (and 
  all higher order words $\chh_1(y)$, $y \neq x $ of $c_1$ remain unchanged in $c_2$), hence 
  the $\sqsubseteq$ relation holds for these higher order words in $c_2, c_4$. 
      The same   function $f$ between positions of $\chh_1(x)$ and 
      $\chh_5(x)$ can be used on positions of $\chh_3(x)$   
    of $c_2$ and  $\chh_7(x)$ of $c_4$ to see that $c_2 \sqsubseteq c_4$ and $c_3 \xrightarrow[\proc]{\lambda: \$r=x} c_4$.

\item Consider the transition $c_1 \xrightarrow[\proc]{{\lambda: x = \$r}} c_2$. Then, there is a position $k$ in $\chh_1(x)$ such that 
$k=\ptr(p,\chh_1(x))$. Let the memory type at $\chh(x)[k]$ be $(-, v_1, -, S_1 \cup \{p\})$. 
After the transition, we obtain 
$\chh_3(x)$ such that $\ptr(p,\chh_3(x)) = j-1 > k$. There are 2 possibilities. 
\begin{itemize}
	\item[(a)]  $j-1, j$ form  the positions of the 2 symbols $\#, (\msgg, R(\$r), p, \{p\})$ in the newly added 
simple word in $\chh_1(x)$.  Figure \ref{wt-order} depicts this case. 
Notice that in $\chh_1(x)$, $k=\ptr(p, \chh_1(x))$, and positions $j-3, j-2$ represent 
the last two positions of a simple word. The new simple word 
is added right after this in $\chh_3(x)$, at positions $j-1, j$.

\begin{figure}[h]
\includegraphics[scale=.13]{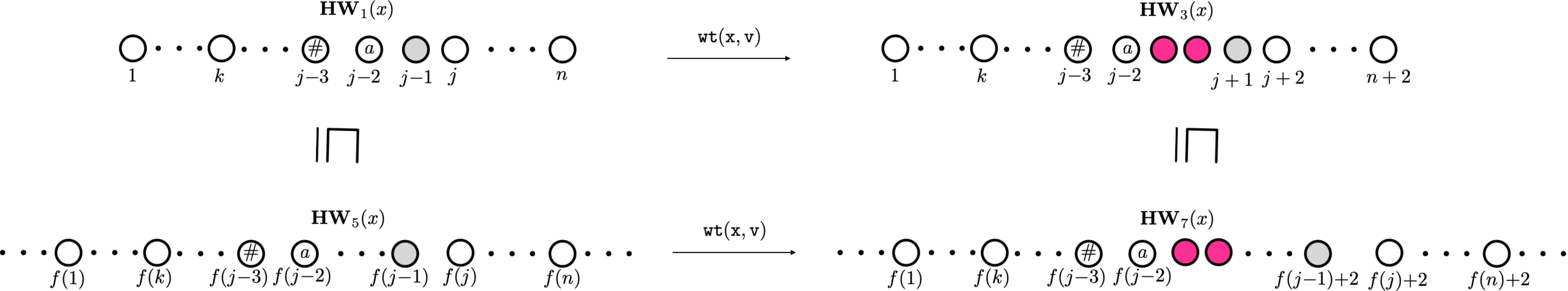}
\caption{The higher order words in $c_1, c_2, c_3, c_4$ in case(a). The two pink positions correspond to the newly added simple word. 
The positions $j-3, j-2$ have $\#$ and $a \in \Sigma \cup \Gamma$ 
denoting the end of a simple word in $\chh_1(x)$, so that a new simple word can be inserted right after. The position $k$ 
in $\chh_1(x)$ is $\ptr(p, \chh_1(x))$. 
$\chh_1(x) \sqsubseteq \chh_5(x)$ witnessed by the increasing function $f$.
$\chh_3(x), \chh_7(x) $ respectively are  obtained from $\chh_1(x), \chh_5(x)$ by the $\wt(x,v)$ transition.   
}	
\label{wt-order}
\end{figure}

 Since $c_1 \sqsubseteq c_3$, let $f$ be an 
increasing function from the positions of $\chh_1(x)$ to those 
of $\chh_5(x)$. $\chh_7(x)$ is obtained from $\chh_5(x)$ by inserting the new simple word right after position $f(j-2)$, at positions $f(j-2)+1, f(j-2)+2$. The position $f(j-1)$ 
in $\chh_5(x)$ is shifted to the right by two positions in $\chh_7(x)$. Thus, we can define an increasing function from positions 
of $\chh_3(x)$ and $\chh_7(x)$ as follows.

\begin{itemize}
\item For $i \in \{1, \dots, j-2\}$, $g(i)=f(i)$,
\item $g(j-1)=f(j-2)+1, g(j)=f(j-2)+2$, (note that $g(j-1),g(j)$ are the  two new positions in $\chh_7(x)$ corresponding to the new positions $j-1, j$ in $\chh_3(x)$), 
\item For $i \in \{j+1, \dots, n+2\}$, $g(i)=f(i-2)+2$
\end{itemize}

It is easy to see that $g$ is an increasing function between 
the positions of $\chh_3(x)$ and $\chh_7(x)$ : we know that $f(j-2) < f(j-1)$. Hence, $g(j) =f(j-2)+2 < f(j-1)+2 = g(j+1)$. 
This also gives $\chh_3(x) \sqsubseteq \chh_7(x)$.

\item[(b)]  $j-1$ is the position obtained by appending to a simple word  in $\chh_1(x)$.

\begin{figure}[h]
\includegraphics[scale=.13]{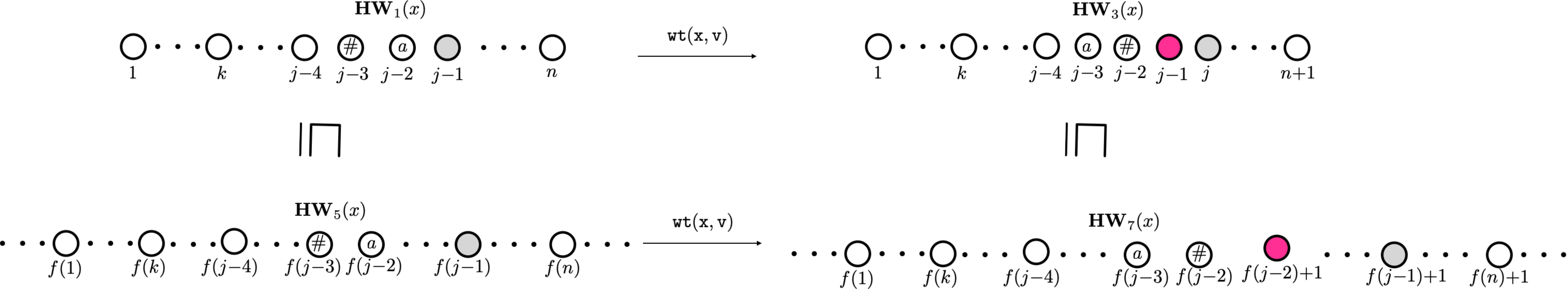}
\caption{The higher order words in $c_1, c_2, c_3, c_4$ in case(b). The  pink position in $\chh_3(x)$ corresponds to the newly added memory type, right after $\#$ at position $j-3$ in $\chh_1(x)$.  
$a \in \Sigma$ at position $j-2$ in $\chh_1(x)$ is shifted to the 
left of $\#$ in $\chh_3(x)$. The position $k$ 
in $\chh_1(x)$ is $\ptr(p, \chh_1(x))$. 
$\chh_1(x) \sqsubseteq \chh_5(x)$ witnessed by the increasing function $f$.
$\chh_3(x), \chh_7(x) $ respectively are  obtained from $\chh_1(x), \chh_5(x)$ by the $\wt(x,v)$ transition.   
}	
\label{wt-orderb}
\end{figure}

\end{itemize}

Figure \ref{wt-orderb} illustrates this case. $\chh_1(x) \sqsubseteq \chh_5(x)$ is witnessed 
by the increasing function $f$. The new memory type is added at position $f(j-2)+1$ (right next to $\#$), and all subsequent symbols 
are shifted right by one position. 
It is easy to see that $\chh_7(x)$ is obtained from $\chh_5(x)$ 
by the $\wt(x,v)$ transition. The increasing function $g$ from the 
positions of $\chh_3(x)$ to that of $\chh_7(x)$ is defined as follows.

\begin{itemize}
	\item For $i \in \{1, \dots, j-2\}, g(i)=f(i)$,
	\item $g(j-1)=f(j-2)+1$,
	\item For $i \in \{j, \dots, n+1\}$, $g(i)=f(i-1)+1$
\end{itemize}

Notice that $g$ is an increasing function: $g(j-2)=f(j-2) < f(j-2)+1=g(j-1)$, 
$g(j)=f(j-1)+1 > f(j-2)+1=g(j-1)$, and the same relationship holds for subsequent indices.

\item The case of $c_1 \xrightarrow[\proc]{\lambda: \arw(x_i, \$r_1, \$r_2)} c_2$  is similar to the write.
\item The case of a promise rule is exactly same as the write rule, as 
far as monotonicity is concerned.
\item The case of promise fulfilment is trivial for monotonicity, 
since we only shift the pointer of $p$, and update $\prm$ to $\msgg$ 
in the memory type.
\item The case of reservation follows exactly like case (b) 
of the write rule. 
\item The case of cancellation is trivial for monotonicity
since the operation does not change the length of the word. 
\item The case of $\xrightarrow[\proc]{\fence}$ is trivial
by using the observation that   
the relative ordering of the pointers $p$ and  $g$ are same 
 in $\chh_1(x)$ and $\chh_5(x)$.  $\chh_3$ and $\chh_7$ are obtained respectively 
 by moving the pointers of $p, g$ to the rightmost one (whichever it is). So the same 
 increasing function that was used for $\chh_1 \sqsubseteq \chh_5$ will work for 
 $\chh_3 \sqsubseteq \chh_7$.
 \end{enumerate}

\newpage	

\section{Source to Source  Translation and Proof of Correctness}

\subsection{Intuition for the Translation}

\paragraph{$2K$ Timestamps} We bound the number of essential events by $K$. Why do $2K$ timestamps suffice?. Intuitively timestamps are used to determine relative order between the events. We track timestamps of the view-switching messages (messages read by other processes), promises and reservations. For each view-switch there are two timestamps of consequence. The timestamp of the reading process before the read and the timestamp of the message to be read. Hence for each view switch, the comparison operation requires us to maintain two timestamps. For a promises (reservation) we maintain the timestamp of the promise (reservation). We do not explicitly store timestamps of messages that will not view switch. These messages however may be read by the same process that generated them. We keep track of whether the latest write can be read by the same process by using some thread-local state.

\paragraph{$K+n$ Contexts} It suffices to have $K+n$ contexts since we can run the processes in the order in which they generate view-switching messages. In each context, the process only depends on the essential messages generated in previous contexts. If this were not the case we would get a deadlock. We require $n$ additional contexts to initiallize each process.

\subsection{Glossary of Global and Local Variables used in the SC Program}
\label{app:ctc}

We first give a glossary of all the variables used in the code. The list contains variables global to all processes or local to a process. A small description of their role is also mentioned, which serve as invariants.  

\newcommand{\appvar}[1]{\lstinline[style=customc,basicstyle=\normalsize\ttfamily]{#1}}

\begin{enumerate}
  \item \appvar{numEE} : a global variable, initialized to 0, keeps track of the number of essential events (promises, reservations and view switches) so far. Each time an essenial event occurs, \appvar{numEE} is incremented.
  \item \appvar{numContexts} : a global variable, initialized to 0, keeps track of the number of context  switches so far. This is used in the translation to SC.
  \item \appvar{view[x].v} : a local variable, stores the value of  $x \in \varset$ in the local view of the process
  \item \appvar{view[x].t} : local variable, stores the time stamp $\in \mathsf{Time}$ of  $x \in \varset$ in the local view of the process.
  \item \appvar{view[x].l} : local variable, boolean, which is set to true when \appvar{view[x].t} is a valid timestamp, and can be used in comparisons with timestamps of other messages. 
  \item \appvar{view[x].f} : local variable, boolean. A true value indicates that \appvar{view[x].v} is recent, and can be used for reading locally.
  \item \appvar{view[x].u} : local variable, boolean. A true value indicates that the sequence of events starting from the one that resulted in the timestamp \appvar{view[x].t} till the most recent, form a chain of $\arw$ operations on $x$. Whenever a write is published, \appvar{view[x].u} is set to true. \appvar{view[x].u} is set to false on an unpublished write. On a sequence of $\arw$ operations,  \appvar{view[x].u} is left unchanged.
  \item \appvar{checkMode} : local variable, boolean. Set to true when the process is in certification phase, which means the process is making and certifying promises. 
  \item \appvar{liveChain[x]} : local variable, for each $x \in \varset$, boolean. Can be true only when  \appvar{checkMode} is true.  A true value represents that the last write done while the process is in certification phase is not a  published promise message. 
  \item \appvar{extView[x]} : local variable, for each $x \in \varset$, boolean. A true value represents that the local value \appvar{view[x].v} of the process comes from a message generated external to the certification phase.
  \item \appvar{blockPromise[x]} : a global boolean array, which for each $x \in \varset$ stores whether promises should be blocked on variable $x$. This is used in the case of $\ra$ writes when we cannot have promises on the same variable later (refer to \ps, $\ra$ accesses).
  \item \appvar{avail[x][t]} : for each $x \in \varset$, a global boolean array of length $2K+1$ corresponding to  the $2K+1$ time stamps, checks availability of a time stamp on a fresh write.
  \item \appvar{usedReservations[x][t]} : denotes whether the reservation on variable $x$ with timestamp $t$ has been used by the process during the certification check. If this not true, the reservation will be cancelled.
  \item \appvar{reserv[x][t]} : denotes whether the reservation following timestamp t on variable $x$ has been claimed, and if so which process has claimed it. 
  \item \appvar{upd[x][t]}  : for each $x \in \varset$, a global boolean array of length 
  $2K+1$ corresponding to  the $2K+1$ time stamps, checks whether a certain timestamp has been used to read in a $\arw$. 
  \item \appvar{globalTimeMap[x]}
   : global variable, for each $x \in \varset$, stores 
  a time stamp $\in \mathsf{Time}$. This is used for simulating SC Fences where this functions as the $G$ timemap from \ps.
  \item \appvar{messageStore} : This is an array of messages, where each message is of type $\mathsf{Message}$ as described in the main paper. The length of the array is $K$, the bound on the number of promises + view switches. 
  \item \appvar{messagesUsed} : a number from 0 to $K$ which keeps track of the number of populated messages in $\mathit{messageStore}$. 
  \item \appvar{messageNum} : a number from  0 to $K$ which chooses a number from the available free cells in \appvar{messageStore}.
\end{enumerate}

In addition, the \appvar{message} object stores the following data:

\begin{enumerate}
  \item \appvar{mess.var} is the shared variable on which the message has been generated
  \item \appvar{mess.t[x]} stores for each $x \in \varset$ the timestamp of $x$ in the view object stored in the message
  \item \appvar{mess.l[x]} stores for each variable $x \in \varset$, a boolean signifying whether the corresponding timestamp stored in \appvar{mess.t[x]} was one of the exact timestamps $\in \{0 ... K\}$ or an abstract timestamp.
  \item \appvar{mess.val} stores the value of the message
  \item \appvar{mess.flag} stores the promise state of the message, that is whether (1) it is has been fulfilled/is not a promise (2) if it is a promise then the process that it belongs to. \appvar{mess.flag} takes values from \appvar{0, -1, PIDs}. If it is a simple message (not a promise), \appvar{mess.flag = 0}. If it is a promise, \appvar{mess.flag} is set to the \appvar{PID} of the process which has made the promise. \appvar{mess.flag} is set to \appvar{-1} when the process has temporarily certified it in the current certification phase but will be reset tp \appvar{PID} after exiting the certification phase.  
\end{enumerate}
Next we discuss the context switching modules.
\subsection{Context Switching Modules}

\paragraph{$\mathsf{CSI}$ Context-Switch-In}
The $\mathsf{CSI}$ module  switches the process into context by setting \appvar{active} to true and incrementing \appvar{numContexts}. Finally we check \appvar{numContexts} does not exceed the context switch bound.
 
\begin{lstlisting}[caption=$\mathsf{CSI}$,style=customc,mathescape]
if (!active){
	atomic_begin();
	active = true;
	numContexts++;
	assume(numContexts <= K + n);
}	
\end{lstlisting}

\paragraph{$\mathsf{CSO}$ Context-Switch-Out}
The $\mathsf{CSO}$ module has two functions- (1) moving the process from normal to check mode and (2) switching the process out of context. When a process enters the CSO block, with \appvar{checkMode} set to false, it enters the `if' branch on line 2, sets \appvar{checkMode} to true and saves the return label (of the current instruction pointer) in \appvar{retAddr} and saves the process state before entering check mode (lines 9-10). This ensures that the process returns to the current instruction after the consistency check. Now after the consistency check phase the process switches out of context. At this point, \appvar{checkMode} is \appvar{true}, and hence the process enters the `else' branch on line 13. Consequently, we check whether there are no outstanding uncertified promises for the process (line 15). All the promises that have been certified are reset to belong to the process by setting \appvar{mess.flag} to the PID (lines 16-18). Then it is checked that there are no uncertified splitting insertions, by ensuring that \appvar{liveChain[x]} is not true (lines 20-22). Finally we check for unused reservations during ceritification and cancel them (lines 23-30). Once these checks for cnsistent configuration are complete, we reload the saved state from before the consistency check phase and reload the return address from \appvar{retAddr}. Then we move control to the instruction label in \appvar{retAddr}. After returning control to \appvar{label}, we set \appvar{checkMode} and \appvar{active} to false and exit context.

\begin{lstlisting}[caption=$\mathsf{CSO}$,multicols=2,style=customc,mathescape]
if (*){
	if (!checkMode){
		if (!active){
			atomic_begin();
			active = true;
			numContexts++;
			assume(numContexts <= K+n);
		}
		checkMode = true;
		retAddr = label_i;
		saveState(PID);
	}
	else {
		for (mess in messageStore){
			assume(mess.flag != PID);
			if (m.flag == -1){
				m.flag = PID;
			}
		}
		for (x in $\varset$){
			assume(!liveChain[x]);
		}
		for (x in $\varset$ and t in $\mathsf{Time}$){
			if (reserv[x][t] == PID){
				if (!usedReservation[x][t]){
					reserv[x][p] = 0;
					upd[x][t] = 1;
				}
			}
		}
		loadState(PID);
		gotoLabel(retAddr);
		label_i:
		checkMode = false;
		active = false;
		atomic_end();
	}
}
\end{lstlisting}

\paragraph{\appvar{loadState} and \appvar{saveState} subroutines}
The \appvar{saveState} subroutine copies the local state of the calling process and the global state into a what we refer to as `copy' variables. We note that it does not however copy \appvar{numEE}, \appvar{reserv[x][t]} and contents of \appvar{messageStore}. The reason for this being, the promises the process makes in check mode are retained even after exiting check mode is made false. Hence the increments made to \appvar{numEE} and the messages added to \appvar{messageStore} should be maintained even after exiting check mode. This is even true for reservations, which are marked in \appvar{reserv[x][t]}, which are maintained evef after the process exits check mode.

Analogously in \appvar{loadState}, we load the contents of the (saved) `copy variables' into their original counterparts. Another subtle point to be noted is that when the process publishes a message (as a promise) when \appvar{checkMode} is true, we also update the `copy' variables corresponding to \appvar{avail[x][t]}. This is done so that when the process returns to normal mode, the changes are reflected in their original counterparts (which is essential since promise messages are maintained beyond the time \appvar{checkMode} is false and hence their timestamps must be unavailable).

\subsection{Reads}

We provide the translation codes for reads of both access types, $\rlx$ and $\ra$. We will first explain with respect to $\rlx$ access reads. 

\paragraph{$\rlx$ reads}

The read can be one of two types, view switching, in which a message from \appvar{messageStore} is acquired or a non view-switching (local) read. We guess non-deterministically, one amongst these.

In case of a local read (line 2), the process checks that the local value is usable (line 3) by checking \appvar{view[x].f} which denotes whether \appvar{view[x].v} is a valid value which can be read. It then loads its local value \appvar{view[x].v} into $\reg$. The local value may become unusable if the process crosses an SC-fence which increases its \appvar{view[x].t} (see $\fence$).

In the case of a view-switching read (line 6), we check that we have not reached the essential-event bound $K$ (line 7). We ensure that \appvar{liveChain[x]} is false before the read in order to forbid additive insertions when checking consistency. Recall from the \appvar{liveChain} invariant that \appvar{liveChain[x]} is true only when the process is in certification mode and the last write on $x$ was neither published as a promise message nor was it certified with a reservation. Reading a message from the memory when $\mathit{liveChain[x]}$ is true implies additive insertion during certification, as illustrated by the following example.

\textit{liveChain}
Assume the process is in the promise certification mode, with $\mathit{view[x].t}$ set to $t_1$, and let the first write use a timestamp $t_2 > t_1$ with the message not published as promise, with $\mathit{liveChain}[x]$ as true. Now the instruction a:=x uses a message in the memory with a timestamp $t_3 \geq  t_2$.

\setlength{\columnsep}{7pt}
\setlength{\intextsep}{7pt}
\begin{wrapfigure}{r}{3cm}
\fcolorbox{black}{yellow!20}{
\begin{tabular}{ll} 
x:=1; & // $t_2$  \\
a:=x; & // $t_3$ \\
x:=2; & // $t_3+1$\\
\end{tabular}}
\end{wrapfigure}

If the next write certifies a promise message, the interval in the message will be $t_3+1$, since \appvar{liveChain[x]} is true. This results in two writes during the certification, with non-adjacent timestamps  $t_2, t_3+1$, with \textit{only} the latter being promised. This behaviour is forbidded in \ps due to capped memories. Notice that if the earlier write also resulted in a promise message then we do not have additive insertion (since both are promised) and the read with timestamp $t_2$ is allowed since \appvar{liveChain[x]} is false.

Finally a new message is fetched from \appvar{messageStore} with a larger timestamp that the one in the current view (lines 8-11), the process view is updated to include that new message. Whenever a process makes a global read during check mode, it must reads from a message which has been created outside its current certification phase. Hence, \appvar{extView[x]} will be set to true (see \appvar{extView} invariant in the glossary).  

\begin{lstlisting}[caption=$\texttt{read}_\rlx$,multicols=2,style=customc]
// local read 
if(*){
	ASSUME(view[x].f);
}
// (non-local) view-switching read 
else {
	/* ensure that there is no active chains on x */
	ASSUME(!liveChain[x]);
	ASSUME(numEE < K);
	messNum = nondet(0, messageUsed-1);
	mess = messageStore[messNum];
	ASSUME(mess.var == x);
	ASSUME(mess.t[x] > view[x].t or (mess.t[x] = view[x].t and view[x].l == true));
	
	// merge views on x
	view[x].t = mess.t[x];
	view[x].l = true;
	view[x].v = message.val;
	
	extView[x] = true;	
	numEE++;
}
val($r) = view[x].v;
\end{lstlisting}

\paragraph{$\ra$ reads}. This case is almost similar to the earlier and hence only state the point of difference. The main difference is that due to $\ra$ access, we merge (take the join of) all the timestamps rather than just $x$ as we did for $\rlx$.

\begin{lstlisting}[caption=$\texttt{read}_\ra$,multicols=2,style=customc]
// local read
if(*){
	ASSUME(view[x].f);
}
// (non-local) view-switching read 
else{
	ASSUME(numEE < K);
	messNum = nondet(0, messageUsed-1);
	mess = messageStore[messNum];
	ASSUME(mess.var == x);
	ASSUME(mess.t[x] > view[x].t or (mess.t[x] = view[x].t and view[x].l == true));

	// merge views
	for (y in X){
		if (mess.t[y] == view[y].t){
			view[y].l = (mess.l[y]) and (view[y].l);
		}
		else if (mess.t[y] > view[y].t){
			ASSUME(!liveChain[y]);
			view[y].t = mess.t[y];
			view[y].l = mess.l[y];	
		}		
	}
	view[x].v = mess.val;
	
	extView[x] = true;	
	numEE++;
}
val($r) = view[x].v;
\end{lstlisting}

\subsection{Writes}
We now provied the translation of a write instruction $x=\reg$ of process. Once again we simulate two access modes, $\rlx$ and $\ra$. we first describe the relaxed mode and then discuss the changes for the $\ra$ mode.

\paragraph{$\rlx$ writes} \textit{When in normal mode} \newline
Let us first consider execution in the normal phase (i.e., when \appvar{checkMode} is false). The value of $\mathit{val}(\reg)$ is recorded in the local view, \appvar{view[x].v} and \appvar{view[x].f} is set to true meaning that the value in \appvar{view[x].v} is a valid value and can be read from.
Then, we non-deterministically choose one of three possibilities for the write:
it either (i) is not assigned a fresh timestamp,
(ii) is assigned a fresh timestamp,
(iii) fulfils some outstanding promise. These nondeterministic branches are given on lines 5, 24 and 60 of the code.

\begin{lstlisting}[caption=$\texttt{write}_\rlx$,multicols=2,style=customc]
view[x].v = val($r);
view[x].f = true;

// no fresh timestamp
if (*){ 
	view[x].l = false;
	if (checkMode and !liveChain[x]){
		// new write does not rely on reservation
		if (*){
			// only true if process is in checkMode
			liveChain[x] = true;
			view[x].t = nondet(view[x].t, MAXTS);
		}
		// new write relies on reservation
		else {
			view[x].t = nondet(view[x].t, MAXTS);
			
			ASSUME(upd[x][view[x].t] or reserve[x][view[x].t] == p);
			reserve[x][view[x].t] = p;
			upd[x][view[x].t] = false;
			usedReservation[x][view[x].t] = true;
}	}	}
// a new timestamp is assigned to this write
else if(*){
	view[x].l = true;
	
	for (y in X){
		newView[y] = 0;
		newViewL[y] = true;
	}

	if (liveChain[x]){
		newView[x] = view[x].t + 1;
	}
	else {
		newView[x] = nondet(view[x].t + 1, MAXTS);
	}
	
	view[x].t = newView[x];
	ASSUME(avail[x][newView[x]]);
	avail[x][newView[t]] = false;

	// essential message
	if (*){
		if (checkMode){
			ASSUME(!blockPromise[x]);
			mess = genMessage(x, newView, newViewL, val($r), -1);
			liveChain[x] = false; numEE++;
		}
		else {
			mess = genMessage(x, newView, newViewL, val($r), 0);
		}
		Publish(mess);
	}
	else{
		ASSUME(!checkMode);
	}
}
// a previous Promise is certified
else{
	ASSUME(!blockPromise[x]);
	view[x].l = true;
	messageNum = nondet(0, messageUsed-1);
	mess = messageStore[messageNum];

	// ensure that the message is a promise and matches variable and value
	ASSUME(mess.var == x and mess.t[x] > view[x].t);
	ASSUME(mess.val == val($r) and mess.flag == p);
	
	if (checkMode){
		mess.flag = -1;
		ASSUME(!liveChain[x] or (view[x].t + 1 == mess.t[x]));
		liveChain[x] = false;
	}
	else {
		mess.flag = 0;
	}

	view[x].t = mess.t[x];
	messageStore[messageNum] = mess;
}

if (!checkMode){
	extView[x] = true;
} else {
	extView[x] = false;
}

view[x].u = true;
\end{lstlisting}

In case (i), no message is created, and \appvar{view[x].l} is set to false, signifying that the timestamp recorded in the view does not correspond to the most recent write to $x$ and should therefore not be used in the comparisons. The `if' branch on line 7 is not taken \appvar{checkMode} is false.

In case (ii), since in this case, the timestamp in the view is by definition valid, we set \appvar{view[x].l} to true (line 25). Since the write is relaxed, the message generated will only store the timestamp on the variable written to (i.e. $x$) and 0 for all other variables (line 27-30). Now we allocate a new timestamp to the write. Since we are in normal mode, \appvar{liveChain[x]} is false (see \appvar{liveChain} invariant in glossary). Thus we choose a timestamp nondeterministically (line 36) and store it into \appvar{view[x].t}. We use the \appvar{avail[x][.]} array to ensure that allocated timestamps are unique: (1) we check that the selected timestamp is available (i.e., not allocated) on line 40, and remove it from the array of available stamps (line 41).
Now this message can either be published (for cnsumption by another thread) or not. In the former case, the appropriate message is constructed with \appvar{newView}, \appvar{newViewL}. Note that the last component of the message stores the flag \appvar{mess.flag}. This flag is set to false since the message is not a promise (see \appvar{mess.flag} invariant in glossary). In the latter case non of this is done (`else' branch on line 55). The \appvar{assume(!checkMode)} is satisfied.

In case (iii) Finally, if the process decides to fulfill a promise, a message is fetched from \appvar{messageStore} and checked to be an unfulfilled promise by the current process (checking \appvar{flag == p} on line 68), and \appvar{mess.flag} is set to $0$ and message reinserted into \appvar{messageStore}. Additionally we set \appvar{extView[x]} to true maintaining the \appvar{extView} invariant.

\paragraph{$\rlx$ writes} \textit{When in check mode} \newline

Let us now consider a write executing in the certification phase (i.e., when \appvar{checkMode} is true). We will only highlight differences between the normal and certification phase writes. 

In case (i), that is when a fresh timestamp is not assigned, the write is certified either by deferring certification to a promise by using splitting insertion (line 9) or by the a presence of a reservation (line 15). In the case where, \appvar{liveChain[x]} is already true (line 7), certification for the current sequence of writes is already deferredand hence we do none of the two. While certification by either of splitting/reservation we nondeterministically choose an timestamp \appvar{t} after which the current write occurs (line 12). We note that this is not the timestamp of the write itself, but specifies between which two timestamps from $\mathsf{Time}$ the write occurs. If we rely on splitting insertion (line 9), we set \appvar{liveChain[x]} to true, and  In case of certification by reservation we reserve an interval adjacent to the timestamp \appvar{t} (line 19) after ensuring that it is available (line 18). Finally since this reservation has been used in some certification, we mark this fact (line 20).  

In cases (ii), the write is assigned a timestamp from $\mathsf{Time}$ and hence consequently published as a promise. We allocate a fresh timestamp and store it into \appvar{view[x].t}. The most important point to note is that we maintain and use the \appvar{liveChain} invariant whenever a fresh timestamp is assigned.
Indeed, if \appvar{liveChain} is true, the process must assign consecutive timestamps, otherwise it can non-deterministically choose any timestamp greater than \appvar{view[x].t} (line 32-37). Additionally, when generating a message, the \appvar{mess.flag} is set to \appvar{-1} denoting that the message is promise but has been certified and publish the message. We also increment \appvar{numEE} (line 48) as a promise is an essential event.

In case (iii) we fulfill an older promise, and thus first retrieve an uncertified promise belonging to the current process (\appvar{mess.flag == PID}) from \appvar{messageStore} (line 68). The main difference with the normal mode is that we set \appvar{mess.flag} to \appvar{-1} signifying that the promise is (temporarily) certified but not fulfilled. We set the \appvar{extView[x]} to false signifying that the processes' view has come from \appvar{checkMode} and hence is not external.

\begin{lstlisting}[caption=$\texttt{write}_\ra$,multicols=2,style=customc]
view[x].v = val($r);
view[x].f = true;

// no timestamp is assigned
if (*){
	view[x].l = false;
	
	/* for RA writes we must have P(x) = empty
	// also note  that reservations can be 
	// temporarily cancelled, so they are allowed
	// on the other hand concrete promises are not
	// thus new writes must rely on reservations */
	if (checkMode){
		blockPromise[x] = true;
		ASSUME(!liveChain[x]);

		view[x].t = nondet(view[x].t, MAXTS);
		
		// new write must rely on reservation
		ASSUME(upd[x][view[x].t] or reserve[x][view[x].t] == p);
		reserve[x][view[x].t] = p;
		usedReservation[x][view[x].t] = true;	
	}
}
// a new timestamp is assigned to this write
else {
	ASSUME(!checkMode);

	view[x].l = true;
	
	for (y in X){
		newView[y] = view[y].t;
		newViewL[y] = view[y].l;
	}

	newView[x] = nondet(view[x].t + 1, MAXTS);
	
	view[x].t = newView[x];
	ASSUME(avail[x][newView[t]]);
	avail[x][newView[t]] = false;

	// essential message
	if (*){
		mess = genMessage(x, newView, newViewL, val($r), 0);
		Publish(mess);
	}
}
\end{lstlisting}

\paragraph{$\ra$ writes}

The $\ra$ writes have some minor differences w.r.t $\rlx$. Firstly, the timestamps for all variables \appvar{view[x][t]} are added to the published messages, (lines 27-30). Next we set \appvar{blockPromise[x]} to true signifying that henceforth there cannot be any promises on $x$ (refer to \ps, $\ra$ accesses). This also implies that cases (ii) and (iii) (generating new promises and certifying earlier promises) is not possible for $\ra$ writes as enforced on (line 23). Note that \appvar{blockPromise[x]} is also assumed to be false in $\rlx$ writes when either generating new promises (ii) or certifying earlier ones (iii). 

\subsection{$\cas$ operations}

We only provide code for the $\cas(\rlx,\rlx)$ variant since the others are implemented similarly, carrying over the access dependent changes from the corresponding \texttt{read} and \texttt{write} codes. $\cas$ is bootstrapping a $\texttt{read}$ and $\texttt{write}$, additioanlly enforcing that the timestamps are consecutive.   

\begin{lstlisting}[caption=$\cas$,multicols=2,style=customc]
view[x].v = val($r);
view[x].f = true;

// no fresh timestamp
if (*){ 
	view[x].l = false;
	if (checkMode and !liveChain[x]){
		// new write does not rely on reservation
		if (*){
			// only true if process is in checkMode
			liveChain[x] = true;
			view[x].t = nondet(view[x].t, MAXTS);
		}
		// new write relies on reservation
		else {
			view[x].t = nondet(view[x].t, MAXTS);
			
			ASSUME(upd[x][view[x].t] or reserve[x][view[x].t] == p);
			reserve[x][view[x].t] = p;
			upd[x][view[x].t] = false;
			usedReservation[x][view[x].t] = true;
}	}	}
// a new timestamp is assigned to this write
else if(*){
	view[x].l = true;
	
	for (y in X){
		newView[y] = 0;
		newViewL[y] = true;
	}

	if (liveChain[x]){
		newView[x] = view[x].t + 1;
	}
	else {
		newView[x] = nondet(view[x].t + 1, MAXTS);
	}
	
	view[x].t = newView[x];
	ASSUME(avail[x][newView[x]]);
	avail[x][newView[t]] = false;

	// essential message
	if (*){
		if (checkMode){
			ASSUME(!blockPromise[x]);
			mess = genMessage(x, newView, newViewL, val($r), -1);
			liveChain[x] = false; numEE++;
		}
		else {
			mess = genMessage(x, newView, newViewL, val($r), 0);
		}
		Publish(mess);
	}
	else{
		ASSUME(!checkMode);
	}
}
// a previous Promise is certified
else{
	ASSUME(!blockPromise[x]);
	view[x].l = true;
	messageNum = nondet(0, messageUsed-1);
	mess = messageStore[messageNum];

	// ensure that the message is a promise and matches variable and value
	ASSUME(mess.var == x and mess.t[x] > view[x].t);
	ASSUME(mess.val == val($r) and mess.flag == p);
	
	if (checkMode){
		mess.flag = -1;
		ASSUME(!liveChain[x] or (view[x].t + 1 == mess.t[x]));
		liveChain[x] = false;
	}
	else {
		mess.flag = 0;
	}

	view[x].t = mess.t[x];
	messageStore[messageNum] = mess;
}

if (!checkMode){
	extView[x] = true;
} else {
	extView[x] = false;
}

view[x].u = true;
\end{lstlisting}

\subsection{Fences}

\paragraph{$\fence$} The $\fence$ command essentially merges the thread local view with the globally stored view in \appvar{globalTimeMap}. For each shared variable $x$ we do the following. On line 3 we check whether the globally stored view \appvar{globalTimeMap[x]} is greater than the process local view. if that is the case, we increase the process-local view \appvar{view[x].t} to the globally stored view. Additionally, we set \appvar{view[x].f} to false since, the value in \appvar{view[x].val} is no more valid (cannot be read from again, since the process timestamp has increased). In the order case, (line 8), we raise the \appvar{globalTimeMap[x]} either to \appvar{view[x].t} (if it is valid, checked by line 9) or to the next higher timestamp, \appvar{view[x].t + 1}.

\begin{lstlisting}[caption=$\fence$,multicols=2,style=customc,mathescape]
assume(!checkMode);
for (x in $\varset$){
	if (globalTimeMap[x] > view[x].t){
		view[x].t = globalTimeMap[x];
		view[x].f = 0;
		view[x].l = true;	
	}
	else {
		if (view[x]){
			globalTimeMap[x] = view[x].t;
		}
		else {
			globalTimeMap[x] = view[x].t + 1;
		}
	}
}
\end{lstlisting}

\newcommand{\psrlx}{\ps}
\newcommand{\cso}{CSO}
\newcommand{\scfence}{\texttt{sc-fence}}

\subsection{Correctness of Translation (Proof of Theorem \ref{thm:s2s})}
\label{app:bv}
The proof is in two parts. In the first part, we show that that every $K+n$ context bounded run of $\prog'$ in SC corresponds to a $K$-bounded run of $\prog$ under $\psrlx$, and in the second part, we show that for every $K$-bounded run in $\psrlx$, there is a $K+n$ context bounded run in SC.  

At the outset we review a high level description of the translation. We denote by \emph{normal} 
and $\mathit{checkMode}$, the two phases 
respectively where  
$\mathit{checkMode}$ is false and $\mathit{checkMode}$ is true. These are the two phases  in which a process functions. Each process executes instructions in the \emph{normal} phase by skipping over the $\cso$ blocks of code. When a process needs to switch out, it enters the $\cso$ block following the most recent instruction executed and sets $\mathit{checkMode}$ to true. Now, it makes a ``ghost'' run in $\mathit{checkMode}$, a terminology to indicate that this phase of the run does not change the the global state and local state of the process permanently (this is facilitated by the saveState and loadState functions). One exception to this is the writes that the process makes as reservations, and published promises which are maintained permanently. Hence, this part of the run is equivalent to the process making fresh promises after a \emph{normal} execution; providing a witness for consistency and then switching out of context. The run then is a sequence of interleaved \emph{normal} and $\mathit{checkMode}$ phases. Moreover, the local states of the process is identical at the start and end of any given $\mathit{checkMode}$ phase.

We request the reader to refer to the glossary [\ref{app:ctc}] of the variables used which will aid in better understanding of the translation.

We give the proof of correctness of the translation through two 
sections. 

\textbf{Intuition} The translation relies on the fact that in a run of the $K$-bounded $\ps$ program, it suffices to store the relative order only between $K$ totally ordered timestamps for each variable. Additionally, these $K$-timestamps are precisely those corresponding to the $K$ essential events - promises, reservations, view-altering reads. While we maintain an exact ordering between essential events, those of non-essential events (which are none of  view-altering reads, reservations or promises) are abstracted in the SC run. Thus in the original run under $\psrlx$, all timestamps are exact, while in the run under SC, the non-essential timestamps are abstracted away. 

The correctness of the translation then relies on being able to faithfully concretize the abstract timestamps from the SC run. We account for these concretizations by separating the essential timestamps by sufficiently large intervals, so that, the non-essential timestamps can be inserted in between, respecting their order.

\subsection*{SC to $\psrlx$}
\textbf{Details} We start from SC to $\psrlx$. We show that every $K+n$ context bounded run of $\prog'$ under SC corresponds to a $K$-bounded run of $\prog$ under $\psr$. Keeping in mind the description above, we split this proof into two parts. 
\begin{enumerate}
	\item First, we consider only runs in \emph{normal} mode and prove that they have an analog in $\psrlx$.
	\item Second, we prove that any run in $\mathit{checkMode}$  is indeed an analog of a process making fresh promises and reservations and certifying them along with previous unfulfilled promises, before switching out of context. 
\end{enumerate}
Combining these two, indeed, we will have a run under $\psrlx$.

We  begin by defining some terminology. Consider a run $\tau$ of program $\prog'$. Each event of the run $\tau$ is an execution of either a read, write, $\arw$ or $\fence$. 
A read in this run is called \emph{global} (and otherwise \emph{local}) if the process decides to read from the global array $\mathit{messageStore}$. Only global reads can be view-altering in the corresponding run under $\psrlx$. A write can be of three types - publishedS, publishedF and local. These represent, `simple published', `fulfilling published', and `timestamp not assigned writes' respectively. Note that each of these types can be performed in \emph{normal} as well $\mathit{checkMode}$. A $\arw$ can therefore be of 6 types since it involves a read and write. At a high level this translation is facilitated by the following two \textbf{key observations}:
\begin{itemize}
\item The number of publishedS, publishedF writes are bounded due to the bound $K$, and hence the requisite data-structure for these can be maintained using bounded space.
\item Local writes are unbounded, however, these writes are only used (read-from) locally by the writing process and need not be stored permanently by the algorithm. 
\end{itemize}

Let $w_1$ be the number of $write$ events in the \emph{normal} mode of run $\tau$, $w_2$ be the maximum number of $write$ events, maximum being taken over all $\mathit{checkMode}$ phases of the run, $u-1$ be the number of $\arw$ events in the run, and let $\ell = w_1 + w_2 + u$. Let $\mathsf{M_x}$, for each shared variable $x$, be an increasing function from $[2K]$ to $\mathsf{N}$ representing a mapping from the notion of time-stamps in SC to time-stamps in $\psrlx$. For each  variable $x$, and each process $p$, let $\mathsf{View_{SC}(x)} = \mathit{view}[x].t$ (defined above) and $\mathsf{View_{\psrlx}(x)}$ be the time stamp of $x$ in the view of $p$ in $\rho$. Given a run $\tau$, we will construct a $K$ bounded run $\rho$ of $\prog$ which reaches the same set of labels after $i$ events, for any $i$. 

We will first treat the \emph{normal} (non-$\mathit{checkMode}$) part of the run. While going through the steps, we will also construct the increasing functions $\mathsf{M_x}$. In addition to the invariants in $\ref{app:ctc}$, we maintain the following timestamp-based invariants for all processes $p$ and variables $x$. 

\begin{enumerate}
    \item If $\mathit{view}[x].l$ is true for a process in $\tau$, then $\mathsf{M_x}(\mathsf{View_{SC}(x))} = \mathsf{View_{\psrlx}(x)}$.
    \item If $\mathit{view}[x].l$ is true and the time-stamp $\mathit{view}[x].t$ corresponds to a write message instead of a message added due to a $\arw$, then $\mathsf{M_x}$($\mathit{view}[x].t$) =  $\mathit{view}[x].t \cdot \ell  \cdot u$
    \item If $\mathit{view}[x].l$ is false, then $\mathsf{M_x}(\mathit{view}[x].t) < \mathsf{View_{\psrlx}(x)} < (\mathit{view}[x].t+1) \cdot \ell \cdot u$. Moreover, if the last event to assign false to $\mathit{view}[x].l$ was a write,  then $\mathsf{View_{\psrlx}(x)}$ is a multiple of $u$.
    \item If a message is of type $\arw$, then its time-stamp $t$ in $\rho$ satisfies $t \not \equiv  0 \mod u$
    \item The sum of view-switch points and promises is $\leq K$ in $\rho$.
    \item The time-stamps of  essential messages in $\tau$ and the corresponding message in $\rho$ are related by $\mathsf{M_x}$. That is, $\mathsf{M_x}(\mathsf{View_{SC}(x)}) = \mathsf{View_{\psrlx}(x)}$.
\end{enumerate}

The base case, that is, after 0 events ($i=0$) is trivial since the configurations are semantically equivalent and we define $\mathsf{M_x}(0) = 0$ for all variables, which satisfies the invariants. We make the following three cases depending on the $i^{th}$ event of $\tau$. \\
\begin{itemize}
    \item Case 1. $e_i$ is an execution of a write for process $p$, variable $x$ and value $v$. 
    \begin{itemize}
    \item 
    If the write is of publishedS or publishedF type, then $\mathit{view}[x].t$ is updated from $t$ to a new time-stamp $t'$ (which in the case of publishedF is the timestamp of the retrieved message) and $\mathit{view}[x].l$ is assigned true. In $\rho$, if we can make $\mathsf{View_{\psrlx}(x)}$ = $t'' = t' \cdot \ell  \cdot u$ then the invariants are satisfied. It is not possible for $t''$ to have been assigned already to some write message in $\rho$ since $t'$ was not  assigned to some message in $\tau$ (checked using $\mathit{avail}[x][t']$). A $\arw$ message could not have been assigned $t''$ either, by the fourth invariant. Since $t<t'$, $\mathsf{View_{\psrlx}(x)} < t''$ (by invariants 2 and 3). Hence, $\mathsf{View_{\psrlx}(x)}$ can be updated to $t''$ since it is available and is greater than the current view. If the write is published, then the message is added to $\mathit{messageStore}$. This is done to maintain invariant (6). Note how, if the write is of publishedF type, the message flag is set to 0, effectively removing it from the promise bag and maintaining the $\mathit{flag}$ invariant (refer to [\ref{app:ctc}]).
    \item If the write is local, then we pick the smallest available multiple of $u$ between $\mathsf{M_x}(\mathit{view}[x].t)$ and $(\mathit{view}[x].t+1) \cdot \ell  \cdot u$. This can always be done since there are $\ell-1$ multiples of $u$ between $\mathit{view}[x].t \cdot \ell  \cdot u$ and $(\mathit{view}[x].t+1) \cdot \ell  \cdot u$ and there are $\leq (\ell-1)$ messages (even considering those produced in $\mathit{checkMode}$) in total. Notice that multiples of $u$ have been reserved for writes by invariant 4.
    \end{itemize}
    \item Case 2. $e_i$ is an execution of a read for process $p$, variable $x$. 
    \begin{itemize}
    \item If the read is local in $\tau$, then the process is either reading a local message written by itself or a useful message (a useful message is one which is read by a process, but does not create a change of view).      In either case, this read can be performed in $\rho$ without any change in time-stamps. Note that this cannot be a view-switching event. Moreover note that the local value in $\mathit{view}[x].v$ has been ascertained to be usable. 
    \item If the read is global, then $numEE < K$ before the read and therefore $numEE \leq K$ afterwards. In this case, a message is fetched from $\mathit{messageStore}$ and the process view is updated according to this message. Since $\mathsf{M_x}$ is an increasing function, the results of comparisons in SC will be the same as in $\psrlx$ and the read operation has the same effect on values and time-stamps of the variables. Moreover $\mathit{view}[x].f$ is set to true maintaining the $\mathit{view}[x].f$ invariant [\ref{app:ctc}].
    \end{itemize}
    \item Case 3. $e_i$ is an execution of a $\arw$  for process $p$, variable $x$ and values $v$, $v'$. 
    \begin{itemize}
    \item If the read here is local, and $\mathit{view}[x].u$ is true then we need to ensure that the timestamp chosen for the write immediately follows $\mathsf{M_x}(\mathit{view}[x].t)$. It is first checked if $\mathit{view}[x].t$ has been used for an update earlier or not. If it has not been, then the time-stamp $\mathsf{M_x}(\mathit{view}[x].t) + 1$ is available in $\psrlx$ since all messages that come from writes have time-stamps in multiples of $u$ and $\mathsf{M_x}(\mathit{view}[x].t)$ is a multiple of $u$. Note, that we also ensure that $\mathit{view}[x].f$ is true in this case, which implies that the local value is usable.
    \item If the read here is local and $\mathit{view}[x].u$ is false (and hence so is $\mathit{view}[x].l$), then it definitely has not been used for an update ($\arw$) in $\tau$ since the process reading the message is the only one that knows of its existence. Now, if this message was a result of a local write, then its time-stamp $t$ in $\psrlx$ is a multiple of $u$ and $t+1$ is available for the update message. Otherwise, this message was a result of a $\arw$ whose write was local and has a time-stamp of the form $a \cdot u + b$ where $b<u$. Note that this implies $b-1$ consecutive $\arw$s were made to get here since all the messages that are a result of (non-$\arw$) write operations get time-stamps that are multiples of $u$. Since $u-1$ is the total number of $\arw$s in $\tau$, $b < u-1$ (at most $u-2$ $\arw$s have taken place before this one). This implies $a \cdot u+b+1$ is available and can be used for the write.
    \item If the read is global, then it is done correctly as explained in Case 2. The write part of the $\arw$ goes through as explained above.
\end{itemize}
    \item Case 4: $e_i$ is an $\fence$
    \begin{itemize}
        \item We iterate over the variables, updating $\mathit{globalTimeMap}[x]$ and $\mathit{view}[x].t$ to the maximum of the two.
        \item In case the former was greater, we set $\mathit{view}[x].l$ to true,  signifying that $\mathit{view}[x].t$ is valid and maintaining invariant (1) above. Moreover we set $\mathit{view}[x].f$ to false. This is necessary since, the timestamp of the message corresponding to $\mathit{view}[x].v$ is now less than $\mathit{view}[x].t$ and hence the locally stored value is unusable.
        \item If the latter is greater, we check whether $\mathit{view}[x].l$ is true (which signifies that $\mathit{view}[x].t$ is valid). If it is we can set $\mathit{globalTimeMap}[x]$ to it. If not, then the $\mathsf{M_x}(\mathit{view}[x].t) < \mathsf{View_{\psrlx}(x)}$ (by invariant (6)), and hence we set it to $\mathit{view}[x].t + 1$. Finally we note that $\mathsf{View_{\psrlx}(x)} < (\mathit{view}[x].t + 1)\cdot \ell \cdot u$ and hence $\mathsf{M_x}(\mathit{globalTimeMap}[x])$ now matches the essential event immediately following the event with timestamp $\mathit{view}[x].t$.
    \end{itemize}
\end{itemize}

We now briefly justify the $\mathit{checkMode}$ phase of the run. For any such phase, we need to ascertain that the run has analogous run in $\psrlx$ which respects the notion of consistency. The management of timestamps is identical to the \emph{normal} phase explained above so we only highlight the special aspects. First we recall some invariants:
\begin{enumerate}
    \item $\mathit{liveChain}[x]$ is true only when the most recent write made in the \textit{current} $\mathit{checkMode}$ phase was unpublished (was not a promise) and neither was it certified using a reservation.  
    \item $\mathit{extView}[x]$ is true if $\mathit{view}[x].v$ corresponds to a message from outside $\mathit{checkMode}$.
    \item For the process $p$ currently in $\mathit{checkMode}$, $message\_flag$ is -1 for temporarily (only within current $\mathit{checkMode}$ phase) certified promises and is $p$ for as yet uncertified promises. If it is $p' \neq p$, then the message is in the promise bag of some other process. Additionally if it is 0, it is not in the promise bag of any process. Note how this is maintained in the write, $\arw$ sections above.
\end{enumerate}
We review how these invariants are maintained and used throughout the code. When entering $\mathit{checkMode}$, $\mathit{liveChain}[x]$ is false. For any write happening in \emph{normal} phase we set $\mathit{extView}[x]$ to true. Otherwise we set it to false. Once again we consider cases for a particular event $e_i$:
\begin{itemize}
    \item Case 1. $e_i$ is a write event.
    \begin{itemize}
        \item In the case, the process performs a local write, the process can either set $\mathit{liveChain}[x]$ is set to true, maintaining the invariant or it can generate a reservation which will be used to certify the write. In this case the reservation is marked as used.
        \item In the case the process decides to publish a write it must publish it as a promise, incrementing $numEE$ (after checking that the bound of $K$ has not been crossed), setting the promise flag to -1, maintaining invariant (3) above (leading to a publishedS write). Also, if it decides to certify a previous promise , it does so, similar to the \emph{normal} phase, though it now sets the timestamp to -1, indicating that the certification is local to the current phase and must be reset when normal phase resumes. Moreover (publishedF write) note that $\mathit{liveChain}[x]$ is set to false  maintaining invariant (1).
        \item Also, note that $\mathit{extView}[x]$ is set to true maintaining invariant (2).
    \end{itemize}
    \item Case 2. $e_i$ is a read event.
    \begin{itemize}
        \item The main highlight of read events in $\mathit{checkMode}$, is that we ascertain that $\mathit{liveChain}[x]$ is false while making a global read. This is to ensure that we forbid additive insertion. Indeed, following invariant (1) above, if $\mathit{liveChain}[x]$ were true during a global read, it would mean that the interval corresponding to the previous message (which caused $\mathit{liveChain}[x]$ to be true) is additively.
    \end{itemize}
    \item Case 3. $e_i$ is a $\arw$ event.
    \begin{itemize}
        \item Once again similar to \emph{normal} phase we guess whether we make a local or a global read. Crucially however, we note that we forbid making a local write for a $\arw$ when $\mathit{extView}[x]$ is true. Considering the invariant (2) above, this is done precisely to forbid $\arw$ where, the promised interval containing the write is non-adjacent to the message being read from. The remainder book keeping is identical to previous cases.
    \end{itemize}
    \item Case 4. $e_i$ is a $\fence$ event. This case does not arise since a process in $\mathit{checkMode}$ may not execute a $\fence$ instruction,  as otherwise  the run will not be consistent \cite{promising,promising2}.
\end{itemize}

To conclude, note due to loadState and saveState functions, only used reservations and promises are retained after the $\mathit{checkMode}$ phase. Moreover due to the check of message flags after termination of a $\mathit{checkMode}$ phase, it is ensured that the process is in a consistent state while switching contexts. Noting that we keep track of promises as well as view-switches using $numEE$ we may only generate a run in which the sum of the two is bounded by $K$.

Next, we consider the converse direction from $\psr$ to SC.  
\subsection*{$\psrlx$ to SC}
We now prove the second part, from $\psrlx$ to SC. We prove that for every $K$-bounded run $\rho$ in $\psrlx$, there is a $K+n$ context bounded run $\tau$ in SC. We will show this in two steps.
\begin{itemize}
\item  Given the $K$-bounded $\rho$, first we will construct a run $\rho''$ which is $K$-bounded and $K+n$ context bounded that reaches the same configuration as $\rho$. 
\item We will then construct a run $\tau$ of SC using $\rho''$. 	
\end{itemize}
\textbf{Intuition} While we concretized the abstract (non-essential) timestamps when going from SC to $\psrlx$ earlier now we do the opposite. However, we will additionally show that $K+n$ SC contexts suffice for the translation. The way we account for the $K+n$ contexts is as follows - $n$ contexts for the process initializations and (atmost) one context for each essential event.    

Hence, we ensure that atleast one essential event occurs in each context. This is possible for the following reason. Consider a run with $K$ essential events occuring in some order executed by processes $p_1$ to $p_K$. If we schedule the processes $p_i$ in the run under SC in the same order, then we will get a valid run under SC. Since view-switches account for all the external reads-from dependencies, the runw which we obtain is also valid.

More concretely, we ensure that each process only switches out of context only when it is awaiting a message for an external read from another process or when it has made atleast one promise or reservation. Since the total number of such essential events along a \emph{normal} phase + additional messages in all $\mathit{checkMode}$ phases is bounded above by $K$, we need at most $K+n$ context switches. We add $n$ for the concluding contexts required to reach the $\terminated$ configurations.
\newline\textbf{Details}
Let $\mathit{rf}$ (called $reads$-$from$) be a binary relation on events such that $(e_a, e_b) \in \mathit{rf}$ iff $e_b$ reads from a message $\textit{published}$ by $e_a$. Note that every run under $\psrlx$ semantics defines a $\mathit{rf}$ relation as the reads are executed. For construction of $\rho''$, the intuition is that a context switch is required only when the current process has reached $\terminated$ or it needs a message that is yet to be published by some other process. At a configuration $\conf_i$ of $\rho$, we say that an event of $\rho$ is a \emph{requesting} event if it is a view-altering event in $\rho$ and it reads a message that is not in the message pool at $\conf_i$. Also, we call the events that publish messages for these events as \emph{servicing} events ($\mathfrak{write}$ or $\arw$, either simple or promises). Note that the set of servicing and requesting events is dependent on the configuration $\conf_i$. The two sets change along the run $\rho$. Specifically, an event is removed from the requesting event set as soon as the servicing event corresponding to it is executed. Let the size of the set of requesting events be $r$. At $\initconf$, $r = K$. We will prove by induction that given a set of processes ($n$), the $\mathit{rf}$ relation, and a run $\rho$ in $\psrlx$ that maintains the $\mathit{rf}$ relation, there is a run which uses at most $r+n$ context switches and defines the same $\mathit{rf}$ relation.

\noindent{\bf{The Base Case}}.
For $r+n=1$, there is only one process so the number of context switches is $0$ and the run $\rho$ itself uses 0 context switches. 

\noindent{\bf{The Inductive Step}}.
Assume the hypothesis for $r+n=\ell$ and we prove the claim for $r+n=\ell+1$. Clearly at $\initconf$, there is at least one process which either has no requesting events, or has a servicing event before any requesting events in its instruction sequence. Otherwise, the run $\rho$ will not be able to execute all the events since no process will be able to move past its requesting event.
If we have a process that can reach termination directly, then in $\rho''$, we run that process and reduce $r+n$. Otherwise, consider the instructions of the process ($p_j$) that has a servicing event before any of its requesting events.
The instructions of $p_j$, till the first requesting event, can be executed since all the messages they need are already in the pool and hence we can create a new run $\rho_t$ in which these instructions are executed first and the remaining ones follow the same order as $\rho$. Note that $\rho_t$ reduces $r$ by at least $1$ while executing the instructions of $p_j$. By applying the hypothesis on the remaining sequence of instructions, we have a run that uses $r-1+n$ context switches and that maintains $\mathit{rf}$ of the remaining instructions. This can now be combined by the instructions of $p_j$ that have already been executed to give $\rho''$.  \\
We now construct the run $\tau$ from $\rho''$. As explained in the text above, at most $2K$ time-stamps are needed to simulate the $\rho''$. Let the set of such time-stamps be $U\_x$ for each variable $x$. Let $\mathsf{M_x}$ be an increasing (mapping) function for each variable from $U\_x \cup \{0\}$ to $\{0, \dots 2K\}$ such that $\mathsf{M_x}(0)=0$. 

We will construct the run $\tau$ in SC from $\rho''$, event by event, while maintaining the following invariants
\begin{enumerate}
    \item All the time-stamps, in a particular message in $\mathit{messageStore}$, are related to the time-stamps in the corresponding essential messages in $\psrlx$ by $M_x$.
    \item  For a process $p$, $\mathsf{View_{\psrlx}(x)} \in U\_x$ iff $\mathit{view}[x].l$ is true at that point in SC and $\mathit{view}[x].t$ = $\mathsf{M_x}(\mathsf{View_{\psrlx}(x)}))$ 
\end{enumerate}
The $i^{th}$ event of $\rho''$ can be one of the following: 
\begin{itemize}
    \item Case 1. $e_i$ is a write to variable $x$ with value $v$.
    \begin{itemize}
    \item If the time-stamp $t$ of this write belongs to $U\_x$, then we first allocate $M_x(t)$ in SC to this write and make $\mathit{view}[x].l$ true. This maintains invariant (2).
    \item If the event is a servicing event, then the time-stamp of this message satisfies the requirements of invariant (1) and hence it can be added to $\mathit{messageStore}$. 
    Otherwise, we do not update the $\mathsf{View_{SC}(x)}$ of the process and make $\mathit{view}[x].l$ false.
    \end{itemize}  
    \item Case 2. $e_i$ is a read of variable $x$. \\
         If this event is a view-altering event, then the current timestamp in the 
         $\mathsf{View_{\psrlx}}$ will be used for comparison. The effect of the read in SC will be same as in $\psrlx$ since $V\_x$ is an increasing function. All the invariants will still hold after this, since all the messages in $\mathit{messageStore}$ satisfy the invariants.
    \item Case 3. $e_i$ is a $\arw$ to variable $x$ with values $v, v'$.
    If this event is not view-altering, then the process either reads some other process's message again or reads its own. If it reads its own message, then no change to the $\mathsf{View_{SC}(x)}$ has to be done for the read part and the new message is added to $\mathit{messageStore}$ if $e_i's$ message is essential. If it reads some other processes' message again, then $\mathit{view}[x].l$ is true, and since this message has not been used for a $\arw$ yet, the check of $upd\_x[\mathit{view}[x].t]$ will go through in $Prog'$. Now, it needs to be decided if the new message is essential. If the read is view-altering, then it is similar to Case 2 followed by the decision of adding the new message to $\mathit{messageStore}$.
    \item Case 4. $e_i$ is an $\fence$
        If $\mathit{globalTimeMap}[x]$ is greater than $\mathit{view}[x].t$, we maintain invariants (2) by setting $\_\mathit{view}[x].l$ to true and the $\mathit{view}[x].f$ invariant [\ref{app:ctc}] by setting it to $\mathit{view}[x].f$. On the other hand, if $\mathit{view}[x].t$ is greater, we set $\mathit{globalTimeMap}[x]$ to the smallest member $t \in \mathsf{Time}$, which satisfies $t \geq \mathsf{M_x}(\mathsf{View_{\psrlx}}(x))$. In case $\mathit{view}[x].l$ is true, $t$ is $\mathit{view}[x].t$ itself by invariant (2). If not, then we set it to $\mathit{view}[x].t + 1$, since we note that $\mathit{view}[x].t$ is the largest member of $\mathsf{Time}$, that $p$ has had as $\mathsf{View_{\psrlx}}(x)$, and currently the former is lower than $\mathsf{M_x}(\mathsf{View_{\psrlx}}(x))$.
\end{itemize}

\newpage
\section{Complete Experimental Results}

We report the results of experiments we have performed with \tool. We have two objectives: (1) studying the performance of \tool~ on benchmarks which are unsafe only with promises and (2) comparing \tool~ with other model checkers when operating in the promise free mode.
In the first case, we show that \tool~ is able to uncover bugs in examples with low interaction with the shared memory. When this interaction increases, however, \tool~ performs poorly, owing to the huge non-determinism required by \ps. However, with partial promises, \tool~ is once again able to uncover bugs in reasonable amounts of time. 
In the second case, our observations highlight the ability to detect hard to find bugs with small $K$ for unsafe benchmarks, and scalability by altering $K$ as discussed earlier in case of safe benchmarks. We compare \tool with three state-of-the-art stateless model checking tools, $\cdsc$ \cite{cdsc}, $\genmc$ \cite{genmc} and $\rcmc$ \cite{rcmc} that support the promise-free subset of the \ps semantics. 

We now report results of all the experiments we have performed with \tool. In the tables that follow we 
provide the value of $K$ used (for our tool only). We also specify the value of $L$ used (for all tools).

We do not consider compilation time for any tool while reporting the results. For our tool, the time reported is the time taken by the CBMC backend for analysis. The timeout used is 1 hour for all benchmarks. All experiments are conducted on a machine equipped with a 3.00 GHz Intel Core i5-3330 CPU and 8GB RAM running a Ubuntu 16 64-bit operating system. We denote timeout by `TO', and memory limit exceeded `MLE'.

\subsection{Experimenting with Promises}

In this section we experiment with \tool~ in the \textit{promise-enabled} mode. 

\paragraph{Litmus Tests}
We first test the tool on a number of litmus tests obtained from various sources. This has two objectives: (a) to perform sanity checks on the correctness of the tool (b) to gain an understanding of the causes of performance bottlenecks when handling promises. The results of these tests are summarized in Table \ref{tab:prom0} below. We tested \tool~ on many litmus tests from \cite{promising,promising2,weakestmo,Svendsen:2018}. In these \tool~ terminated with the correct result within one minute, with the value of $K$
used for the unsafe trace being atmost 5. We also tested \tool~ on the Java Causality Tests of Pugh \cite{jmm}, which were also experimented on in \citet{mrder}. In these too we were able to verify most examples within one minute. However, \tool~ timed out (TO = 30 mins) on two tests.
\begin{table}[!htb]
\small
\begin{minipage}{0.5\textwidth}\centering
\begin{tabular}{cccc}
\hline
\textbf{testcase} & $K$ & \textbf{\tool}
\\
\hline\hline
ARM\_weak  & 4 & 0.765s  \\ 
Upd-Stuck & 4 & 1.252s \\
split & 4 &  25.737s  \\ \hline
LB & 3 & 1.469s \\
LBd & 3 & 1.481s \\
LBfd & 3 & 1.512s \\
LBcu & 4 &  5.253s \\ 
LB2cu  & 4 & 5.748s \\ \hline
CYC & 5 & 1.967s \\
Coh-CYC & 5 & 42.67s
\end{tabular}
\end{minipage}\begin{minipage}{0.5\textwidth}\centering
\begin{tabular}{cccc}
\hline
\textbf{testcase} & Testcase-Safety & $K$ & \textbf{\tool}
\\
\hline\hline
Pugh2  & Unsafe &  3 & 13.725s  \\ 
Pugh3  & Unsafe & 3 & 12.920s  \\ 
Pugh6  & Unsafe & 3 & 0.360s  \\ 
Pugh8  & Unsafe & 3 & 1.67s  \\ \hline
Pugh4  & Safe & 5 & 3.244s  \\ 
Pugh5  & Safe & 5 & 4.811s  \\ 
Pugh10 & Safe & 5 & 3.868s \\
Pugh13 & Safe & 5 & 3.345s \\ \hline
Pugh14 & - & 3 & TO \\
Pugh15 & - & 3 & TO \\
\end{tabular}
\end{minipage}
\caption{Performance of \tool~ on \ps~ idioms}
\label{tab:prom0}
\end{table}
\vspace{-0.6cm}

\paragraph{Modular Promises}

In this section we ask whether the source-to-source translation technique can effectively scale while handling promises for \ps. In conclusion, we note that our approach performs well on programs requiring limited global memory interaction. When this interaction increases \tool~ times out, owing to the huge non-determinism of \ps. However, the modular approach of partial-promises enables us to recover effective verification.

\begin{table}[h]
\small
\begin{tabular}{cccc}
\hline
\textbf{testcase} & $K$ & \textbf{\tool}[1p]
\\
\hline\hline
fib\_global\_2  & 4 & 55.972s \\
fib\_global\_3  & 4 & 2m4s  \\ 
fib\_global\_4 & 4 &  4m20s  \\ \hline
exp\_global\_1  & 4 & 19m37s \\
exp\_global\_2  & 4 & 41m12s  \\  \hline
tri\_global\_2  & 4 & 52.973s \\
tri\_global\_3  & 4 & 1m57s  \\  
tri\_global\_4  & 4 & 3m58s \\ \hline
\end{tabular}
\caption{Performance of \tool~ on cases with global update}
\end{table}
\vspace{-0.6cm}

\subsection{Comparing Performance with Other Tools}

\begin{table}[h!]
\small
\begin{tabular}{ccccccc}
\hline
\textbf{benchmark} & $L$ & $K$ & \textbf{\tool} & \textbf{CDSChecker} & \textbf{GenMC} & \textbf{RCMC} \\ \hline\hline
exponential\_5\_unsafe      & 10 & 10 & 1.312s          & 0.900s          & 0.135s & 6.692s           \\
exponential\_10\_unsafe      & 10 & 10 & 1.854s          & 1.921s          & 0.367s & 3m41s           \\
exponential\_25\_unsafe      & 25 & 10 & 3.532s         & 7.239s          & 3.736s & TO                 \\
exponential\_50\_unsafe      & 50 & 10 & 6.128s         &  36.361s            & 39.920s  & TO                 \\
exponential\_70\_unsafe      & 10 & 10 & 9.509s          & 1m33s          & 2m29s & TO           \\
\hline
fibonacci\_2\_unsafe                 & 2  & 20 & 2.746s          & 2.332s     & 0.084s     & 0.086s          \\
fibonacci\_3\_unsafe                 & 3  & 20 & 9.392s         &  46m8s         & 0.462s     & 0.544s          \\
fibonacci\_4\_unsafe                 & 4  & 20 & 34.019s         & TO         & 12.437s     & 18.953s            \\ \hline 
fibonacci\_2\_safe                 & 2  & 20 & 6.454s          & 8.900s     & 0.096s     & 0.162s          \\
fibonacci\_3\_safe                 & 3  & 20 & 30.936s         &  TO         & 0.910s     & 3.884s          \\
fibonacci\_4\_safe                 & 4  & 20 & 2m16s         & TO         & 1.140s     & 2m36s            \\ \hline 
\end{tabular}
\caption{Comparison of performance on a set of parameterized benchmarks}
\vspace{-0.6cm}
\end{table}

\begin{table}[h]
\vspace{0.6cm}
\small
\begin{tabular}{ccccccc}
\hline
\textbf{benchmark} & $L$ & $K$ & \textbf{\tool} & \textbf{CDSChecker} & \textbf{GenMC} & \textbf{RCMC} \\ \hline\hline
hehner2\_unsafe    & 4 & 5 & 7.207s & 0.033s & 0.094s & 0.087s \\ 
hehner3\_unsafe    & 4 & 5 & 28.345s & 0.036s & 2m53s & 1m13s \\ \hline
linuxlocks2\_unsafe    & 2 & 4 & 0.547s & 0.032s & 0.073s & 0.078s \\ 
linuxlocks3\_unsafe    & 2 & 4 & 1.031s & 0.031s & 0.083s & 0.081s \\ \hline
queue\_2\_safe & 4 & 4 & 0.180s & 0.031s & 0.082s & 0.085s \\
queue\_3\_safe & 4 & 4 & 0.347s & 0.037s & 0.090s & 0.092s \\ \hline
\end{tabular}
\caption{Comparison of performance on concurrent data structures based benchmarks}
\vspace{-0.6cm}
\end{table}

\begin{table}[h]
\small
\begin{tabular}{ccccccc}
\hline
\textbf{benchmark} & $L$ & $K$ & \tool & \textbf{CDSChecker} & \textbf{GenMC}& \textbf{RCMC}  \\ \hline\hline
readerwriter\_7              & 0 & 5 &  0.719s     & 0.005s & 0.057s & 0.690s        \\ 
readerwriter\_8              & 0 & 5 & 0.839s     & 0.006s & 0.056s & 7.425s        \\ 
readerwriter\_9              & 0 & 5 & 1.068s     & 0.007s & 0.053s & 1m17s         \\ 
readerwriter\_10             & 0 & 5 & 1.393s
    & 0.007s & 0.056s & 14m49s        \\ \hline
redundant\_co\_10  & 10 & 5 & 0.470s          & 0.114s              & 0.087s & 38m12s        \\ 
redundant\_co\_20  & 20 & 5 & 1.031s          & 0.548s              & 0.218s & TO            \\ 
redundant\_co\_50  & 50 & 5 & 3.219s          & 8.965s              & 4.143s & TO            \\
redundant\_co\_70  & 70 & 5 & 6.093s          & 13.843s             & 18.185s & TO            \\\hline
\end{tabular}
\caption{Evaluation using two synthetic safe benchmarks. We note that the value of $K$ is chosen to be large enough to consider all executions.}
\end{table}

\begin{table}[h]
\small
\begin{tabular}{ccccccc}
\hline
\textbf{benchmark} & $L$ & $K$ & \textbf{\tool} & \textbf{CDSChecker} & \textbf{GenMC} & \textbf{RCMC} \\ \hline\hline
peterson1U(4)    & 1 & 6 & 1.408s & 0.039s & TO & 9.129s \\ 
peterson1U(6)    & 1 & 6 & 7.286s & 0.010s & TO & TO \\
peterson1U(8)    & 1 & 6 & 47.786s & TO & TO & TO \\ 
peterson1U(10)   & 1 & 6 & 4m19s & TO & TO & TO \\ 
\hline
szymanski1U(4)   & 1 & 2 & 1.015s & 0.043s & MLE & TO \\
szymanski1U(6)   & 1 & 2 & 2.771s & TO & MLE & TO \\ 
szymanski1U(8)   & 1 & 2 & 6.176s & TO & TO & TO \\ 
szymanski1U(10)  & 1 & 2 & 12.203s & TO & TO & TO \\
\hline
\end{tabular}
\caption{Comparison of performance on mutual exclusion benchmarks with a single unfenced process}
\vspace{-0.5cm}
\end{table}

\begin{table}[h]
\vspace{0.5cm}
\small
\begin{tabular}{ccccccc}
\hline
\textbf{benchmark} & $L$ & $K$ & \textbf{\tool} & \textbf{CDSChecker} & \textbf{GenMC} & \textbf{RCMC} \\ \hline\hline
peterson1C(3)    & 1 & 2 & 0.487s & 0.053s & 0.083s & 0.087s \\ 
peterson1C(4)    & 1 & 2 & 1.193s & 3.500s & TO & 3.360s \\ 
peterson1C(5)    & 1 & 2 & 2.713s & TO & TO & TO \\
peterson1C(6)    & 1 & 2 & 6.045s & TO & TO & TO \\ 
peterson1C(7)    & 1 & 2 & 11.008s & TO & TO & TO \\ \hline
peterson2C(3)    & 1 & 2 & 0.481s & 0.032s & 0.099s & 0.091s \\ 
peterson2C(4)    & 1 & 2 & 1.241s & 0.037s & TO & 9.162s \\ 
peterson2C(5)    & 1 & 2 & 2.801s & 1m47s & TO & TO \\
peterson2C(6)    & 1 & 2 & 6.528s & TO & TO & TO \\ 
peterson2C(7)    & 1 & 2 & 11.030s & TO & TO & TO \\ \hline
\end{tabular}
\caption{Comparison of performance on completely fenced peterson mutual exclusion benchmarks with a bug introduced in the critical section of a single process}
\end{table}
\vspace{-0.5cm}

\begin{table}[h]
\small
\begin{tabular}{ccccccc}
\hline
\textbf{benchmark} & $L$ & $K$ & \tool & \textbf{CDSChecker} & \textbf{GenMC} & \textbf{RCMC} \\ \hline\hline
peterson(3)    & 1 & 2 & 0.878s & TO & 9.665s & 26.208s \\
peterson(2)    & 1 & 2 & 0.321s & 0.325s & 0.087s & 0.068s \\ \hline
peterson(3)    & 2 & 4 & 1.695s & TO & MLE & TO \\
peterson(2)    & 2 & 4 & 0.539s & 15m22s & 0.039s & 0.428s \\ \hline
peterson(3)    & 4 & 4 & 15.900s & TO & MLE & TO \\
peterson(2)    & 4 & 4 & 3.412s & TO & TO & TO \\ \hline
\end{tabular}
\caption{Evaluation using safe mutual exclusion protocols}
\vspace{-0.5cm}
\end{table}
\vfill

\end{document}